%% file: main.tex
\documentclass[a4paper,UKenglish,cleveref, autoref]{lipics-v2019}

\newif\ifappendixtoc\appendixtoctrue
\newif\ifappendix\appendixtrue

\title{On the Balancedness of Tree-to-word Transducers} 

\titlerunning{Balancedness}

\author{Raphaela L\"obel}{TU M\"unchen, Germany}{raphaela.loebel@tum.de}{}{}
\author{Michael Luttenberger}{TU M\"unchen, Germany}{luttenbe@in.tum.de}{
https://orcid.org/0000-0002-4677-9561}{}
\author{Helmut Seidl}{TU M\"unchen, Germany}{seidl@in.tum.de}{}{}

\authorrunning{R.\ L\"obel, M.\ Luttenberger and H.\ Seidl}

\Copyright{Raphaela L\"obel, Michael Luttenberger and Helmut Seidl}



\category{}

\supplement{}

\nolinenumbers 

\hideLIPIcs  

\EventEditors{}
\EventNoEds{0}
\EventLongTitle{}
\EventShortTitle{}
\EventAcronym{}
\EventYear{}
\EventDate{}
\EventLocation{}
\EventLogo{}
\SeriesVolume{}
\ArticleNo{}

\theoremstyle{plain}

\newtheorem{fct}{Fact}

\usepackage{macros}

\begin{document}
\ifappendixtoc
\mtcprepare 
\faketableofcontents 
\fi

\maketitle

\begin{abstract}
A language over an alphabet $\Br=\al\cup\ial$ of opening ($\al$) and closing ($\ial$) brackets, is balanced if 
it is a subset of the Dyck language $\DBr$ over $\Br$, and it is well-formed if all words
are prefixes of words in $\DBr$.
We show that well-formedness of a context-free language is decidable in polynomial time, and that
the longest common reduced suffix can be computed in polynomial time.
With this at a hand we decide for the class \oneLTW\ of non-linear tree transducers
with output alphabet $\Br^*$ whether or not the output language is balanced.
\end{abstract}

\newpage
\input{related}

\input{sec-prelims}

\input{balancedeness2}
\input{sec-wf}

\input{conclusion}

\bibliography{lit}

\appendix

\ifappendix
\newpage
\input{proofs}

\fi

\end{document}

%% file: related.tex
\section{Introduction}
Structured text requires that pairs of opening and closing brackets are properly nested.
This applies to text representing program code as well as
to XML or HTML documents.
Subsequently, we call properly nested words over an alphabet $\Br$ of opening and closing brackets 
\emph{balanced}.
Balanced words, i.e.
structured text, need not necessarily be constructed in a structured way.
Therefore, it is a non-trivial problem whether the set of words produced by some
kind of text processor, consists of balanced words only.
For the case of a single pair of brackets and context-free languages,
decidability of this problem has been settled by
Knuth \cite{DBLP:journals/iandc/Knuth67} where a polynomial time algorithm
is presented by Minamide and Tozawa \cite{DBLP:conf/aplas/MinamideT06}.
Recently, these results were generalized to the output languages of monadic second-order logic (MSO) definable 
tree-to-word transductions \cite{DBLP:journals/ipl/ManethS18}.
The case when the alphabet $\Br$ consists of \emph{multiple} pairs of brackets, though, 
seems to be more intricate.
Still, balancedness for context-free languages was shown to be decidable 
by Berstel and Boasson \cite{DBLP:journals/acta/BerstelB02} where 
a polynomial time algorithm again has been provided by Tozawa and Minamide 
\cite{DBLP:conf/fossacs/TozawaM07}.
Whether or not these results for $\Br$ can be generalized to MSO definable transductions
as e.g.\ done by finite copying macro tree transducers with regular look-ahead,
remains as an open problem.
Reynier and Talbot \cite{ReynierT2016} considered visibly pushdown transducers and showed decidability of
this class with well-nested output in polynomial time.

Here, we provide a first step to answering this question. We consider deterministic
tree-to-word transducers
which process their input at most twice by calling in their axioms at most 
two \emph{linear} transductions of the input. 
Let \oneLTW\ denote the class of these transductions.
Note that the output languages of \emph{linear} deterministic tree-to-word transducers is
context-free, which does not need to be the case for \oneLTW\ transducers.
\oneLTW\ forms a subclass of MSO definable transductions which allows to 
specify transductions such as 
\emph{prepending} an XML document with the list of its section headings,
or \emph{appending} such a document with the list of figure titles.
For \oneLTW\ transducers we show that balancedness is decidable --- and this in
polynomial time.
In order to obtain this result, we first generalize the notion of balancedness to the notion
of \emph{well-formedness} of a language, which means that each word is a \emph{prefix} of a 
balanced word.
Then we show that well-formedness for context-free languages is decidable in polynomial time.
A central ingredient is the computation of the {\em longest common suffix} of a context-free language $L$ over $\Br$ 
{\em after reduction} i.e.\ after canceling all pairs of matching brackets.
While the proof shares many ideas with the computation of the longest common prefix of a context-free language
\cite{DBLP:conf/stacs/LuttenbergerPS18} 
we could not directly make use of the results of~\cite{DBLP:conf/stacs/LuttenbergerPS18} s.t.\ the results of this paper fully subsume the results of~\cite{DBLP:conf/stacs/LuttenbergerPS18}.
Now assume that we have verified that the output language of the 
first linear transduction called in the axiom of the \oneLTW\ transducer
and the \emph{inverted} output language of the second linear transformation
both are well-formed.
Then balancedness of the \oneLTW\ transducer in question, effectively reduces 
to the \emph{equivalence} of two deterministic linear tree-to-word transducers
--- modulo the reduction of opening followed by corresponding closing brackets.
Due to the well-formedness 
we can use the equivalence of
linear tree-to-word transducers over the \emph{free group} which can be decided in polynomial time \cite{Loebel2020}.

This paper is organized as follows.
After introducing basic concepts in Section \ref{s:prelim}, 
Section \ref{s:balanced_2TW} shows how balancedness for 
\oneLTW\ transducers can be reduced to equivalence over the free group and
well-formedness of \LTW s.
Section \ref{sec:CFGwf} considers the problem of deciding well-formedness of 
 context-free languages in general.

\ifappendix\else Missing proofs can be found in the extended version of this paper \cite{Loebel2019}.\fi

\noindent \textbf{Acknowledgements.}
We also like to thank the anonymous reviewers for their 
detailed comments and valuable advice.

%% file: sec-prelims.tex
\section{Preliminaries}\label{s:prelim}
\renewcommand{\al}{\Sigma}

As usual, $\N$ ($\N_0$) denotes the natural numbers (including $0$).
The power set of a set $S$ is denoted by $2^S$.
$\al$ denotes some generic (nonempty) alphabet, 
$\als$ and $\al^\omega$ denote 
the set of all finite words and the set of all infinite words, respectively. 
Then $\al^\infty=\als\cup\al^\omega$ is the set of all countable words.
Note, that the transducers considered here output finite words only;
however, for the operations needed to analyze the output infinite words are very helpful.
We denote the empty word by $\ew$. 
For a finite word $w=w_0\ldots w_l$, its reverse $w^R$ is defined by $w^R=w_{l}\ldots w_1 w_0$; 
as usual, set $L^R:=\{w^R\mid w\in L\}$ for $L\subseteq \als$.
$\sA$ is used to denote an alphabet of \emph{opening brackets} with $\inv{\sA}=\{\inv{a}\mid a\in\sA\}$ the derived alphabet of \emph{closing brackets}, and $\sB:=\sA\cup\inv{\sA}$ the resulting alphabet of \emph{opening and closing brackets}. 

\paragraph*{Longest common prefix and suffix}
Let $\al$ be an alphabet.
We first define the \emph{longest common prefix} of a language,
and then reduce the definition of the \emph{longest common suffix} to it by means of the reverse.
We write $\ple$ to denote the prefix relation on $\al^\infty$, i.e.\ we have $u\ple w$ if either 
(i) $u,w\in\als$ and there exists $v\in\als$ s.t.\ $w=uv$, or
(ii) $u\in\als$ and $w\in\al^\omega$ and there exists $v\in\al^\omega$ s.t.\ $w=uv$, or
(iii) $u,w\in\al^\omega$ and $u=w$.
We extend $\al^\infty$ by a greatest element $\top\not\in\al^\infty$ w.r.t.\ $\ple$
s.t.\ $u\ple \top$ for all $u\in \al_{\top}^\infty := \al^{\infty}\cup\{\top\}$.
Then every set $L\subseteq \al^\infty_{\top}$ has an infimum w.r.t.\ $\ple$
which is called the \emph{longest common prefix} of $L$, abbreviated by $\lcp(L)$.
Further, define $\ew^\omega:=\top$, $\top^R:=\top$, and $\top w := \top =: w\top$ for all $w\in \al_{\top}^\infty$.

In Section~\ref{sec:CFGwf} we will need to study the 
\emph{longest common suffix ($\lcs$)} of a language $L$.
For $L\subseteq \als$, we can simply set $\lcs(L):=\lcp(L^R)^R$,
but also certain infinite words are very useful for describing how the $\lcs$ changes when concatenating two languages (see e.g.\ Example~\ref{ex:lcsext}).
Recall that for $u,w\in\als$ and $w\neq \ew$ the $\omega$-regular expression $u w^\omega$ denotes the unique infinite word $uwww\ldots$ in $\bigcap_{k\in\N_0} uw^k\al^\omega$; such a word is also called \emph{ultimately periodic}.
For the $\lcs$ we will use the expression $w^\iomega u$ to denote 
the \emph{ultimately left-periodic} word $\ldots www u$ that ends on the suffix $u$ with infinitely many copies of $w$ left of $u$; these words are used to abbreviate the fact that we can generate a word $w^k u$ for unbounded $k\in\N_0$. 
As we reduce the $\lcs$ to the $\lcp$ by means of the reverse,
we define the reverse of $w^\iomega u$, denoted by $(w^\iomega u)^R$, by means of $(w^\iomega u)^R:=u^R (w^R)^\omega$.
\begin{definition}
\label{def:lcs}
Let $\alup$ denote the set of all expressions of the form $w^\iomega u$ with $u\in\als$ and $w\in\al^+$.
$\alup$ is called the set of \emph{ultimately left-periodic words}.
Define the reverse of an expression $w^\iomega u\in \alup$ by means of $(w^\iomega u)^R:=u^R (w^R)^\omega$.
Accordingly, set $(u w^\omega)^R := (w^R)^\iomega u^R$ for $u\in\als$, $w\in\al^+$.

The \emph{suffix order} on $\als\cup \alup\cup\{\top\}$ is defined by $u\sle v:\Leftrightarrow u^R \ple v^R$.
The \emph{longest common suffix ($\lcs$)} of a language $L\subseteq \als\cup\alup$ is $\lcs(L):=\lcp(L^R)^R$.
\end{definition}
\noindent
For instance, we have $\lcs((bba)^\iomega, (ba)^\iomega a)=
a$,
and $\lcs((ab)^\iomega, (ba)^\iomega b) =
(ab)^\iomega$.

As usual, we write $u\slt v$ if $u\sle v$, but $u\neq v$.
As the $\lcp$ is the infimum w.r.t.\ $\ple$, we also have for $x,y,z\in\{\top\}\cup\als\cup \alup$ and $L,L'\subseteq \{\top\}\cup\als\cup \alup$ that (i) $\lcs(x,y)=\lcs(y,x)$, (ii) $\lcs(x,\lcs(y,z))=\lcs(x,y,z)$, (iii) $\lcs(L)\sle \lcs(L')$ for $L\supseteq L'$, and (iv) $\lcs(Lx)=\lcs(L)x$ for $x\in \{\top\}\cup\als$.
In \refA{lem:lcs-calculus}{Lemma~8} 
we derive further equalities for $\lcs$ that allow to simplify its computation.
In particular, the following two equalities (for $x,y\in\als$) are very useful:
\[
\begin{array}{lcl}
\lcs(x,xy) & = & \lcs(x,y^\iomega) = \lcs(x,xy^k) \quad \text{ for every } k\ge 1\\[2mm]
\lcs(x^\iomega,y^\iomega) & = & \begin{cases}
(xy)^\iomega & \text{ if } xy=yx\\
\lcs(xy,x^\iomega)=\lcs(xy,yx^k) & \text{ if } xy\neq yx, \text{ for every } k\ge 1\\
\end{cases}\\[2mm]
\end{array}
\]
For instance, we have $\lcs((ab)^\iomega, (bab)^\iomega) = bab = \lcs(abbab, (ab)^\iomega )$.
Note also that by definition we have $\ew^\iomega =\top$ s.t.\ $\lcs(x^\iomega, \ew^\iomega)=
(x\ew)^\iomega$.
We will use the following observation frequently:
\begin{lemma}\label{lem:lcs-witness-main-text}
Let $L\subseteq\al^\ast$ be nonempty. 
Then for any $x\in L$ we have $\lcs(L)=\lcs(\lcs(x,z)\mid z\in L)$;
in particular, there is some {\em witness} $y\in L$ (w.r.t.\ $x$)
s.t.\ $\lcs(L)=\lcs(x,y)$.
\end{lemma}

\renewcommand{\al}{\sA}

\paragraph*{Involutive monoid} We briefly recall the basic definitions and properties of the finitely generated involutive monoid, but refer the reader for details and a formal treatment to e.g.\ \cite{elements}.
Let $\al$ be a finite alphabet (of opening brackets/letters). From $\al$ we derive the alphabet $\ial:=\{\inv{a}\mid a\in\al\}$ (of closing brackets/letters) where we assume that $\al\cap \ial =\emptyset$. Set $\Br:=\al\cup\ial$.
We use roman letters $p,q,\ldots$ to denote words over $\al$, while Greek letters $\alpha,\beta,\gamma,\ldots$ will denote words over $\Br$.

We extend $\inv{\cdot}$ to an involution on $\Br^\ast$ by means of $\overline{\ew}:=\ew$, $\overline{\overline{a}}:= a$ for all $a\in \al$, and $\overline{\alpha\beta}:= \overline{\beta}\,\overline{\alpha}$ for all other $\alpha,\beta\in \Br^\ast$.
Let $\rdto$ be the binary relation on $\Br^\ast$ defined by $\alpha a \inv{a} \beta \rdto \alpha \beta$ for any $\alpha,\beta\in\Br^\ast$ and $a\in\al$, i.e.\ $\rdto$ cancels nondeterministically one pair of matching opening and closing brackets.
A word $\alpha\in\Br^\ast$ is \emph{reduced} if it does not contain any infix of the form $a\inv{a}$ for any $a\in\al$,
i.e.\ $\alpha$ is reduced if and only if it has no direct successor w.r.t.\ $\rdto$.
For every $\alpha\in\Br^\ast$ canceling all matching brackets in any arbitrary order always results in the same unique reduced word which we denote by $\rd(\alpha)$; we write $\alpha\rdeq \beta$ if $\rd(\alpha)=\rd(\beta)$.
Then $\Br^\ast/\rdeq$ is the free involutive monoid generated by $\al$, and $\rd(\alpha)$ is the shortest word in the $\rdeq$-equivalence class of $\alpha$.
For $L\subseteq \Br^\ast$ we set $\rd(L) := \{\rd(w)\mid w\in L\}$. 

\paragraph*{Well-formed languages and context-free grammars}
We are specifically interested in context-free grammars (CFG) $G$ over the alphabet $\Br$. 
We write $\to_G$ for the rewrite rules of $G$.
We assume that $G$ is reduced to the productive nonterminals that are reachable from its axiom $S$.
For simplicity, we assume for the proofs and constructions that the rules of $G$ are of the form 
\[
X\to_G YZ \qquad X\to_G Y \qquad X\to_G \inv{u}v
\]
for nonterminals $X,Y,Z$ and $u,v\in\als$. 
We write $L_X:=\{\alpha\in\Br^\ast \mid X\to_G^\ast \alpha\}$ for the language generated by the nonterminal $X$. 
Specifically for the axiom $S$ of $G$ we set $L:=L_S$. 
The height of a derivation tree w.r.t.\ $G$ is measured in the maximal number of nonterminals occurring along a path from the root to any leaf, i.e.\ in our case any derivation tree has height at least $1$. 
We write $L_X^{\le h}$ for the subset of $L_X$ of words that possess a derivation tree of height at most $h$ s.t.:
\[
L_X^{\le 1} = \{ \inv{u}v \mid X\to_G \inv{u}v\} \quad L_X^{\le h+2} = L_X^{\le h+1} \cup \bigcup_{X\to_G YZ} L_Y^{\le h+1} L_Z^{\le h+1} \cup \bigcup_{X\to_G Y} L_Y^{\le h+1} 
\]
We will also write $L_X^{< h}$ for $L_X^{\le h-1}$ and $L_X^{=h}$ for $L_X^{\le h}\setminus L_X^{<h}$.
The {\em prefix closure} of $L\subseteq\Br^\ast$ is denoted by $\pfcl(L):=\{ \alpha' \mid \alpha'\alpha''\in L\}$.

\begin{definition}\label{def:wf-main-text}
Let $\alpha\in\Br^\ast$ and $L\subseteq \Br^\ast$.
\begin{enumerate}
\item
Let $\hd(\alpha):=\abs{\alpha}_{\al}-\abs{\alpha}_{\ial}$ be the difference of opening brackets to closing brackets.
$\alpha$ is {\em nonnegative} if $\forall \alpha'\ple \alpha\colon \hd(\alpha')\ge 0$.
$L\subseteq \Br^\ast$ is {\em nonnegative} if every $\alpha\in L$ is nonnegative.
\item
A context-free grammar $G$ with $L(G)\subseteq \Br^\ast$ is {\em nonnegative} if $L(G)$ is nonnegative. For a nonterminal $X$ of $G$ let
$d_X:=\sup (\{-\hd(\alpha') \mid \alpha'\alpha'' \in L_X\})$.
\item
A word $\alpha$ is {\em \twwf\ (short: $\wwf$) resp.\ \twf\ (short: $\wf$)} if $\rd(\alpha)\in \ial^\ast\al^\ast$ resp.\ if $\rd(\alpha)\in \al^\ast$.
A context-free grammar $G$ is $\wf$ if $L(G)$ is $\wf$.
$L\subseteq \Br^\ast$ is {\em \wwf\ resp.\ \wf} if every word of $L$ is \wwf\ resp.\ \wf.
\item
A context-free grammar $G$ is {\em bounded well-formed (\bwf)} if it is \wwf\ and for every nonterminal $X$ there is a (shortest) word $r_X\in\al^\ast$ with $\abs{r_X}=d_X$ s.t.\ $r_X L_X$ is $\wf$.

\end{enumerate}
\end{definition}

\noindent
Note that $d_X\ge 0$ as we can always choose $\alpha'=\ew$ in the definition of $d_X$.

As already mentioned in the abstract and the introduction, we have that $L$ is $\wf$ {\em iff} $\pfcl(L)$ is $\wf$ {\em iff} $L$ is a subset of the prefix closure of the Dyck language generated by $S\to \ew$, $S\to SS$, $S\to aS\inv{a}$  (for $a\in\al$).
We state some further direct consequences of above definition:
(i) $L$ is nonnegative {\em iff} the image of $L$ under the homomorphism that collapses $\al$ to a singleton is $\wf.$
Hence, if $L$ is $\wf$, then $L$ is nonnegative. $\hd$ is an $\omega$-continuous homomorphism from the language semiring generated by $\Br$ to the tropical semiring $\cg{\Z\cup\{-\infty\},\min,+}$. Thus it is decidable in polynomial time if $G$ is nonnegative using the Bellman-Ford algorithm~\cite{DBLP:journals/tcs/EsparzaKL11}.
(ii) If $L$ is not $\wf$, then there exists some $\alpha\in \pfcl(L)\setminus\{\ew\}$ s.t.\ $\hd(\alpha)<0$ or $\alpha\rdeq ua\inv{b}$ for $u\in\als$ and $a,b\in\al$ (with $a\neq b$).
(iii) If $L_X$ is $\wwf$, then $d_X = \sup\{\abs{y} \mid \gamma\in L_X, \rd(\gamma)=\inv{y}z\}$.

\noindent
In particular, because of context-freeness, it follows that, if $G$ is $\wf$, then for every nonterminal $X$ there is $r_X\in\als$ s.t.\ (i) $\inv{r_X}\in \rd(\pfcl(L_X))$, (ii) $\abs{r_X}=d_X$ and (iii)$r_X L_X$ is $\wf$. Hence:
\begin{lemma}\label{lem:char-wf-main-text}
A context-free grammar $G$ is $\wf$ iff $G$ is $\bwf$ with $r_S=\ew$ for $S$ the axiom of $G$.
\end{lemma}
The words $r_X$ mentioned in the definition of bounded well-formedness can be computed in polynomial time using the Bellman-Ford algorithm similar to \cite{DBLP:conf/fossacs/TozawaM07}; more precisely, a {\em straight-line program ({{\SLP}})} (see e.g.~\cite{Lohrey2012} for more details on {\SLP}s), i.e.\ a context-free grammar generating exactly one derivation tree and thus word, can be extracted from $G$ for each $r_X$.

\begin{lemma}\label{lem:r_X-main-text}
Let $L=L(G)$ be $\wf$. Let $X$ be some nonterminal of $G$. 
Let $r_X\in \al^\ast$ be the shortest word s.t.\ $r_X L_X$ is \wf.
We can compute an {{\SLP}} for $r_X$ from $G$ in polynomial time.
\end{lemma}

\paragraph*{Tree-to-word transducers}
\input{transducers}

%% file: transducers.tex
\renewcommand{\LTW}{${\sf LT}_\Br$}
We define a \emph{linear} tree-to-word transducer (\LTW) $M = (\Sigma, \Br, Q, S, R)$ where
$\Sigma$ is a finite ranked input alphabet, $\Br$ is the finite (unranked) output alphabet, $Q$ is a finite set of states,
the axiom $S$ is of the form $u_0$ or $u_0 q(x_1) u_1$ with $u_0, u_1 \in \Br^*$ and $R$ is a set of rules of the form
$q(f(x_1,\ldots, x_m)) \to u_0 q_1(x_{\sigma(1)}) u_1 \ldots q_n(s_{\sigma(n)}) u_n$ with
$q, q_i \in Q$, $f \in \Sigma$, $u_i \in \Br^*$, $n \leq m$ and $\sigma$ an injective mapping from $\{1,\ldots, n\}$ to $\{1,\ldots, m\}$.
Since non-deterministic choices of linear transducers can be encoded into the input symbols,
we may, w.l.o.g., consider \emph{deterministic} transducers only.
For simplicity, we moreover assume the transducers to be \emph{total}. 
This restriction can be lifted by additionally taking a top-down deterministic tree 
automaton for the domain into account. 
The constructions introduced in Section~\ref{s:balanced_2TW} would
then have to be applied w.r.t.\ such a domain tree automaton.
As we consider total deterministic transducers there is exactly one rule for each pair $q \in Q$ and $f \in \Sigma$.

A $2$-copy tree-to-word transducer (\oneLTW) is a tuple $N = (\Sigma, \Br, Q, S, R)$ that is defined in the
same way as an \LTW\ but the axiom $S$ is of the form $u_0$ or $u_0 q_1(x_1) u_1 q_2(x_1) u_2$, with $u_i \in \Br^*$.

$\T_\Sigma$ denotes the set of all trees/terms over $\Sigma$.
We define the semantics $\sem{q}: \T_\Sigma \to  \Br^*$ of a state $q$ with rule
$q(f(t_1,\ldots, t_m)) \to u_0 q_1(t_{\sigma(1)}) u_1 \ldots q_n(t_{\sigma(n)}) u_n$
inductively by 
$$\sem{q}(f(t_1,\ldots, t_m)) = \rd(u_0 \sem{q_1}(t_{\sigma(1)}) u_1 \ldots \sem{q_n}(t_{\sigma(n)}) u_n)$$
The semantics $\sem{M}$ of an \LTW\ $M$ with axiom $u_0$ is given by $\rd(u_0)$;
if the axiom is of the form $u_0 q(x_1) u_1$ it is defined by $\rd(u_0 \sem{q}(t) u_1)$ for all $t \in \T_\Sigma$;
while the semantics $\sem{N}$ of a \oneLTW\ $N$ with axiom $u_0$ is again given by $\rd(u_0)$ and
for axiom $u_0 q_1(x_1) u_1 q_2(x_1) u_2$ it is defined by
$\rd(u_0 \sem{q_1}(t) u_1 \sem{q_2}(t) u_2)$ for all $t \in \T_\Sigma$.
For a state $q$ we define the output language $\mathcal{L}(q) = \{\sem{q}(t) \mid t \in \T_\Sigma\}$;
For a \oneLTW\ $M$ we let $\mathcal{L}(M) = \{\sem{M}(t) \mid t \in \T_\Sigma\}$.
Note that the output language of an \LTW\ is context-free and a corresponding context-free grammar
for this language can directly read from the rules of the transducer.

Additionally, we may assume w.l.o.g.\ that all states $q$ of an \LTW\ are \emph{nonsingleton}, i.e.,
$\L(q)$ contains at least two words.
We call a \oneLTW\ $M$ \emph{balanced} if $\mathcal{L}(M) = \{\ew\}$.
We say an \LTW\ $M$ is \emph{well-formed} if $\mathcal{L}(M) \subseteq \al^*$.
Balanced and well-formed states are defined analogously.
We use $\inv{q}$ to denote the inverse transduction of $q$ which is obtained from a copy of the
transitions reachable from $q$ by involution of the right-hand side of each rule.
As a consequence, $\sem{\inv{q}}(t) = \inv{\sem{q}(t)}$ for all $t \in \T_\Sigma$, and thus, $\mathcal{L}(\inv{q}) = \inv{\mathcal{L}(q)}$.
We say that two states $q$, $q'$ are \emph{equivalent} iff for all $t \in \T_\Sigma$, $\sem{q}(t) = \sem{q'}(t)$.
Accordingly, two \oneLTW s $M$, $M'$ are equivalent iff for all $t \in \T_\Sigma$, $\sem{M}(t) = \sem{M'}(t)$.

%% file: balancedeness2.tex
\section{Balancedness of \oneLTW s}\label{s:balanced_2TW}

Let $M$ denote a \oneLTW.
W.l.o.g., we assume that the axiom of $M$ is of the form
$q_1(x_1) q_2(x_1)$ for two states $q_1,q_2$.
If this is not yet the case, an equivalent \oneLTW\ with this property can be constructed
in polynomial time.
We reduce balancedness of $M$ to decision problems for 
\emph{linear} tree-to-word transducers alone.
\begin{proposition}\label{p:reduction}
The \oneLTW\ $M$ is balanced iff the following two properties hold:
\begin{itemize}
\item Both $\mathcal{L}(q_1)$ and $\inv{\mathcal{L}(q_2)}$ are well-formed;
\item $q_1$ and $\inv{q_2}$ are equivalent.
\end{itemize}
\end{proposition}
\begin{proof}
Assume first that $M$ with axiom $q_1(x_1) q_2(x_1)$ is balanced, i.e., $\L(M) = \ew$.
Then for all $w', w''$ with $w = w'w'' \in \L(M)$,
$\rd(w') = u \in \al^*$ and $\rd(w'') = \inv{u}$.
Thus, both $\L(q_1)$ and $\inv{\L(q_2)}$ consist of well-formed words only.
Assume for a contradiction that $q_1$ and $\inv{q_2}$ are not equivalent. 
Then there is some $t \in \T_\Sigma$ such that
$\sem{q_1}(t) \not \rdeq \sem{\inv{q_2}}(t)$.
Let $\sem{q_1}(t) = u \in \al^*$ and $\sem{\inv{q_2}}(t) = \inv{\sem{q_2}(t)} = v$
with $v \in \al^*$ and $u \neq v$.
Then $\rd(\sem{q_1}(t)\sem{q_2}(t)) = \rd(u\inv{v}) \neq \ew$ as $u\neq v$, $u,v\in\al^*$.
Since $M$ is balanced, this is not possible.

Now, assume that $\L(q_1)$ and $\inv{\L(q_2)}$ are well-formed, i.e.,
for all $t\in \T_\Sigma$, $\sem{q_1}(t) \in \al^*$ and $\sem{q_2}(t) \in \ial^*$.
Additionally assume that $q_1$ and $\inv{q_2}$ are equivalent, i.e.,
for all $t \in \T_\Sigma$, $\sem{q_1}(t) = \sem{\inv{q_2}}(t) = \inv{\sem{q_2}(t)}$.
Therefore for all $t\in \T_\Sigma$, $\sem{q_2}(t) = \inv{\sem{q_1}(t)}$
and hence, 
$${\small\rd(\sem{q_1}(t)\sem{q_2}(t)) = \rd(\sem{q_1}(t)\inv{\sem{q_1}(t)}) = \ew}$$
Therefore, the \oneLTW\ $M$ must be balanced. \qed
\end{proof}

\noindent
The output languages of states $q_1$ and $\inv{q_2}$ are generated by means of
context-free grammars of polynomial size.
\begin{example}\label{ex:LTWCFG}
Consider \LTW\  $M$ with input alphabet $\Sigma=\{f^{(2)}, g^{(0)}\}$ (the superscript denotes the rank),
output alphabet $\sB=\{a, \inv{a}\}$, axiom $q_3(x_1)$ and rules
\[
\begin{array}{l@{\quad}l}
q_3(f(x_1,x_2)) \to a q_2(x_1) q_2(x_2) \inv{a} & q_2(g) \to \ew \\
q_2(f(x_1,x_2)) \to a q_1(x_1) q_1(x_2) \inv{a} & q_2(g) \to \ew \\
q_1(f(x_1,x_2)) \to q_3(x_1) q_3(x_2) & q_1(g) \to aa\\
\end{array}
\]
We obtain a CFG\ producing exactly the output language of $M$ by nondeterministically guessing the input symbol, i.e.\ the state $q_i$ becomes the nonterminal $W_i$.
The axiom of this CFG is then $W_3$, and 
as rules we obtain
\[
\begin{array}{l@{\quad}l@{\quad}l}
W_3 \to a W_2W_2 \inv{a} \mid \ew & W_2 \to a W_1W_1 \inv{a} \mid \ew & W_1 \to W_3 W_3 \mid aa
\end{array}
\]
Note that the rules of $M$ and the associated CFG use a form of iterated squaring, i.e.\ $W_3\to^2 W_3^4$,
that allows to encode potentially exponentially large outputs within the rules (see also Example~\ref{ex:the-one}).
In general, words thus have to be stored in compressed form as \SLP s~\cite{Lohrey2012}.
\end{example}

Therefore, Theorem \ref{t:xyz} of Section~\ref{sec:CFGwf} implies that 
well-formedness of $q_1$, $\inv{q_2}$ can be decided in polynomial time.
Accordingly, it remains to consider the equivalence problem for \twf\ \LTW s.
Since the two transducers in question are well-formed, they are equivalent as
\LTW s iff they are equivalent when their outputs are considered over the free group
$\F_\al$.
In the free group $\F_\al$, we additionally have that
$\inv{a}a \rdeq \ew$ --- which does not hold in our rewriting system.
If sets $\L(q_1),\L(\inv{q_2})$ of outputs for $q_1$ and $\inv{q_2}$, however, are \twf, 
it follows for all $u\in\L(q_1),v\in\L(\inv{q_2})$ that $\rd(u\inv{v})=\rd(\rd(u)\rd(\inv{v}))$ 
cannot contain $\inv{a}a$. Therefore, $\rd(u\inv{v})=\ew$ iff $u\inv{v}$ is equivalent to $\ew$
over the free group $\F_\al$.
In \cite[Theorem 2]{Loebel2020}, we have proven that equivalence of \LTW s where the output is interpreted
over the free group, is decidable in polynomial time. 
Thus, we obtain our main theorem.
\begin{theorem}\label{t:main}
Balancedness of \oneLTW s is decidable in polynomial time.
\end{theorem}

%% file: sec-wf.tex
\section{Deciding whether a context-free language is well-formed}\label{sec:CFGwf}
\renewcommand{\al}{\Sigma}
As described in the preceding sections, given a \oneLTW\ we split it into the two underlying \LTW s that process a copy of the input tree.
We then check that each of these two \LTW s are equivalent w.r.t.\ the free group.
As sketched in \Cref{ex:LTWCFG} we obtain a context-free grammar
for the output language of each of these \LTW~s.
It then remains to check that both context-free grammars are well-formed.
In order to prove that we can decide in polynomial time whether a context-free grammar is well-formed (short: $\wf$), we proceed as follows:

First, we introduce in Definition~\ref{def:lcsext-main-text} 
the \emph{maximal suffix extension} of a language $L\subseteq \als$ w.r.t.\ the $\lcs$ (denoted by $\lcsext(L)$), 
i.e.\ the longest word $u\in\al^\infty$ s.t.\ $\lcs(uL) = u\lcs(L)$.
We then show that the relation $L\tseq L':\Leftrightarrow \lcs(L)=\lcs(L')\wedge \lcsext(L)=\lcsext(L')$ is an equivalence relation on $\als$ that respects both union and concatenation of languages (see Lemma~\ref{lem:lcs-computation-main-text}).
It then follows that for every language $L\subseteq\als$ there is some subset $\tsn(L)\subseteq L$ of size at most $3$ with $L\tseq \tsn(L)$.

We then use $\tsn$ to compute a finite $\tseq$-equivalent representation $\sT_X^{\le h}$ of the \emph{reduced} language generated by each nonterminal $X$ of the given context-free grammar inductively for increasing derivation height $h$.
In particular, we show that we only have to compute up to derivation height $4N+1$ (with $N$ the number of nonterminals)
in order to decide whether $G$ is $\wf$:
In~Lemma~\ref{thm:converge-wf-main-text} we show that, if $G$ is $\wf$,
then we have to have $\sT_X^{\le 4N+1} \tseq \sT_X^{\le 4N}$ for all nonterminals $X$ of $G$.
The complementary result is then shown in Lemma~\ref{thm:not-wf--main-text},
i.e.\ if $G$ is not $\wf$, then 
we either cannot compute up to $\sT_X^{\le 4N+1}$
as we discover some word that is not $\wf$,
or we have $\sT_X^{\le 4N}\not\tseq \sT_{X}^{\le 4N+1}$ for at least one nonterminal $X$.

\paragraph*{Maximal suffix extension and $\lcs$-equivalence}

We first show that we can compute the longest common suffix of the union $L\cup L'$ and the concatenation $LL'$ of two languages $L,L'\subseteq \als$
if we know both $\lcs(L)$ and $\lcs(L')$, and in addition, the longest word $\lcsext(L)$ resp.\ $\lcsext(L')$ by which we can extend $\lcs(L)$ resp.\ $\lcs(L')$ when concatenating another language from left.
In contrast to the computation of the $\lcp$ presented in~\cite{DBLP:conf/stacs/LuttenbergerPS18},
we have to take the maximal extension $\lcsext$ explicitly into account.
In this paragraph we do not consider the involution, thus let $\al$ denote an arbitrary alphabet.
\begin{definition}\label{def:lcsext-main-text}
For $L\subseteq\al^\ast$ with $R=\lcs(L)$ the \emph{maximal suffix extension ($\lcsext$)} of $L$ is defined by $\lcsext(L):= \lcs( z^{\iomega} \mid zR \in L)$. 
\end{definition}

\noindent
Recall that by definition $\lcsext(\emptyset)=\lcs(\emptyset)=\top$ and $\lcsext(\{R\})=\lcs(\ew^\iomega)=\top$.
The following example motivates the definition of $\lcsext$:
\begin{example}\label{ex:lcsext}
Consider the language $L=\{R,xR,yR\}$ with $\lcs(L)=R$ and $\lcsext(L)=\lcs(x^\iomega,y^\iomega)$. 
Assume we prepend some word $u\in\als$ to $L$ resulting in the language $uL = \{uR,uxR,uyR\}$, see the following picture for an illustration (dotted boxes represent copies of $z\in\{x,y\}$ stemming from the usual line of argumentation that, if $z$ is a suffix of $u=u'z$, then $uzR=u'zzR$, and thus eventually covering all of $u$ by $z^\iomega$):
\begin{center}
\scalebox{0.7}{
\begin{tikzpicture}
\begin{scope}
\draw (0,0) rectangle (2,0.5);
\node at (1,0.25) {$R$};
\draw (-5,0) rectangle (0,0.5);
\node at (-2.5,0.25) {$u$};
\end{scope}
\begin{scope}[yshift=0.5cm]
\draw (0,0) rectangle (2,0.5);
\node at (1,0.25) {$R$};
\draw (-7,0) rectangle (-2,0.5);
\node at (-4.5,0.25) {$u$};
\draw (-2,0) rectangle (0,0.5);
\node at (-1,0.25) {$x$};
\end{scope}
\begin{scope}[yshift=1cm]
\draw[dotted] (0,0) rectangle (2,0.5);
\node at (1,0.25) {$R$};
\draw[dotted] (-2,0) rectangle (0,0.5);
\node at (-1,0.25) {$x$};
\node at (-3,0.25) {$x$};
\node at (-5,0.25) {$x$};
\node at (-7,0.25) {$x$};
\draw[dotted] (-4,0) rectangle (-2,0.5);
\draw[dotted] (-6,0) rectangle (-4,0.5);
\draw[dotted] (-8,0) rectangle (-6,0.5);
\end{scope}
\begin{scope}[yshift=-0.5cm]
\draw (0,0) rectangle (2,0.5);
\node at (1,0.25) {$R$};
\draw (-8,0) rectangle (-3,0.5);
\node at (-5.5,0.25) {$u$};
\draw (-3,0) rectangle (0,0.5);
\node at (-1.5,0.25) {$y$};
\end{scope}
\begin{scope}[yshift=-1cm]
\draw[dotted] (0,0) rectangle (2,0.5);
\node at (1,0.25) {$R$};
\draw[dotted] (-3,0) rectangle (0,0.5);
\node at (-1.5,0.25) {$y$};
\node at (-4.5,0.25) {$y$};
\node at (-7.5,0.25) {$y$};
\draw[dotted] (-3,0) rectangle (-6,0.5);
\draw[dotted] (-6,0) rectangle (-9,0.5);
\end{scope}
\end{tikzpicture}
}
\end{center}
As motivated by the picture, $\lcs(uL)$ is given by $\lcs(u,x^\iomega,y^\iomega)R$.
Using the concept of ultimately left-periodic words, we may also formalize this as follows:
\[
\begin{array}{lcl@{\quad}l}
\lcs(u\{xR,yR,R\})
& = &\lcs(u,ux,uy)R\\
& = & \lcs(\lcs(u,ux),\lcs(u,uy)) R & (\text{as }\lcs(u,ux)=\lcs(u,x^\iomega))\\
& = &  \lcs(\lcs(u,x^\iomega),\lcs(u,y^\iomega))R\\
& = & \lcs(u,\lcs(x^\iomega,y^\iomega))R
& =  \lcs(u, \lcsext(L)) \lcs(L)
\end{array}
\]
In particular, if $xy=yx$, we can extend $\lcs$ by any finite suffix of $\lcsext(L)=(xy)^\iomega$ (note that, if $x=\ew=y$, then $\lcsext(L)=\ew^\iomega=\top$ is defined to be the greatest element w.r.t.\ $\sle$); 
but if $xy\neq yx$, we can extend it at most to $\lcsext(L)=\lcs(x^\iomega,y^\iomega)=\lcs(xy,yx)\slt xy$.
Essentially, only three cases can arise as illustrated by the following three examples:

First, consider $L_1 = \{ab, cb\}$ with $\lcs(L_1) = b$.
Obviously, for every word $u\in\als$ we have that $\lcs(uL_1) = \lcs(L_1)$ and so we should have $\lcsext(L_1)=\ew$.
Instantiating the definition we obtain indeed $\lcsext(L_1) = \lcs(a^\iomega,b^\iomega) = \lcs(\ew) = \ew$.

As another example consider $L_2 = \{a, baa\}$ with $\lcs(L_2)=a$.
Here, we obtain 
$\lcsext(L_2)=\lcs(\ew^\iomega, (ba)^\iomega)=\lcs(\top,(ba)^\iomega)= (ba)^\iomega$,
i.e.\ the suffix of $L_2$ can be extended by any finite suffix of $(ba)^\iomega=\ldots bababa$.

Finally, consider $L_3=\{b,ba^nb,aba^nb\}$ with $\lcs(L_3)=b$ for some fixed $n\in\N$.
As mentioned in Section~\ref{s:prelim}, we have $\lcs(x^\iomega,y^\iomega)=\lcs(xy,yx)$ for $xy\neq yx$.
We thus obtain in this case 
$\lcs((ba^n)^\iomega,(aba^n)^\iomega)=\lcs(ba^n\,aba^n,aba^n\,ba^n)= a^nba^n$.
The classic result by Fine and Wilf states that, if $xy\neq yx$, then $\abs{\lcs(x^\iomega,y^\iomega)}< \abs{x}+\abs{y}-\gcd(\abs{x},\abs{y})$.
Thus $x=ba^n$ and $y=aba^n$ constitute an extremal case where the $\lcs$ is only finitely extendable.
\end{example}

If $\lcs(L)$ is not contained in $L$, then $\lcs(L)$ has to be a strict suffix of every shortest word in $L$, and thus immediately $\lcsext(L)=\ew$.
As in the case of the $\lcs$, also $\lcsext(L)$ is already defined by two words in $L$:
\begin{lemma}\label{lem:lcsextwit-main-text}
Let $L\subseteq\al^\ast$ with $\abs{L}\ge 2$ and $R:=\lcs(L)$.
Fix any $xR\in L\setminus\{R\}$.
Then there is some $yR\in L\setminus\{R\}$ s.t.\
$\lcsext(L)=\lcs(x^\iomega,y^\iomega) =\lcs(x^\iomega,y^\iomega,z^\iomega)
$ for all $zR\in L$.
If $xy=yx$, then $R\in L$.
\end{lemma}
\noindent

\noindent
We show that 
we can compute the $\lcs$ and the extension $\lcsext$ of the union resp.\ the concatenation of two languages
solely from their $\lcs$ and $\lcsext$.
To this end, we define the $\lcs$-\emph{summary} of a language as:
\begin{definition}\label{def:lcs-cong-main-text}
For $L\subseteq \al^\ast$ set $\sig(L):= (\lcs(L),\lcsext(L))$.
The equivalence relation $\tseq$ on $2^{\al^\ast}$ is defined by:
$L\tseq L' \text{ iff }\sig(L)=\sig(L')$.
\end{definition}

\begin{lemma}\label{lem:lcs-computation-main-text}
Let $L,L'\subseteq \al^\ast$ with $\sig(L)=(R,E)$ and $\sig(L')=(R',E')$. 
If $L=\emptyset$ or $L'=\emptyset$, then $\sig(L\cup L')=(\lcs(R,R'),\lcs(E,E'))$, and $\sig(LL')=(\top,\top)$. Assume thus $L\neq \emptyset \neq L'$ which implies $R\neq \top \neq R'$. Then:
\begin{itemize}
\item 
$\lcs(L\cup L') = \lcs(R,R')$ and $\lcs(LL') = \lcs(R,E')R'$.
\item 
If $\lcs(R,R') \not\in \{R,R'\}$, then $\lcsext(L\cup L')=\ew$; 
else w.l.o.g.\ $R'=\delta R$ and $\lcsext(L\cup L') = \lcs(E,\lcs(E',E'\delta)\delta)$. 
\item
If $\lcs(R,E')\slt R$, then $\lcsext(LL') = \ew$;
else $E'=\delta R$ and $\lcsext(LL') = \lcs(E,\delta)$.
\end{itemize}
\end{lemma}

\begin{example}\label{ex:lcs-ops}
Lemma~\ref{lem:lcs-computation-main-text} can be illustrated as follows: 
\begin{center}
\scalebox{0.7}{
\begin{tikzpicture}[yscale=0.5]
\node at (-14,-1) {$(L\cup L')$};
\draw[white] (-15,-2) rectangle (-13,0);
\draw (0,0) rectangle (-3,1);
\draw (0,0) rectangle (-5,-1);
\draw (0,-1) rectangle (-5,-2);
\node at (-1,0.5) {$\lcs(L)$};
\node at (-1,-0.5) {$\lcs(L')$};
\node at (-1,-1.5) {$\lcs(L')$};
\draw[dotted] (-3,0) -- (-3,-1);
\node at (-4,-0.5) {$\delta$};
\draw[dotted] (-8,0) rectangle (-3,1);
\node at (-6,0.5) {$\lcsext(L)$};
\draw[dotted] (-11,-1) rectangle (-5,-2);
\node at (-6,-1.5) {$\lcsext(L')$};
\draw[dotted] (-7,-1) rectangle (-5,0);
\node at (-6,-0.5) {$\delta$};
\draw[dotted] (-9,-1) rectangle (-7,0);
\node at (-8,-0.5) {$\delta$};
\draw[dotted] (-11,-1) rectangle (-7,0);
\node at (-10,-0.5) {$\delta$};
\end{tikzpicture}
}
\end{center}
\begin{center}
\scalebox{0.7}{
\begin{tikzpicture}[yscale=0.5]
\node at (-12,-1) {$(LL')$};
\draw[white] (-13,-2) rectangle (-11,0);
\draw (-3,1) rectangle (0,0);
\draw[dotted] (-8,-1) rectangle (0,0);
\draw (0,0) rectangle (2,1);
\draw (0,0) rectangle (2,-1);
\node at (1,-0.5) {$\lcs(L')$};
\node at (1,0.5) {$\lcs(L')$};
\node at (-1,0.5) {$\lcs(L)$};
\node at (-1,-0.5) {$\lcs(L)$};
\node at (-5.5,-0.5) {$\delta$};
\draw[dotted] (-3,-1) -- (-3,0);
\draw[dotted] (-3,1) rectangle (-9,0);
\node at (-4,0.5) {$\lcsext(L)$};
\draw[dotted] (-8,-2) rectangle (0,-1);
\draw (0,-1) rectangle (2,-2);
\node at (-4,-1.5) {$\lcsext(L')$};
\node at (1,-1.5) {$\lcs(L')$};
\end{tikzpicture}
}
\end{center}
For instance, consider $L=\{a,baa\}$ and $L'=\{aa,baaa\}$ s.t.\ $\sig(L)=(a,(ba)^\iomega)$ and $\sig(L')=(aa,(ba)^\iomega)$.
Applying Lemma~\ref{lem:lcs-computation-main-text}, we obtain for the union
$\lcs(L\cup L')=\lcs\bigr(a,aa\bigl)=a$ and 
$\lcsext(L\cup L')= \lcs\bigl((ba)^\iomega,\lcs((ba)^\iomega,(ba)^\iomega a)a\bigr) =a$.
In case of the concatenation, Lemma~\ref{lem:lcs-computation-main-text} yields
$\lcs(LL')=\lcs\bigr(a,(ba)^\iomega\bigl)aa=aaa$
and 
$\lcsext(LL')=\lcs\bigl((ba)^\iomega, (ab)^\iomega)=\ew$.
\end{example}

As both the $\lcs$ and the $\lcsext$ are determined by already two words (cf.~\Cref{lem:lcs-witness-main-text,lem:lcsextwit-main-text}),
it follows that every $L\subseteq\als$ is $\tseq$-equivalent to some sublanguage $\tsn(L)\subseteq L$ consisting of at most three words where the words $xR,yR$ can be chosen arbitrarily up to the stated constraints (with $R=\lcs(L)$):
\[
\tsn(L):=\begin{cases}
L & \text{ if } \abs{L}\le 2\\
\{R,xR,yR\} & \text{ if } \{R,xR,yR\} \subseteq L\wedge \lcsext(L)=\lcs(x^\iomega,y^\iomega)\\
\{xR,yR\} & \text{ if } R=\lcs(xR,yR)\wedge R\not\in L\wedge \{xR,yR\}\subseteq L
\end{cases}
\]

\renewcommand{\al}{\textsf{A}}

\paragraph*{Deciding well-formedness}
\renewcommand{\al}{\sA}
\noindent
For the following, we assume that $G$ is a context-free grammar over $\Br=\al\cup\ial$ with nonterminals $\vars$. Set $N:=\abs{\vars}$.
We further assume that $G$ is nonnegative, 
and that we have computed for every nonterminal $X$ of $G$ a word $r_X \in \als$ (represented as an {\SLP}) s.t.\ $\abs{r_X} = d_X$ and $\inv{r_X}\in \pfcl(\rd(L_X))$.\footnote{
$\inv{r_X}$ is (after reduction) a longest word of closing brackets in $\rd(L_X)$ (if $G$ is $\wf$, then $r_X$ is unique).
An {\SLP} encoding $r_X$ can be computed in polynomial time while checking that $G$ is nonnegative;
see Definition~\ref{def:wf-main-text} and the subsequent explanations, and the proof of \refA{lem:r_X-main-text}{\Cref{lem:r_X-main-text}}.
All required operations on words run in time polynomial in the size of the \SLP s representing the words, see e.g.~\cite{Lohrey2012}.
}
In order to decide whether $G$ is $\wf$ we compute the languages $\rd(r_XL_X^{\le h})$ modulo $\tseq$ for increasing derivation height $h$ using fixed-point iteration. 
Assume inductively that (i) $r_X L_X^{\le h}$ is $\wf$ and (ii) that we have computed $\sT_X^{\le h}:=\tsn(\rd(r_XL_X^{\le h}))\tseq \rd(r_XL_X^{\le h})$
 for all $X\in\vars$ up to height $h$. Then we can compute $\tsn(\rd(r_XL_X^{\le h+1}))$ for each nonterminal as follows:
\[
{\small\begin{array}{cl}
 & \rd(r_X L_X^{\le h+1})\\
 = & \rd(r_X L_X^{\le h})
\cup  \bigcup_{X\to_G Y} \rd(r_X\inv{r_Y}\ r_Y L_Y^{\le h})
\cup  \bigcup_{X\to_G YZ} \rd(r_X \inv{r_Y}\ r_Y L_Y^{\le h}\ \inv{r_Z}\ r_Z L_Z^{\le h})\\
\tseq & \sT_X^{\le h}
\cup  \bigcup_{X\to_G Y} \rd(r_X\inv{r_Y}\ \sT_Y^{\le h})
\cup  \bigcup_{X\to_G YZ} \rd(r_X \inv{r_Y}\ \sT_Y^{\le h}\ \inv{r_Z}\ \sT_Z^{\le h})\\
 \tseq & \tsn\Bigl(\rd\Bigl(\sT_X^{\le h}
 \cup \bigcup_{X\to_G Y} r_X\inv{r_Y}\ \sT_Y^{\le h}
 \cup \bigcup_{X\to_G YZ} r_X \inv{r_Y}\ \sT_Y^{\le h}\ \inv{r_Z}\ \sT_Z^{\le h}\Bigr)\Bigr) =: \sT_X^{\le h+1}\\
\end{array}}
\]
Note that, if all constants $r_X\inv{r_Y}$ and all $\sT_{X}^{\le h}$ are $\wf$, but $G$ is not $\wf$,
then the computation has to fail while computing $r_X\inv{r_Y} \sT_Y^{\le h} \inv{r_Z}$;
see the following example.
\begin{example}\label{ex:the-one}
Consider the nonnegative context-free grammar $G$ given by the rules (with the parameter $n\in\N$ fixed) 
\[
\begin{array}{lcl@{\quad}lcl@{\quad}lcl@{\quad}lcl@{\quad}lcl@{\quad}lcl}
S & \to & Uc & U & \to & AV \mid W_{n} & V       & \to & U\inv{B} & W_i & \to & W_{i-1} W_{i-1} & (2 \leq i \leq n)\\
A & \to & a  & B & \to & b             & \inv{B} & \to & \inv{b}  & W_1 & \to & B B\\
\end{array}
\]
with axiom $S$. Except for $\inv{B}$ all nonterminals generate nonnegative languages.
Note that the nonterminals $W_n$ to $W_1$ form an {\SLP} that encodes the word $b^{2^n}$ by means of iterated squaring which only becomes productive at height $h=n+1$.
For $h\ge n+3$ we have:
\[
\begin{array}{lcl@{\quad}lcl@{\quad}lcl}
L_S^{\le h} & = & \{a^k b^{2^n} \inv{b}^k c \mid k\le \lfloor\frac{h-(n+3)}{2}\rfloor\}\\
L_U^{\le h} & = & \{a^k b^{2^n} \inv{b}^k \mid k\le \lfloor\frac{h-(n+2)}{2}\rfloor\} & 
L_{W_i}^{\le h} & = & \{b^{2^i}\} & 
L_B^{\le h} & = & \{b\}\\
L_V^{\le h} & = & \{a^k b^{2^n} \inv{b}^{k+1} \mid k\le \lfloor\frac{h-(n+3)}{2}\rfloor\} & 
L_A^{\le h} & = & \{a\} & 
L_{\inv{B}}^{\le h} & = & \{\inv{b}\}\\
\end{array}
\]
Here the words $r_X$ used to cancel the longest prefix of closing brackets (after reduction)
are $r_S=r_U=r_V=r_W=r_A=r_B= \ew$ and $r_{\inv{B}}=b$.
Note that $r_X L_X^{\le h}$ is $\wf$ for all nonterminals $X$ up to $h\le h_0=2^{n+1}+(n+2)$ s.t.\ 
$\tsn(\rd(r_SL_S^{\le h}))\tseq \sT_S^{\le h} = \{b^{2^n}c, a^k b^{2^{n}-k(h)}c\}$ 
for $k(h)=\lfloor (h-(n+3))/2\rfloor$ and $n+3\le h\le h_0$;
in particular, the $\lcs$ of $\sT_S^{\le h}$ has already converged to $c$ at $h=n+3$, 
only its maximal extension $\lcsext$ changes for $n+3\le h \le h_0$.
We discover the first counterexample $a^{2^n}\inv{b}$ that $G$ is not $\wf$
while computing $\sT_V^{\le h_0+1}=\tsn(\rd(\sT_U^{\le h_0}\inv{b}))$.
\end{example}

As illustrated in Example~\ref{ex:the-one},
if $G$ is not $\wf$, then the minimal derivation height $h_0+1$ 
at which we discover a counterexample might be exponential in the size of the grammar. The following lemma states that up to this derivation height $h_0$
the representations $\sT_X^{\le h}$ cannot have converged (modulo $\tseq$).
\begin{lemma}\label{thm:not-wf--main-text}
If $L=L(G)$ is not $\wf$, then there is some least $h_0$ s.t.\ $r_X L_Y^{\le h_0}\inv{r_Z}$ is not $\wf$ with $X\to_G YZ$.
For $h\le h_0$, all $r_X L_X^{\le h}$ are $\wf$ s.t.\ $\sT_X^{\le h}\tseq \rd(r_X L_X^{\le h})$.
If $h_0 \ge 4N+1$, then at least for one nonterminal $X$ we have $\sT_X^{\le 4N+1}\not\tseq \sT_X^{\le 4N}$.
\end{lemma}
\noindent
The following Lemma~\ref{thm:converge-wf-main-text} states the complementary result, i.e.\ 
if $G$ is $\wf$ then the representations $\sT_X^{\le h}$ have converged at the latest for $h=4N$ modulo $\tseq$.
The basic idea underlying the proof of Lemma~\ref{thm:converge-wf-main-text}
is similar to~\cite{DBLP:conf/stacs/LuttenbergerPS18}:
we show that from every derivation tree of height at least $4N+1$ we can construct a derivation tree of height at most $4N$ such that both
trees carry the same information w.r.t.\ the $\lcs$ (after reduction).
\iftrue
In contrast to~\cite{DBLP:conf/stacs/LuttenbergerPS18} we need not only to show that $\sT_X^{\le 4N}$ has the same $\lcs$ as $\rd(r_XL_X^{\le 4N})$, but that $\sT_X^{\le 4N}$ has converged modulo $\tseq$ if $G$ is $\wf$; to this end, we need to explicitly consider $\lcsext$, and re-prove stronger versions of the results regarding the combinatorics on words which take the involution into account (see \refA{lem:ustst,lem:uststw,lem:pospos,lem:pospos,lem:negpos}{A.6}).
\else
In contrast to~\cite{DBLP:conf/stacs/LuttenbergerPS18} we need not only to show that $\sT_X^{\le 4N}$ has the same $\lcs$ as $\rd(r_XL_X^{\le 4N})$, but that $\sT_X^{\le 4N}$ has converged modulo $\tseq$ if $G$ is $\wf$; to this end, we need to explicitly consider $\lcsext$, and re-prove stronger versions of the results regarding the combinatorics on words which take the involution into account (see Section~\ref{app:slg} in the appendix of the extended version).
What prevents us to apply the results of~\cite{DBLP:conf/stacs/LuttenbergerPS18}
to the reduced $\lcs$
is, roughly spoken, that given a well-formed linear grammar 
of the following form
$S \to u X$ and $X \to s_i X \inv{r_i} t_i r_i$ and $X\to \ew$ (with $u,w,s_i,r_i,t_i\in\als$),
we cannot in general find suitable conjugates of $u,w,s_i,t_i,r_i$ that allow us to cancel the factors $\inv{r_i}$ in each rule
while preserving the structure of the grammar and its language after reduction; see the following example.\footnote{
Still, we can proceed as in Lemma~\ref{l:latest} to remove closing brackets for all nonterminals that produce an ultimately periodic language after reduction;
although, we currently do not know if we can transform in polynomial time a well-formed context-free grammar $G$ over $\Br$ into a context-free grammar $G'$ over $\al$ s.t.\ $\rd(L(G))=L(G')$ (and further s.t.\ derivations are in bijection).}
\begin{example}\label{ex:pospos}
Consider the grammar given by the rules
$S\to ut_1r X$ and $X\to s_1 X \inv{r} t_1 r \mid s_2 X t_2 \mid \ew$
where we assume that 
(i) $R:=\lcs(\rd(L))=t_1 r$, 
(ii) $r = r' s_1$ with $r'\neq \ew$,
(iii) $r\not\sle rs_2$ (i.e.\ there is no conjugate of $s_2$ w.r.t.\ $r$), and
(iv) $t_2 = t_2' R$ with $t_2'\neq\ew$.
As the grammar is $\wf$, 
there is a conjugate $\hat{s}_1$ s.t.\ $r's_1 = \hat{s}_1 r'$ that allows us 
to cancel $\inv{r}$ in $r s_1\inv{r}\rdeq r' s_1 \inv{r'} \rdeq \hat{s}_1$.
Subsequently, there is a conjugate $\check{s}_2$ of $s_2$ with $r's_2 =\check{s}_2r'$
that allows us to cancel $r s_2 s_1 \inv{r} \rdeq \hat{s}_1 r' s_2 \inv{r'} \rdeq \hat{s}_1 \check{s}_2$.
These conjugates allow us to remove the closing brackets but only by splitting the grammar depending on which rules are used in a derivation:
\[
S \to u R \mid  ut_1\hat{s}_1 [r'X] R \qquad [r'X] \to \hat{s}_1 [r'X] t_1 \mid \check{s}_2 [r'X] Rt'_2 \mid \ew \mid \check{s}_2 r' t_2'\qedhere
\]
See Example~\ref{ex:pospos-app} in the appendix of the extended version for the detailed calculation.\qed
\end{example}
\fi
\begin{lemma}\label{thm:converge-wf-main-text}
Let $G$ be a context-free grammar with $N$ nonterminals and $L(G)$ be $\wf$.
For every nonterminal $X$ let $r_X\in\als$ s.t.\ $\abs{r_X}=d_X$ and $r_XL_X$ $\wf$.
Then $\rd(r_X L_X) \tseq \rd(r_X L_X^{\le 4N})$, and $\sT_X^{\le 4N}\tseq \sT_X^{\le 4N+1}$ for every nonterminal $X$.
\end{lemma}
\iftrue
The following example sketches the main idea underlying the proof of Lemma~\ref{thm:converge-wf-main-text}.
\begin{example}\label{ex:lcpconv}
The central combinatorial observation\footnote{This observation strengthens the combinatorial results in~\cite{DBLP:conf/stacs/LuttenbergerPS18} and also allows to greatly simplify the original proof of convergence given there.} 
is that for any well-formed language $\mathcal{L}\subseteq\Br^\ast$ of the form 
\[
\mathcal{L}=(\alpha,\beta)[(\mu_1,\nu_1)+(\mu_2,\nu_2)]^\ast \gamma := \{\alpha \mu_{i_1} \ldots \mu_{i_l} \gamma \nu_{i_l} \ldots \nu_{i_1} \beta \mid i_1\ldots i_l \in \{1,2\}^\ast\}
\]
we have that its \emph{longest common suffix after reduction} $\rlcs(\mathcal{L}):=\lcs(\rd(\mathcal{L}))$
is determined by the reduced longest common suffix of $\alpha\gamma\beta$ and either $(\alpha,\beta)(\mu_i,\nu_i)\gamma=\alpha \mu_i \gamma \nu_i \beta$ or $(\alpha,\beta)(\mu_i,\nu_i)(\mu_j,\nu_j)\gamma=\alpha \mu_i \mu_j \gamma \nu_j \nu_i \beta$ for some $i\in\{1,2\}$ but \emph{arbitrary} $j\in\{1,2\}$ in the latter case.\footnote{
To clarify notation, we set $(\alpha,\beta)(\mu,\nu):=(\alpha\mu,\nu\beta)$ and $(\alpha,\beta)\gamma:=\alpha\gamma\beta$,
i.e.\ the pair $(\alpha,\beta)$ is treated as a word with a ``hole'' into which the pair or word on the right-hand side is substituted.}

Assume now we are given a context-free grammar $G$ with $N$ variables.
Further assume that $L:=L(G)$ is well-formed.
W.l.o.g.\ $G$ is in Chomsky normal form and reduced to the productive nonterminals reachable from the axiom of $G$.
Let $L^{\le 4N}$ denote the sublanguage of words generated by $G$ with a derivation tree of height at most $4N$.
Pick a shortest (before reduction) word $\kappa_0\in L:=L(G)$.
Then there is some $\kappa_1\in L$ with $R:=\rlcs(L)=\rlcs(\kappa_0,\kappa_1)$;
we will call any such word a \emph{witness (w.r.t.\ $\kappa_0$)} in the following.
If $R\in L$, then $\kappa_0\rdeq R$, and any word in $L$ is a witness. In particular, there is a witness in $L^{\le 4N}$.
So assume $R\not\in L$.
Then we may factorize (in a unique way) $\kappa_0=\kappa_0''a\kappa_0'$ and $\kappa_1=\kappa_1''b\kappa_1'$ such that
$\rd(\kappa_0')=R=\rd(\kappa_1')$ where $a,b\in\al$ with $a\neq b$.
Then $\rd(\kappa_0)=z_0'aR$ and $\rho(\kappa_1)=z_1'bR$. 
Further assume that $\kappa_1\not\in L^{\le 4N}$, otherwise we are done.
Fix any derivation tree $t$ of $\kappa_1$, 
and fix within $t$ the \emph{main} path from the root of $t$ to the last letter $b$ of the suffix $b\kappa_1'$ of $\rd(\kappa_1)$ (the dotted path in Figure~\ref{fig:lcpHeight}).
We may assume that any path starting at a node on this main path and then immediately turning left towards a letter within the prefix $\kappa_1''$
consists of at most $N$ nonterminals: if any nonterminal occurs twice the induced pumping tree can be pruned without changing the suffix $b\kappa_1'$; as the resulting tree is still a valid derivation tree w.r.t.\ $G$, we obtain another witness w.r.t.\ $\kappa_0$.
\begin{figure}[!h]
\centering
\scalebox{0.7}{
\begin{tikzpicture}
\draw (0,0) -- (-5,-5) -- (5,-5) -- (0,0);
\draw[dashed] (0,0) -- (0.5,-1) -- (-0,-2) -- (0.5,-3) -- (0.5,-4) -- (0,-5);
\draw[dotted] (0,0) -- (0.5,-1) -- (-0,-2) -- (0.5,-3) -- (0.5,-3.75) -- (-1,-5);
\draw (0.5,-1) -- (4.5,-5);
\draw (0.5,-1) -- (-3.5,-5);
\draw (-0,-2) -- (-3,-5);
\draw (0,-2) -- (3,-5);
\draw (0.5,-3) -- (-1.5,-5);
\draw (0.5,-3) -- (2.5,-5);
\draw (0.5,-4) -- (-0.5,-5);
\draw (0.5,-4) -- (1.5,-5);

\node at (0.5-0.3,-1) {$X$};
\node at (-0+0.3,-2) {$X$};
\node at (0.5-0.3,-3) {$X$};
\node at (0.5+0.3,-4) {$X$};

\draw (-5,-5) -- (-5,-6) -- (5,-6) -- (5,-5);
\draw[dashed] (-5,-5.5) -- (5,-5.5);

\node at (-4.25,-5.25) {$u$};
\draw[dotted] (-3.5,-5) -- (-3.5,-5.5);
\node at (-3.25,-5.25) {$s_1$};
\draw[dotted] (-3,-5) -- (-3,-5.5);
\node at (-2.25,-5.25) {$s_2$};
\draw[dotted] (-1.5,-5) -- (-1.5,-5.5);
\node at (-1,-5.25) {$s_3$};
\draw[dotted] (-0.5,-5) -- (-0.5,-5.5);
\node at (0.5,-5.25) {$w$};
\draw[dotted] (1.5,-5) -- (1.5,-5.5);
\node at (2,-5.25) {$\tau_3$};
\draw[dotted] (2.5,-5) -- (2.5,-5.5);
\node at (2.75,-5.25) {$\tau_2$};
\draw[dotted] (3,-5) -- (3,-5.5);
\node at (3.75,-5.25) {$\tau_1$};
\draw[dotted] (4.5,-5) -- (4.5,-5.5);
\node at (4.75,-5.25) {$v$};

\draw (-1.2,-6) -- (-1.2,-5.5);
\draw (-0.8,-6) -- (-0.8,-5.5);
\node at (-1,-5.75) {$b$};
\node at (-2,-5.75) {$\kappa_1''$};
\node at (3.5,-5.75) {$\kappa_1'\rdeq R$};

\end{tikzpicture}
}
\caption{Factorization of a witness $\kappa_1 = (u,v) (s_1,t_1)(s_2,t_2)(s_3,t_3)w=\kappa_1''b\kappa_1'$ w.r.t.\ a nonterminal $X$ occurring at least four times a long the dashed path in a derivation tree of $\kappa_1$ leading to a letter within the suffix $b\kappa_1'$.
The dotted path depicts the main path leading to the $\rlcs$-delimiting occurrence of the letter $b$.
\label{fig:lcpHeight}}
\end{figure}
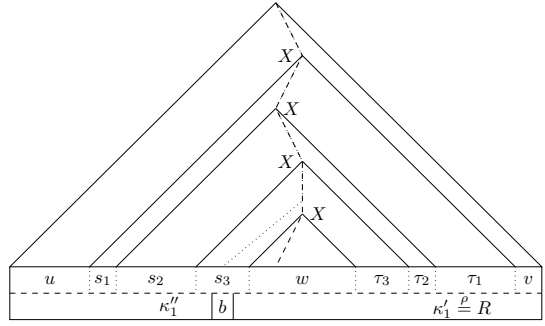
Thus consider any path (including the main path) in $t$ that leads from its root to a letter within the suffix $b\kappa_1'$.
If every such path consists of at most $3N$ nonterminals, then every path in $t$ consists of at most $4N$ nonterminals so that $\kappa_1\in L^{\le 4N}$ follows.
Hence, assume there is at least one such path consisting of $3N+1$ nonterminals.
Then there is some nonterminal $X$ occurring at least four times on this path.
Fix four occurrences of $X$ 
and factorize $\kappa_1$ accordingly
\[
\kappa_1 = us_1s_2s_3w\tau_3\tau_2\tau_1v =: (u,v) (s_1,\tau_1)(s_2,\tau_2)(s_3,\tau_3)w
\]
In the proof of Lemma~\ref{thm:converge-wf-main-text} we show that we may assume --- as $L$ is well-formed --- that $u,v,w,s_1,s_2,s_3\in\als$ with only $\tau_1,\tau_2,\tau_3\in\Br^\ast$.
From this factorization we obtain the sublanguage
$L' := (u,v) [(s_1,\tau_1)+(s_2,\tau_2)+(s_3,\tau_3)]^\ast w$. 
Our goal is to show that 
$(u,v)w$ or $(u,v)(s_i,\tau_i)w$ or $(u,v)(s_i,\tau_i)(s_j,\tau_j)w$ (for $i\neq j$)
is a witness w.r.t.\ $\kappa_0$: note that each of these words result from pruning at least one pumping tree from $t$ which inductively leads to a procedure to reduce $t$ to a derivation tree of height at most $4N$ that still yields a witness for $R=\rlcs(L)$ w.r.t.\ $\kappa_0$.
Assume thus specifically that neither $(u,v)w=uwv$ nor $(u,v)(s_3,\tau_3)w=us_3w\tau_3v$ nor $(u,v)(s_i,\tau_i)(s_3,\tau_3)w=us_is_3w\tau_3\tau_iv$ for $i\in\{1,2\}$  is a witness w.r.t.\ $\kappa_0$, i.e.\  each of these words end on $aR$ after reduction.
Apply now the result mentioned at the beginning of this example to the language
$L'' := (u,v) [(s_1,\tau_1)+(s_2,\tau_2)]^\ast (s_3,\tau_3) w$: by our assumptions $\kappa_1\in L''$ is a witness w.r.t.\ $us_3w\tau_3v\in L''$ so that both $\rlcs(L'')=\rlcs(L)$ and also $(u,v) (s_1,\tau_1)^2(s_3,\tau_3)w$ is a witness w.r.t.\ $us_3w\tau_3v$ as we may choose $j=1$.
Thus also $\rlcs(L)=\rlcs(L''')$ for
$L''':=(u,v) [(s_1,\tau_1)+(s_3,\tau_3)]^\ast w$
as $L'''\subseteq L$ and both $(u,v)w\in L'''$ and $(u,v)(s_1,\tau_1)(s_1,\tau_1)(s_3,\tau_3)w\in L'''$.
Applying the same argument now to $L'''$, but \emph{choosing} $j\neq i$ it follows that $(u,v)(s_1,\tau_1)w$ or $(u,v)(s_3,\tau_3)(s_1,\tau_1)w$ has to be a witness w.r.t.\ $\kappa_0$.

The sketched argument can be adapted so that it also allows to conclude 
the maximal extension after reduction $\lcsext(\rd(L))$ has to have converged at derivation height $4N$ the latest, if $L$ is well-formed. 
For details, see the proof of \refA{thm:converge-wf-main-text}{\Cref{thm:converge-wf-main-text}}.
\end{example}
\fi
\noindent
As $\abs{\sT_X^{\le h}}\le 3$, a straight-forward induction also shows that every word in $\sT_{X}^{\le h}$ can be represented by an {\SLP} that we can compute in time polynomial in $G$ for $h\le 4N+1$;
together with the preceding \Cref{thm:converge-wf-main-text,thm:not-wf--main-text}
we thus obtain the main result of this section:
\begin{theorem}\label{t:xyz}
Given a context-free grammar $G$ over $\Br$ we can decide in time polynomial in the size of $G$ whether $G$ is \wf.
\end{theorem}
\noindent

%% file: conclusion.tex
\section{Conclusion}\label{s:conclusion}
We have shown that well-formedness for context-free languages is decidable in polynomial time.
This allowed us to decide in polynomial time whether or not a \oneLTW\ is balanced.
The presented techniques, however, are particularly tailored for 
\oneLTW s. It is unclear how a generalization to transducers processing three or more
copies of the input would look like.
Thus, the question remains whether balancedness is decidable for general MSO 
definable transductions.
It is also open whether even the single bracket case can be generalized beyond
MSO definable transduction, e.g.\ to output languages of top-down tree-to-word 
transducers \cite{SeidlMK18}.

%% file: proofs.tex
\section{Appendix}
\ifappendixtoc
\secttoc
\fi

The appendix is self-contained,
i.e.\ we restate all definitions and lemmata of the main work, give the missing proofs, and also introduce further lemmata needed to prove the main results.

\input{proofs-lcs-basics}

\input{proofs-rlcs-linear}

\input{proofs-wf-convergence}

\input{proofs-not-wf}

%% file: proofs-lcs-basics.tex
\subsection{Properties of the longest common prefix and suffix}\label{app:lcs}
\renewcommand{\al}{\Sigma}

Fix some generic nonempty finite alphabet $\al$ with $\top\not\in\al$ a fresh, unused symbol. We write $\al_{\top}^\infty$ for $\{\top\}\cup\als\cup\al^\omega$.
As mentioned in the main work,
the $\lcp$ is the infimum w.r.t.\ the prefix order $\ple$ on $\al_\top^\infty$ extended by a greatest element $\top$ in order to handle the empty set. 
We briefly sketch the argument:
For $u\in\al^\infty$, set $h(u):=u\al^\omega$, if $u\in\als$, and $h(u):=\{u\}$ otherwise.
We then have $h(u)=\{w\in\al^\omega \mid u\ple w\}$, 
i.e.\ we can alternatively define $\ple$ by means of $u\ple v:\Leftrightarrow h(u)\supseteq h(v)$.
Thus, $\top$ becomes the greatest element w.r.t.\ $\ple$ by setting $h(\top):=\emptyset$.
Extend $h$ to languages $L\subseteq \al^\infty$ by means of $\hat{h}(L):=\bigcup_{w\in L} h(w)$.
The $\lcp(L)$ is then the unique word in $\al_\top^\infty$ satisfying
$h(\lcp(L))=\hat{h}(\{\lcp(L)\})=\bigcap\{ z\al^\omega \mid \hat{h}(L)\subseteq z\al^\omega\}$,
and thus is the infimum w.r.t.\ $\ple$.

In the following, we summarize some properties of the $\lcs$ which are used in the following proofs.
For easier reference, we restate the definition of the $\lcs$ and $\alup$:
\begin{definition}[Definition~\ref{def:lcs} in the main work]
\hspace{1em}
\begin{itemize}
\item
For $u\in\als$ and $w\in\al^+$,
define the expression $w^\iomega u$ by means of $w^\iomega u:=\bigl(u^R (w^R)^\omega\bigr)^R$, and its reverse by means of $(u w^\omega)^R= (w^R)^\iomega u^R$. The set of \emph{ultimately left-periodic words} is then $\alup:=\{w^\iomega u\mid w\in\al^+, u\in\als\}$.
\item
The \emph{suffix order} on $\als\cup \alup\cup\{\top\}$ is then defined by $u\sle v:\Leftrightarrow u^R \ple v^R$.
\item
The \emph{longest common suffix ($\lcs$)} of a language $L\subseteq \als\cup\alup$ is then $\lcs(L):=\lcp(L^R)^R$.
\end{itemize}
\end{definition}

\noindent
Note that for the $\lcs$ a term like $\lcs(x^\iomega, \ldots)$ is always supposed to be read as $\lcs(\ldots xxxxx, \ldots)=\lcp((x^R)^\omega, \ldots)^R$.

We prove the following properties of the $\lcs$ which allow to simplify the computation of the $\lcs$, in particular in the case of ultimately periodic words.
\begin{lemma}\label{lem:lcs-calculus}
Let $u,v,w,x,y,z\in\als$.
\[
\begin{array}{lcl}
u^\iomega v & = & u^\iomega u^k v \quad (\forall k\ge 0)\\[2mm]
u\sle x^\iomega y & \text{ iff } & x=\ew \vee \exists k\colon u\sle x^k y\\[2mm]
u\sle w^\iomega \slt\top & \text{ iff } & u\slt uw\\
                    & \text{ iff } & \exists w',w'',k\colon w=w'w''\wedge w''\slt w\wedge u = w''w^k \\[2mm]
u\sle w^\iomega & \text{ iff } & u\sle uw\\
                      & \text{ iff } & \exists v\colon uw=vu \wedge v^\iomega u = w^\iomega\\[2mm]
u^\iomega v = w^\iomega & \text{ iff } & \exists p,q\in\al^\ast \colon pv=vq\wedge u\in p^\ast\wedge w\in q^\ast\\[2mm]
u^\iomega v = x^\iomega y & \text{ iff } & u=\ew=x\vee u\neq\ew\neq x\wedge \forall k\exists l\colon u^k v \sle x^l y \wedge x^ky\sle u^l v\\
                        & \text{ iff } & u=\ew=x\vee u\neq\ew\neq x\wedge \forall k,l\colon u^k v\inv{x^ly}\text{ is \twwf}\\ 
                        & \text{ iff } & u=\ew=x\vee \exists p,q\colon u\in q^+\wedge x\in p^+ \wedge py\inv{qv}\rdeq y\inv{v}\\[2mm]
\lcs(x,xy) & = & \lcs(x,xy^k) \quad (\forall k\ge 1)\\
           & = & \lcs(x,y^\iomega)\\[2mm]
\lcs(x^\iomega,y^\iomega) & = & \begin{cases}
y^\iomega & \text{ if } x = \ew\\
x^\iomega & \text{ if } xy=yx\wedge x\neq \ew\\
\lcs(yx,y^\iomega)=\lcs(xy,x^\iomega)=\lcs(xy^l,yx^k) =\lcs(xy,yx) & \text{ if } xy\neq yx, k,l>0\\
\end{cases}\\[2mm]
\lcs(u^\iomega,(wu)^\iomega) & = & \lcs(u^\iomega, w^\iomega)u\\[2mm]
\lcs(u^\iomega,w^\iomega,(wu)^\iomega) & = & \lcs(u^\iomega,w^\iomega)\\[2mm]
\end{array}
\]
\end{lemma}
\begin{prfblk}
\begin{enumerate}
\item
$\forall k\colon u^\iomega v = u^\iomega u^k v$ as

If $u=\ew$, then: 
\begin{block}
$u^\iomega v=\top v=\top=\top u^k v = u^\iomega u^k v$ for all $k$.
\end{block}
So assume $u\neq\ew$.

Then by definition for any $k\in\N_0$:
\[
u^\iomega v = u^\iomega u^k v \text{ iff } \bigcap_{n\ge 0} v^R (u^R)^{n} \al^\omega = \bigcap_{n\ge k} v^R (u^R)^{n} \al^\omega
\]
\item
$u\sle x^\iomega y \text{ iff } x=\ew \vee \exists k\colon u\sle x^k y$ as

If $x=\ew$, then:
\begin{block}
$u\sle x^\iomega y = \top$
\end{block}
So assume $x\neq\ew$.

If $u\sle x^\iomega y$, then:
\begin{block}
 There is some $k\in\N_0$ s.t.\ $u \le \abs{x^{k}y}$.
  
 Thus $u\sle x^\iomega y = x^\iomega x^k y \text{ iff } u\slt x^{k}y$.
\end{block}

If $\exists k\colon u\sle x^k y$, then:
\begin{block}
$u\sle x^ky \sle x^\iomega x^k y = x^\iomega y$
\end{block}
\item
$u\sle w^\iomega\slt\top \text{ iff } u\slt uw \text{ iff } \exists w',w'',k\colon w=w'w''\wedge w''\slt w\wedge u = w''w^k$ as

If $\exists w',w'',k\colon w=w'w''\wedge w''\slt w\wedge u = w''w^k$, then:
 \begin{block}
 $\ew\sle w''\slt w$
 
 Thus $u=w''w^k\slt w^{k+1}\slt w^\iomega \slt \top$.
 \end{block}

If $u\sle w^\iomega\slt \top$, then:
 \begin{block}
 $w\neq \ew$.
 
 Thus there is $k\in\N_0$ s.t.\  $w^k\sle u\slt w^{k+1}$ and thus $u=w''w^k$ for some factorization $w=w'w''$ with $w''\slt w$.
 
 Hence also $u\slt uw=w''w^k w = (w''w') w'' w^k = (w''w') u$.
 \end{block}
If $u\slt uw$, then:
  \begin{block}
  $w\neq \ew$.
  
  We show by induction on $\abs{u}$ that $u\sle w^\iomega\slt \top$.
  
  If $\abs{u}\le \abs{w}$, then:
  \begin{block}
  $w=w'u$ 
  
  Thus $u\sle w\slt w^\iomega\slt\top$
  \end{block}
  So assume $\abs{u}>\abs{w}$ s.t.\ $u=u'w$.
  
  Thus $u'\slt u'w$ as $u'w=u\slt uw=u'ww$.
  
  Hence by induction $u'\sle w^\iomega\slt\top$ and thus $u'\sle w^\iomega w=w^\iomega$.
  \end{block}
\item
$u\sle w^\iomega \text{ iff } u\sle uw \text{ iff } \exists v\colon uw=vu\wedge w^\iomega = v^\iomega u$ as

If $w=\ew$, then:
\begin{block}
$u\sle w^\iomega=\top$ and $u\sle uw=u$ and $\exists v\colon u=uw=wv=v \wedge w^\iomega =\top = v^\iomega u$ are trivially true.
\end{block}
So assume $w\neq\ew$.

If $u\sle w^\iomega\slt \top$, then:
\begin{block}
  $\exists w',w'',m\colon w=w'w''\wedge w''\slt w\wedge u=w''w^m$ by preceding result.
  
  Set $v:=w''w'$ s.t.\ $\abs{v}=\abs{w}>0$.
  
  Then $uw=w''w^mw=(w''w')w''w^m = vu$.
  
  Hence $\forall k\colon w^k \sle v^k u \sle v^\iomega v^k u = v^\iomega u \wedge  v^ku= u w^k \sle w^\iomega w^k = w^\iomega$ 
\end{block}
If $\exists v\colon uw=vu$, then:
\begin{block}
  $u\slt uw=vu$
  
  Thus $u\sle w^\iomega\slt\top$.
\end{block}
\item
$u^\iomega v = w^\iomega \text{ iff } \exists p,q\in\al^\ast \colon pv=vq\wedge u\in p^\ast\wedge w\in q^\ast$ as

W.l.o.g.\ assume $u\neq\top\neq w$ as otherwise $u^\iomega v =\top \vee w^\iomega =\top$ s.t.\ $u=\top=w$.

Let $p$ be the primitive root of $u$, and $q$ that of $w$.

Then $u^\iomega = p^\iomega$ and $w^\iomega=q^\iomega$ s.t.\ $v\sle w^\iomega=q^\iomega\slt \top$.

Thus $\exists \hat{q}\colon vq=\hat{q}v$ s.t.\ $p^\iomega v= q^\iomega=\hat{q}^\iomega v$.

So $p^\iomega = \hat{q}^\iomega$ and thus $p=\hat{q}$ as both are primitive.
\item
$\begin{array}{lcl}
u^\iomega v = x^\iomega y & \text{ iff } & u=\ew=x\vee u\neq\ew\neq x \wedge \forall k\exists l\colon u^k v \sle x^l y \wedge x^ky\sle u^l v\\
                        & \text{ iff } & u=\ew=x\vee u\neq\ew\neq x \wedge \forall k\forall l\colon u^k v\inv{x^ly} \text{ is \twwf}\\ 
                        & \text{ iff } & u=\ew=x\vee  \exists p,q\colon u\in q^+\wedge x\in p^+ \wedge py\inv{qv}\rdeq y\inv{v}\\[2mm]
\end{array}$

W.l.o.g.\ $u\neq\ew\neq x$.

If $u^\iomega v = x^\iomega y$, then: 
\begin{block}
  As $x\neq\ew\neq u$ for $k$ there is some $l$ s.t.\ $\abs{u^kv}\le \abs{x^ly}$ and $\abs{x^ky}\le \abs{u^lv}$.
  
  Thus:  
  $u^kv\sle u^\iomega u^k v =u^\iomega v =x^\iomega y=x^\iomega x^l y$ and $x^ky \sle x^\iomega x^ky = x^\iomega y = u^\iomega v = u^\iomega u^lv$
  
  Hence: $u^kv\sle x^ly$ and $x^k y\sle u^l v$.
  
  In particular, $u^kv$ and $x^ly$ have to be comparable w.r.t\ $\sle$ for all $k,l$.
  
  Thus $u^kv \inv{x^ly}$ is $\wwf$ for all $k,l$.
  
  Finally, as $x\neq\ew\neq u$ we have both $v\sle x^\iomega y \slt \top$ and $y\sle u^\iomega v\slt \top$.
  
  Let $y=y'v$ as w.l.o.g.\ $\abs{v}\le \abs{y}$.
  
  Further, let $q$ be the primitive root of $u$.
  
  Then $y'\sle u^\iomega = q^\iomega\slt \top$.
  
  Hence $\exists p\colon y' q= p y'\wedge q^\iomega = p^\iomega y'$ s.t.:
  \begin{itemize}
  \item
  $py \inv{qv} \rdeq py' \inv{q} = y'q\inv{q} \rdeq y'\rdeq y'v\inv{v}\rdeq y\inv{v}$.
  \item
  $x^\iomega y = u^\iomega v = q^\iomega v = p^\iomega y' v = p^\iomega y$ and $x^\iomega = p^\iomega$.
  \end{itemize}
  Finally $x\in p^+$ as $x\neq\ew$ and $q$ is primitive (as $y'q=py'$ and $p$ primitive).
\end{block}

If $\forall k\exists l\colon u^k v \sle x^l y \wedge x^ky \sle u^l v$, then:
\begin{block}
\[
\forall k\exists l\colon \bigcap_{n\ge 0} y^R (x^R)^n \al^\omega\subseteq y^R(x^R)^l \al^\omega \subseteq v^R (u^R)^k \al^\omega
\]
and thus
\[
\bigcap_{n\ge 0} y^R (x^R)^n \al^\omega \subseteq \bigcap_{k\ge 0} v^R (u^R)^k \al^\omega
\]
i.e.\ $u^\iomega v\sle x^\iomega y$ and symmetrically $x^\iomega y \sle u^\iomega v$ s.t.\ $x^\iomega y = u^\iomega v$.
\end{block}

If  $\forall k\forall l\colon u^k v\inv{x^ly} \text{ is \twwf}$, then
\begin{block}
$\forall k\exists l\colon u^k v \sle x^l y \wedge x^ky \sle u^l v$ directly holds.
\end{block}

If $\exists p,q\colon u\in q^+\wedge x\in p^+ \wedge py\inv{qv}\rdeq y\inv{v}$, then:
\begin{block}
W.l.o.g.\ $\abs{y}\ge\abs{v}$ s.t.\ $y=y'v$. Then:

$x^\iomega y = p^\iomega y = p^\iomega y' v = q^\iomega v = u^\iomega v$.
\end{block}
\item
For all $k\ge 1$: $\lcs(x,xy) = \lcs(x,xy^k) = \lcs(x,y^\iomega)$ as:

If $y=\ew$, then:
\begin{block}
$\lcs(x,xy)=x =\lcs(x,xy^k)=\lcs(x,\top)$
\end{block}
So assume $y\neq \ew$.

Then pick $m$ s.t.\ $x=x'y'' y^m$ for suitable $x',y',y''$ with \ $y=y'y''$ and $y''\slt y$ s.t.\

$\lcs(x,y)=y''y^m$ and $\lcs(x',y')=\ew$.

Hence for all $k\ge 1$:
\[\lcs(x,xy^k)=\lcs(x'y''y^m, x'y''y^{m+k})=\lcs(x'y''y^m,y^{m+1})=\lcs(x',y')y''y^{m}=\lcs(x,y^\iomega)\]
\item
$\lcs(x^\iomega,y^\iomega)  =\begin{cases}
y^\iomega & \text{ if } x = \ew\\
x^\iomega & \text{ if } xy=yx\wedge x\neq \ew\\
\lcs(yx,y^\iomega)=\lcs(xy,x^\iomega)=\lcs(xy^l,yx^k) =\lcs(xy,yx) & \text{ if } xy\neq yx, k,l>0\\
\end{cases}$ as:

W.l.o.g.\ $\abs{x}\le \abs{y}$.

If $x=\ew$, then:
\begin{block}
 $\lcs(x^\iomega, y^\iomega) = \lcs(\top,y^\iomega)=y^\iomega$
\end{block}

So $0<\abs{x}\le \abs{y}$.

If $x\not\sle xy\vee y\not\sle yx$, then:
\begin{block}
W.l.o.g.\ $x\not\sle xy$.

Then $\lcs(x,y^\iomega)=\lcs(x,xy^l)\slt x$ for all $l>0$.
\end{block}
So assume $x\sle xy\wedge y\sle yx$.

Then $\exists \hat{x},\hat{y}\colon xy=\hat{y}x\wedge yx=\hat{x}y$.

If $\hat{x}=x\vee \hat{y}=y$, then:
\begin{block}
$xy=yx$ as $xy=\hat{y}x=yx\vee yx=\hat{x}y=xy$.

Hence, $x^\iomega = y^\iomega$ as $x,y$ have the same primitive root.
\end{block}
So assume that $\hat{x}\neq x \wedge \hat{y}\neq y$.

$y=y'x$ as $\abs{x}\le \abs{y}$ and $y\sle yx=\hat{x}y$.

Hence $\hat{y}=xy'$ as $\hat{y}x = xy = xy'x$.

Then $\lcs(x^\iomega, y^\iomega ) = \lcs(\hat{x}^\iomega \hat{x} y, y^\iomega y' x y)=\lcs(\hat{x},x)y\slt xy$.

From that the second case of the statement follows.
\item
$\lcs(u^\iomega,(wu)^\iomega) = \lcs(u^\iomega, w^\iomega)u$ as
\[\begin{array}{lcl}
\lcs(u^\iomega,(wu)^\iomega) & = & \begin{cases}
w^\iomega  & \text{ if } uw=wu \wedge w\neq \ew\\
u^\iomega  & \text{ if } uw=wu \wedge w= \ew\\
\lcs(uwu,wuu)=\lcs(uw,wu)u  & \text{ if } uw\neq wu
\end{cases}\\
= \lcs(u^\iomega, w^\iomega)u
\end{array}\]
\item
$\lcs(u^\iomega,w^\iomega,(wu)^\iomega) = \lcs(u^\iomega,w^\iomega)$ as
\[\begin{array}{lcl}
\lcs(u^\iomega,w^\iomega,(wu)^\iomega) 
& = & \lcsa{\lcs(u^\iomega,w^\iomega)\\ \lcs(u^\iomega,(wu)^\iomega)}\\
& = & \lcsa{\lcs(u^\iomega,w^\iomega)\\ \lcs(u^\iomega, w^\iomega)u}\\
& = & \lcs(\lcs(u^\iomega,w^\iomega),u^\iomega)\\
& = &\lcs(u^\iomega,w^\iomega)
\end{array}\]
\end{enumerate}
\end{prfblk}

The next lemma formalizes that whenever $L$ is not empty, we can find for any $x\in L$ a witness $y\in L$ s.t.\ $\lcs(L)=\lcs(x,y)$ which we will use in the following frequently without explicitly referring to this lemma every time.
\begin{lemma}[Lemma~\ref{lem:lcs-witness-main-text} in the main work]
Let $L\subseteq\al^\ast$. Then both:
\begin{itemize}
\item
$\forall x\in L\colon \lcs(L)=\lcs(\lcs(x,y)\mid y\in L)$
\item
$\forall x\in L\exists y\in L\forall z\in L\colon \lcs(L)=\lcs(x,y)=\lcs(x,y,z)$
\end{itemize}
\end{lemma}
\begin{proof}
Trivially true if $L=\emptyset$.
Therefore, assume $L\neq \emptyset$.
Let $R=\lcs(L)$ and pick any $x\in L$.
If $x=R$, we have $\lcs(x,y)=R$ for all $y\in L$; so choose any $y\in L$.
Thus, assume $R\slt x$.
Let $L_x=\{ \lcs(x,y)\mid y\in L\}$ and $S:=\lcs(L_x)$.
Then $R\sle S$ as
for all $y\in L$ we have $R\sle y$ and thus $R\sle \lcs(x,y)$.
But also $S\sle R$ as
$\forall y\in L\colon S\sle \lcs(x,y)\sle y$.
Thus, $R\in L_x$ as
$\forall z\in L_x\colon R\sle z\sle x$.
Then there is some $y\in L$ s.t. $R=\lcs(x,y)$; by minimality of $R$ we trivially have that 
$\lcs(x,y,z)=\lcs(x,y)$ for all $z\in L$.
\end{proof}

We recall the definition of the maximal suffix extension:
\begin{definition}[Definition~\ref{def:lcsext-main-text} in the main work]
Let $L\subseteq\al^\ast$ with $R=\lcs(L)$.
\[
\lcsext(L):= \lcs( z^\iomega \mid zR \in L\setminus\{R\})
\]
\end{definition}
By definition of $\top$ we have both $\lcsext(\emptyset)=\lcs(\emptyset)=\top$ and $\lcsext(\{R\})=\lcs(\ew^\iomega)=\top$.

\begin{lemma}
Let $\emptyset\neq L\subseteq\al^\ast$ and $R=\lcs(L)$.
If $R\not\in L$, then $\lcsext(L)=\ew$
\end{lemma}
\begin{proof}
As $R\not\in L$ it exists $uR,vR\in L$ such that $R=\lcs(uR,vR)=\lcs(u,v)R$ with $u \neq \ew \neq v$ and $\lcs(u,v)=\ew$.
Thus, $\lcsext(L)\sle \lcs(u^\iomega,v^\iomega)=\lcs(u,v)=\ew$.
\end{proof}

\subsection{Lemma~\ref{lem:lcsextwit-main-text} in the main work}

\begin{lemma}[Lemma~\ref{lem:lcsextwit-main-text} in the main work]\label{lem:lcsextwit}
Let $L\subseteq\al^\ast$ with $\abs{L}\ge 2$ and $R:=\lcs(L)$.
Fix any $xR\in L\setminus\{R\}$.
Then there is some $yR\in L\setminus\{R\}$ s.t.\
$\lcsext(L)=\lcs(x^\iomega,y^\iomega) =\lcs(x^\iomega,y^\iomega,z^\iomega)
$ for all $zR\in L$.
If $xy=yx$, then $R\in L$.
\end{lemma}
\begin{prfblk}
Note: $x\neq \ew$ as $xR\neq R$.

Let $p$ be the primitive root of $x$.

If $\forall zR\in L\colon zx=xz$, then:
\begin{block}
$\forall zR\in L\colon z\in p^\ast$

Thus $\lcsext(L)=p^\iomega=x^\iomega = z^\iomega = \lcs(x^\iomega, z^\iomega)$ for $z \neq \ew$.

Hence, choose $y=x\neq \ew$.

Further $R\in L$ as otherwise $R=\lcs(L)=p^iR$ for $p^iR$ the shortest word in $L$.
\end{block}

So assume $\exists zR\in L\colon zx\neq xz$, then:
\begin{block}
$z\neq \ew$

$\lcsext(L)=\lcs( \lcs(x^\iomega, y^\iomega)\mid yR\in L\setminus\{R\})\sle \lcs(x^\iomega,z^\iomega)=\lcs(xz,zx)$

Hence, there is some $yR\in L\setminus\{R\}$ s.t. $\lcsext(L)=\lcs(x^\iomega,y^\iomega)=\lcs(xy,yx)$
\end{block}
\end{prfblk}

\subsection{Lemma~\ref{lem:lcs-computation-main-text} in the main work}

We split the proof of Lemma~\ref{lem:lcs-computation-main-text} into Lemma~\ref{lem:lcs-union} (union) and Lemma~\ref{lem:lcs-concat} (concatenation).

\begin{lemma}[Lemma~\ref{lem:lcs-computation-main-text} (union) in the main work]\label{lem:lcs-union}
Let $L,L'\subseteq \al^\ast$ with $\sig(L)=(R,E)$ and $\sig(L')=(R',E')$. 
\[\sig(L\cup L')=\begin{cases}
(R',E') & \text{ if } R=\top\\
(R,E) & \text{ if } R'=\top\\
(R,\lcs(E,\lcs(E',E'\delta)\delta)) & \text{ if } R'=\delta R\slt \top\\
(R',\lcs(E',\lcs(E,E\delta)\delta)) & \text{ if } R=\delta R'\slt \top\\
(\lcs(R,R'),\ew) & \text{ else }\\
\end{cases}
\]
\end{lemma}
\begin{prfblk}
If $R'=\top$:
\begin{block}
$L'=\emptyset$

$\sig(L\cup L')=\sig(L)=(R,E)$
\end{block}
The case $R=\top$ is symmetric.

W.l.o.g.\ $R\neq \top \neq R'$ from here on.

If $R\not\sle R'\wedge R'\not\sle R$:
\begin{block}
$\sig(L\cup L')=(\lcs(R,R'),\ew)$
\end{block}

W.l.o.g.\ $R\sle R'=\delta R$ from here on s.t.\ $\lcs(L\cup L')=R$.

If $R\not\in L$, then:
\begin{block}
$\lcsext(L)=\ew$ and thus $\lcsext(L\cup L')=\ew=\lcs(\lcsext(L),\lcs(\lcsext(L'),\lcsext(L')\delta)\delta)$.
\end{block}
So assume $R\in L$.

If $R'\not\in L'$, then:
\begin{block}
For suitable $xR',yR'\in L'$.
\[\ew=\lcsext(L')=\lcs(x^\iomega,y^\iomega)=\lcs(x,y)\]
Thus 
\[\lcs((z\delta)^\iomega\mid z\delta R \in L') = \lcs(x\delta,y\delta)=\delta\]
and hence 
\[\lcsext(L\cup L')=\lcs(\lcsext(L),\delta)=\lcs(\lcsext(L),\lcs(\lcsext(L'),\lcsext(L')\delta)\delta)\]
\end{block}
Thus also assume that $R'\in L'$. Then:
\[
\lcsext(L\cup L')=\lcs(w^\iomega, \delta^\iomega, (z\delta)^\iomega\mid z\delta R \in L', wR \in L)=\lcsa{\lcsext(L),\\ \lcs(\lcs(\delta^\iomega, (z\delta)^\iomega)\mid z\delta R \in L')}
\]
As shown before
\[\lcs(\delta^\iomega, (z\delta)^\iomega)=\lcs(\delta^\iomega, z^\iomega)\delta\]
Thus
\[
\lcs(\lcs(\delta^\iomega, (z\delta)^\iomega)\mid z\delta R \in L')=\lcs(\delta^\iomega, \lcsext(L'))\delta = \lcs(\lcsext(L'),\lcsext(L')\delta)\delta
\]
Hence again
\[
\lcsext(L\cup L') = \lcs(\lcsext(L), \lcs(\lcsext(L'),\lcsext(L')\delta)\delta)
\]

\end{prfblk}

\begin{lemma}[Lemma~\ref{lem:lcs-computation-main-text} (concatenation) in the main work]\label{lem:lcs-concat}
Let $L,L'\subseteq \al^\ast$ with $\sig(L)=(R,E)$ and $\sig(L')=(R',E')$.
\[
\sig(LL')=\begin{cases}
(\top,\top) & \text{ if } RR'=\top\\
(\lcs(R,E')R', \ew) & \text{ if } RR'\slt\top\wedge R\not\sle E'\\
(RR', E) & \text{ if } RR'\slt \top \wedge E'=\top\\
(RR',\lcs(E,\delta)) & \text{ if } RR'\slt \top \wedge E' = \delta R
\end{cases}
\]
\end{lemma}
\begin{prfblk}
If $R=\top\vee R'=\top$, then:
\begin{block}
$L=\emptyset$ or $L'=\emptyset$ s.t. $LL'=\emptyset$.

$\sig(LL')=(\top,\top)$
\end{block}
So assume $R\neq\top\neq R'$, i.e.\ $L\neq\emptyset\neq L'$.

If $E'=\top$, then:
\begin{block}
$L'=\{R'\}$ s.t.\

$\sig(LL')=\sig(LR')=(RR', E)$
\end{block}

Thus assume also that $E'\neq\top$ s.t.\ $\abs{L'}\ge 2$.

Fix $xR',yR'\in L'\setminus\{R'\}$ s.t.\ 
\[E'=\lcsext(z^\iomega \mid zR'\in L')=\lcs(x^\iomega, y^\iomega)\]

Consider
\[
\lcs(LL')=\lcs( wRzR'\mid wR\in L, zR'\in L') 
\]
If $E'=\ew$, then:
\begin{block}
\[\lcs(LL')=\lcs( wRzR'\mid wR\in L, zR'\in L')\sle \lcs(wRxR',wRyR'\mid wR\in L) =\lcs(xR',yR')=R'\]
and
\[\lcsext(LL') = \lcs((wRz)^\iomega \mid wR\in L, zR'\in L') \sle \lcs(wRx)^\iomega, (wRy)^\iomega \mid wR\in L)\sle \lcs(x,y)=\ew\]
and hence 
\[\sig(LL')=(R',\ew)=(\lcs(R,E')R',\ew)\]
\end{block}

So assume $E'\neq \ew$ s.t.\ $R'\in L'$. Then
\[
\lcs(LL')=\lcs( wRzR'\mid wR\in L, zR'\in L') =\lcsa{R\\ \lcs(\lcs(wR,wRz) \mid wR\in L, zR'\in L')}R'
\]
Consider then:
\[
\lcs(R,E')=\lcs(R,\lcsext(L'))=\lcs(R,\lcs(z^\iomega \mid zR'\in L')) = \lcs(R,Rx,Ry)
\]
If $R\not\sle E'$, then:
\begin{block}
$R\sgt \lcs(R,E')= \lcs(R,Rx,Ry)$

Thus also $\lcs(LL')=\lcs(RL')=\lcs(R,E')R'\slt RR'$.

So $\sig(LL')=(\lcs(R,E')R',\ew)$.
\end{block}

So assume $R\sle E'$, i.e.\ $\forall zR'\in L'\setminus\{R'\}\colon R\sle z^\iomega$ s.t.:
\[\forall zR'\in L'\setminus\{R'\}\exists \hat{z}\colon Rz=\hat{z}R \wedge z^\iomega = \hat{z}^\iomega R\]
Hence:
\[
LL' = \{wRz'R\mid wR\in L, zR'\in L'\} = \{w\hat{z}RR'\mid wR\in L, zR'\in L', Rz = \hat{z}R\}
\]
and $\lcs(LL')=RR'$.

Consider then
\[
\begin{array}{cl}
  & \lcsext(LL')\\
= & \lcs( \rd(wRzR'\inv{
  RR'})^\iomega \mid wR\in L, zR'\in L', wz\neq \ew)\\
= & \lcsa{ 
\lcs(w^\iomega \mid wR \in L\setminus\{R\})\\
\lcs((w\hat{z})^\iomega \mid wR \in L, zR'\in L'\setminus\{R'\}, Rz = \hat{z}R)
}\\
= & \lcsa{ 
\lcsext(L)\\
\lcs((w\hat{z})^\iomega \mid wR \in L, zR'\in L'\setminus\{R'\}, wz\neq \ew, Rz = \hat{z}R)
}\\
\end{array}
\]

If $E=\top$, i.e.\ $L=\{R\}$, then:
\begin{block}
\[\begin{array}{cl}
  & \lcsext(LL')\\
= & \lcs( \hat{z}^\iomega \mid zR'\in L', z\neq \ew, Rz = \hat{z}R)\\
\rdeq & \lcs( z^\iomega \mid zR'\in L', z\neq \ew)\,\inv{R}\\
\rdeq & E'\,\inv{R}\\
= & \lcs(E,\rd(E'\,\inv{R}))
\end{array}\]
So $\sig(LL')=(RR', \lcs(E,\rd(E'\,\inv{R})))$
\end{block}

So assume that also $\abs{L}\ge 2$.

Fix $uR,vR\in L\setminus\{R\}$ s.t.\
\[
\forall wR\in L\colon
E=\lcsext(L)=\lcs(u^\iomega,v^\iomega)=\lcs(u^\iomega,v^\iomega,w^\iomega)
\]

If $E=\ew$, then:
\begin{block}
$\lcsext(LL')=\ew=\lcs(E,\rd(E'\,\inv{R}))$

$\sig(LL')=(RR',\lcs(E,\rd(E'\,\inv{R})))$
\end{block}

Thus also assume $\lcsext(L)\neq \ew$ and thus also $R\in L$ s.t.:
\[
\begin{array}{lll}
\lcsext(LL')
&=& \lcsa{
  \lcs( (w\hat{z})^\iomega \mid wR \in L, zR'\in L', wz\neq \ew)\\
  \lcsext(L)\\
  \lcsext(\hat{z}^\iomega\mid zR'\in L'\setminus\{R'\}, Rz = \hat{z}R)
}\\
&& \\
&=& \lcsa{
  \lcs( (w\hat{z})^\iomega \mid wR \in L, zR'\in L', wz\neq \ew, Rz = \hat{z}R)\\
  E\\
  \rd(E'\,\inv{R})
}
\end{array}
\]

As shown
\[
\lcs( (w\hat{z})^\iomega, \hat{z}^\iomega, w^\iomega)
=
\lcs( w^\iomega, \hat{z}^\iomega)\sge \lcs(E,\rd(E'\,\inv{R}))
\]

Thus
\[
\lcsext(LL') = \lcs(E,\rd(E'\,\inv{R}))
\]
and again $\sig(LL')=(RR', \lcs(E,\rd(E'\,\inv{R})))$.
\end{prfblk}

\begin{corollary}
$\tseq$ is a congruence relation on the language semiring $\cg{2^{\al^\ast}, \cup, \cdot}$ s.t.\ the quotient w.r.t.\ $\tseq$ is a semiring again with the projection $\sig$ a homomorphism.
\end{corollary}

Every $L\subset\als$ is thus $\tseq$-equivalent to some sublanguage $\tsn(L)\subseteq L$ consisting of at most three words:
\begin{fct}
For $L\subseteq\al^\ast$ we define the following sublanguage $\tsn(L)\subseteq L$ 
\[
\tsn(L):=\begin{cases}
L & \text{ if } \abs{L}\le 2\\
\{R,xR,yR\} & \text{ if } \{R,xR,yR\} \subseteq L\wedge \lcsext(L)=\lcs(x^\iomega,y^\iomega)\\
\{xR,yR\} & \text{ if } R=\lcs(xR,yR)\wedge R\not\in L\wedge \{xR,yR\}\subseteq L
\end{cases}
\quad \text{ with } R:=\lcs(L)
\]
Then independent of the concrete choice of $x$ and $y$ (up to the given side constraints):
\[
L\tseq \tsn(L)
\]
\end{fct}

\renewcommand{\al}{\textsf{A}}

\subsection{Lemma~\ref{lem:char-wf-main-text} in the main work}

\begin{lemma}[Lemma~\ref{lem:char-wf-main-text} in the main work]\label{lem:char-wf}
A context-free grammar $G$ is $\wf$ if and only if $G$ is $\bwf$ with $r_S=\ew$ for $S$ the axiom of $G$.
\end{lemma}
\begin{prfblk}
Let $S$ be the axiom of $G$. 

W.l.o.g.\ $G$ is reduced to the nonterminals which are reachable from $S$ and which are productive.

Assume first that $G$ is well-formed. Then:
 \begin{block}
 For every nonterminal $X$ of $G$ we can define its {\em left-context} $L^l_X = \{ \alpha \in \Br^\ast \mid S\to_G^\ast \alpha X \beta \}$. 
 
 As $L:=L(G)$ is well-formed and every $\alpha\in L^l_X$ is a prefix of some word of $L$, also $L^l_X$ is well-formed; hence $m_X:=\min\{ \hd(L_X^l)\}\ge 0$ is defined.
 Fix any $\lambda_X\in L_X^l$ with $\hd(\lambda_X)=m_X$. 
 
 Then for any $\gamma=\gamma'\gamma''\in L_X$ also $\lambda_X \gamma'$ is a prefix of a word of $L$, thus well-formed, and therefore $\hd(\lambda_X\gamma')\ge 0$. It follows that $d_X=\max\{-\hd(\gamma')\mid \gamma'\gamma''\in L_X\}\le m_X$. 
 
 In particular, there is a $\gamma_X=\gamma'_X\gamma''_X\in L_X$ s.t.\ $-\hd(\gamma'_X)=d_X$; as $L$ is well-formed, $L_X$ has to be weakly well-formed s.t.\ $\gamma'_X\rdeq \inv{r_X} s_X$ and $\lambda_X \rdeq x r_X$. Hence, for every $\gamma\in L_X$ we have $\gamma\rdeq \inv{y} z$ and $\lambda_X\gamma\rdeq x r_X \inv{y} z$ well-formed, i.e.\ $y \sle r_X$ and thus $r_XL_X$ is well-formed. In particular, we have $r_S=\ew$ for the axiom $S$.
 \end{block}
Assume that $G$ is $\bwf$ and thus by definition also nonnegative. Then:
  \begin{block}
  We fix for every $X$ any $r_X\in\al^\ast$ s.t.\ $r_X L_X$ is well-formed; hence, $L_X$ is weakly well-formed and $d_X = \max\{\abs{y}\mid \gamma\in L_X, \rd(\gamma)=\inv{y} z\}$.
  
  Then $\hat{r}_X:=\max^{\sle}\{y\mid \gamma\in L_X, \rd(\gamma)=\inv{y}x\}\sle r_X$ is well defined with $d_X = \abs{r_X}$.
  
  As $G$ is also nonnegative, we have $d_S=0$ resp.\ $r_S=\ew$ and thus $L=L_S$ is well-formed.
  \end{block}
\end{prfblk}

\subsection{Lemma~\ref{lem:r_X-main-text} in the main work}

\begin{fct}
Let $G$ be a context-free grammar over the nonterminals $\vars$. Define $G_p$ by the following rules:
\begin{itemize}
\item
If $X\to_G YZ$, then $X \to_{G_p} YZ$, $X_p \to_{G_p} Y_p$, and $X_p\to_{G_p} Y Z_p$ and $X_p \to_{G_p} YZ$.
\item
If $X\to_G Y$, then $X \to_{G_p} Y$, $X_p\to_{G_p} Y_p$, and $X_p\to_{G_p} Y$.
\item
If $X\to_G \inv{u}v$, then $X\to_{G_p} \inv{u}v$, and $X_p\to_{G_p}\inv{u}v$
\end{itemize}
Then $L_X(G)=L_X(G_p)$ and $L_{X_p}(G_p)\cup\{\ew\}=\pfcl(L_X(G))$.
In particular, we can construct $G_p$ in time polynomial in the size of $G$.
\end{fct}

\begin{lemma}\label{lem:nonneg}
Let $L=L(G)=\pfcl(L(G))$ be a prefix-closed context-free language. 
We can decide in time polynomial in $G$ whether there is a word $\alpha\in L$ s.t.\ $\hd(\alpha) < 0$.
\end{lemma}
\begin{prfblk}
Let $N$ be the number of nontermimals of $G$.
Assume there is a word $\alpha\in L$ with $\hd(\alpha)<0$, then w.l.o.g.\ $\hd(\alpha)=-1$ as $L$ is prefix-closed.
Pick any shortest such $\alpha\in L$ with $\hd(\alpha)=-1$.

If $\alpha$ has a derivation tree of height at most $N$, then we simply apply standard fixed-point/Kleene iteration to the operator $F$ obtained from the rewrite rules of $G$ via the homomorphism $\hd$ over the tropical semiring
\[
F(\vars)_X := \min\{ \hd(Y)+\hd(Z), \hd(Y), \hd(\gamma) \mid X\to_G YZ, X\to_G Y, X\to_G\gamma\}
\]
Then $F^{N}(\infty)_S = \min\{ \hd(\beta)\mid \beta \in L^{\le N}\} \le \hd(\alpha)=-1$ with $L^{\le N}$ all words of $L$ that possess a derivation tree of height at most $N$.

Assume thus that every such $\alpha$ has a derivation tree of height at least $N+1$. Pick a longest path from the root to a leaf in such a derivation tree, and moving bottom-up along this path, pick the first nonterminal $X$ occurring a second time in order to obtain a factorization $\alpha=\beta \rho \gamma \varrho \delta$ s.t.\ $\beta \rho^k \gamma \varrho^k \delta\in L$ for all $k\in \N_0$ and $\rho\varrho \neq\ew$; in particular, note that $\rho\varrho\in L_X^{\le N}$.

Then $-1 = \hd(\alpha)=\hd(\beta\gamma\delta)+\hd(\rho\varrho)$ and $\hd(\beta\gamma\delta)\ge 0$;
otherwise there would be a prefix $\pi$ of $\beta\gamma\delta$ with $\hd(\pi)=-1$
contradicting the minimality of $\alpha$.
Hence, $\hd(\rho\varrho)\le -1$.

Thus, we only need to decide whether there is a pumpable derivation tree $X\to_G^{\le N} \rho X\varrho$ of height at most $N$ s.t.\ $\hd(\rho\varrho)<0$. This can be done by transforming the rewrite rules $X\to YZ$ into weighted edges $X\xrightarrow{F_Y^N(\infty)} Z$ and $X\xrightarrow{F_Z^N(\infty)} Y$, and then check for negative cycles in this graph. (This amounts to take the derivative of $F$ at $F^N(\infty)$.)
\end{prfblk}

\begin{lemma}[Lemma~\ref{lem:r_X-main-text} in the main work]\label{lem:r_X}
Let $L=L(G)$ be $\wf$. Let $X$ be some nonterminal of $G$. Derive from $G$ the CFG $G^p$ s.t.\ $\pfcl(L_X)=L(G_X^p)$. Let $r_X\in \al^\ast$ be the shortest word s.t.\ $r_X L_X$ is \wf.
Then:
\begin{enumerate}
\item $r_X$ is also the shortest word s.t.\ $r_X \pfcl(L_X)$ is \wf.
\item There is a shortest $\alpha\in \pfcl(L_X)$ s.t.\ $\rd(\alpha)=\inv{r_X}$.
\item Every shortest $\alpha\in \pfcl(L_X)$ with $\rd(\alpha)=\inv{r_X}$ has a derivation tree w.r.t.\ $G_X^p$ that does not contain any pumping tree and thus has height bounded by the number of nonterminals of $G_X^p$.
\item 
An \SLP for $r_X$ can be computed in polynomial time.
\end{enumerate}
\end{lemma}
\begin{prfblk}
\begin{enumerate}
\item 
As $L_X \subseteq \pfcl(L_X)$, we only need to show that $u_X \pfcl(L_X)$ is still $\wf$.

For every $\alpha'\in \pfcl(L_X)$ there is some $\alpha\in L_X$ s.t.\ $\alpha= \alpha' \alpha''$ (by definition). As $L$ is \wf, $L_X$ is \wwf, hence $\alpha$ is \wwf, and thus $\alpha'$ and $\alpha''$  are \wwf, too. 
Hence, $\rd(\alpha') = \inv{r} s$, $\rd(\alpha'')=\inv{u}v$ and $\rd(\alpha) =\inv{x}y$ for some $r,s,u,v,x,y\in \al^\ast$ s.t.\ $\inv{r}s\inv{u}v \rdeq \inv{x} y$. If $s=s'u$, then $\inv{r}=\inv{x}$; else $u=u's$ and $\inv{x} = \inv{r}\, \inv{u'}$. Thus $\inv{r} \ple \inv{x} \ple \inv{u_X}$.
\item 
There is some $\beta\in L_X$ s.t.\ $\rd(\beta)= \inv{u_X}y$ for some $y\in \al^\ast$. Then there is some prefix $\beta' \ple \beta$ s.t.\ $\rd(\beta')= \inv{u_X}$. By definition, $\beta'\in \pfcl(L_X)$. Hence, there is also some shortest $\alpha \in \pfcl(L_X)$ s.t.\ $\rd(\alpha)=\inv{u_X}$.
\item 
Let $\alpha\in\pfcl(L_X) = L(G_X^p)$ be a shortest word s.t.\ $\rd(\alpha)=\inv{u_X}$. Assume that there is some factorization $\alpha = \beta \rho \gamma \varrho \delta$ s.t.\ $\beta \rho^k \gamma \varrho^k \delta \in L(G_X^p)$ for all $k\in \N_0$. 

Consider $k=0$ and let $\rd(\beta\gamma\delta)=\inv{r}s$. If $r=u_X$, then there would be a prefix of $\beta\gamma\delta$ that would reduce to $\inv{u_X}$ contradicting our assumption that $\alpha$ is a shortest such word. Hence, $\inv{r} \plt \inv{u_X}$. Thus $\hd(\beta\gamma\delta) \ge -\abs{r} > -\abs{u_X}$.

Note that $-\abs{u_X}=\hd(\alpha)=\hd(\beta\gamma\delta)+\hd(\rho\varrho)$.
Thus, $\hd(\rho\varrho) < 0$; a contradiction to the \wf ness of $u_X \pfcl(L_X)$.

\item 
Split every nonterminal $Y$ of $G^p_X$ into $N+1$ copies $Y_0,\ldots,Y_N$, split every rule $Y\to UV$ into the rules $Y_{i+1} \to U_i V_i\mid Y_i$, and derive from every rule $Y\to \gamma$ the rule $Y_0 \to \gamma$. In other words, unfold $G^p_X$ into an acyclic grammar that generates exactly all derivation trees of height at most $N$. 

We compute inductively for every nonterminal a pair of \SLP s representing a {\wwf} word $\inv{u}v$ as follows:

For every rule $Y_0\to \gamma\in \Br$, we choose either $u = \gamma$, $v=\ew$ or
$u=\ew$, $v=\gamma$ such that $\inv{u}v = \gamma$.

For every rule $Y_{i+1} \to U_i V_i$, we have inductively computed \SLP s for $U_i$ and $V_i$ representing words $\inv{r}s$ and $\inv{u}v$, respectively.
Then we can compute \SLP s representing the reduct $\rd(\inv{r}s\inv{u}v)$: either $s=s'u$ (i.e.\ $\abs{s}\ge \abs{u}$) or $u=u's$ (i.e.\ $\abs{s}\le \abs{u}$), i.e.\ we simply have to restrict and then concatenate the respective \SLP s.
For the rule $Y_{i+1}\to Y_i$ there is nothing to do.

We are thus left for $Y_{i+1}$ with a family of \SLP s representing words $\inv{u_i}v_i$: w.l.o.g.\ assume $\abs{u_0}\ge \abs{u_i}$ for all $i$; as for every derivation $X_{N} \to^\ast \alpha Y_{i+1}\gamma$ we need to have that $\alpha \inv{u_i} v_i \gamma$ is \wwf\ for every $i$, we also have that $\alpha \inv{u_0} u_0 \inv{u_i} v_i \gamma$ is \wwf\ for every $i$. We thus may normalize all \SLP\ pairs by means of $\inv{u_i} v_i \mapsto \inv{u_0} \rd(u_0\inv{u_i}) v_i$. As we want to maximize the descent, we then assign to $Y_{i+1}$ the pair of \SLP s encoding $\inv{u_0}$ and the shortest of all $\rd(u_0\inv{u_i})v_i$.
This amounts to a constant amount of \SLP\ operations per rule of the unfolded grammar.

\end{enumerate}
\end{prfblk}

%% file: proofs-rlcs-linear.tex
\subsection{Reduced LCS of simple linear $\wf$ languages}\label{app:slg}

The following \cref{lem:negneg,lem:negpos,lem:pospos} state the central combinatorial results underlying the proof of Lemma~\ref{thm:converge-wf-main-text}.
They are concerned with the reduced $\lcs$ of \emph{simple} linear grammars of the form
\[
S\to \alpha X \beta \qquad X \to \sigma_1 X \tau_1 \mid \ldots \mid \sigma_k X \tau_k \mid \gamma \qquad (\alpha,\beta,\sigma_i,\tau_i,\gamma\in\Br^\ast)
\]
which arise from the factorization of derivation trees:

Given (i) a derivation tree of a context-free grammar $G$ that yields the word $\kappa$, 
(ii) a path within this tree, 
and (iii) a specific nonterminal $X$ of $G$,
we may factorize $\kappa$ into the product of {\em (word) contexts} (finite words with a ``hole'' which represent a pumping tree w.r.t.\ $G$) $(\alpha,\beta)$, $(\sigma_1,\tau_1)$, $\ldots$, $(\sigma_k,\tau_k)$ and a single word $\gamma$ s.t.\ $S\to_G^\ast\alpha X \beta$, $X\to_G^\ast\sigma_i X\tau_i$, and $X\to_G^\ast \gamma$. 
We denote such factorizations by simply writing
$\kappa = (\alpha,\beta)(\sigma_1,\tau_1)\ldots (\sigma_k,\tau_k)\gamma$.
Concatenation of contexts with contexts resp.\ words is thus defined by means of substituting the right operand into the ``hole'' of the context, i.e.\
$(\sigma,\tau) (\mu,\nu)=(\sigma\mu,\nu\tau)$ and $(\sigma,\tau)\gamma=\sigma\tau\gamma$.

Such a factorization then induces the \emph{simple linear language}
\[
(\alpha,\beta)[(\sigma_1,\tau_1)+\ldots+(\sigma_k,\tau_k)]^\ast \gamma := \{ \alpha \sigma_{i_1}\ldots \sigma_{i_l}\gamma \tau_{i_l}\ldots \tau_{i_1}\gamma \mid i_1\ldots i_l\in\{1,\ldots,k\}^\ast\}
\]
which is generated by the \emph{simple linear grammar}
\[
S\to \alpha X \beta \qquad X \to \sigma_1 X \tau_1 \mid \ldots \mid \sigma_k X \tau_k \mid \gamma
\]
and is thus always a sublanguage of $L(G)$.
Assuming that $G$ is well-formed, we show in the proof of Lemma~\ref{thm:converge-wf-main-text} 
that we can rewrite each rule so that the simple linear grammar takes the form
\[
S\to u X \qquad X \to s_1 X \tau_1 \mid \ldots \mid s_k X \tau_k \mid w \quad (u,v,w,s_i,r_i,t_i\in\al^\ast, \tau_i = \inv{r_i}t_ir_i \vee \tau_i=\inv{r_i}\inv{t_i}r_i)
\]
where both grammars generated the same language after reduction,
and there is one-to-one correspondence of the rewrite rules s.t.\ the derivations of both grammars are in bijection.
For the proof of Lemma~\ref{thm:converge-wf-main-text} it suffices to consider where $k=2$, i.e.\ derivation tree has been factorized into two pumping trees.

The central observation in \cref{lem:negneg,lem:negpos} is that, if at least one of the contexts $(s_i,\tau_i)$ is \emph{negative}, i.e.\ $\tau_i\rdeq \inv{r_i} \inv{t_i} r_i$ with $t_i\neq \ew$, then the simple linear well-formed $L$ can be normalized to a regular language over $\al$ whose $\lcs$ and $\lcsext$ are already determined by $(u,v)\ew$ and $(u,v)(s_i,\tau_i)\ew$. See also Example~\ref{ex:neg-app}.

\begin{example}\label{ex:neg-app}
Consider the linear language $L'$ given by the rules $S\to u X$ and $X \to s X \inv{r}\inv{t}r \mid \ew$
where we assume that the language is $\wf$ with $t\neq \ew$ and, for the sake of this example, also $\abs{tr}> \abs{s}$.
As $us^k\inv{r}\inv{t}^kr$ is $\wf$ for all $k\in\N$, 
we have $(s^\iomega)^R = (t^\iomega r)^R$ i.e.\
there is conjugate $p$ of the primitive root $q$ of $t$ s.t.\ (i) $qr=rp$, (ii) $s=p^m$, (iii) $t=q^n$, and (iv) $m\ge n$ for suitable $m,n\in\N_0$.
Property (iv) has to hold as otherwise we could generate a negative word. 
Further as $\abs{tr} > \abs{s}$ we have $tr\inv{s} \rdeq r\inv{p}^{m-n}$ s.t.\ $r=r'p^{m-n}$, $qr'=r'p$, and $u=u'r'$ as $us\inv{r}\inv{t}r$ is $\wf$.
We thus may replace $X\to sX\inv{r}tr$ with $X \to p^{m-n} X$ as
\[
\begin{array}{lcl}
u s^{k+1} \inv{r} \inv{t}^{k+1} r 
& \rdeq &
u' r' (p^m)^{k+1} \inv{p}^{m-n} \inv{r'} (\inv{q}^n)^{k+1} r' p^{m-n}\\
&\rdeq &
u' r' (p^m)^{k} p^n \inv{r'} (\inv{q}^n)^{k+1} r' p^{m-n}\\
&\rdeq &
u' (q^m)^{k} q^n (\inv{q}^n)^{k+1} q^{m-n} r'
\rdeq 
u' (q^{m-n})^{k+1}  r'
\rdeq 
u (p^{m-n})^{k+1}
\end{array}
\]
s.t.\ we obtain a regular language $L'\subseteq \als$ whose derivations are in bijection with those of $L$.
Now, $\sig(L)$ is already determined by $u$ and $up$ which in turn implies that $u$ and $us\inv{r}tr$ determine $\rsig(L)$. In case of multiple contexts $(s_j,\tau_j)$ the existence of one context of the form $(s_i,\inv{r_i}\inv{t_i}r_i)$ enforces that all contexts have to be compatible with the primitive root of $t_i$ which subsequently allows us to replace every rule $X\to s_i X \tau_i$ by a $\rdeq$-equivalent rule $X\to p^{k_i} X$ over $\al$.\qed
\end{example}

On the other hand, if both contexts $(s_i,\tau_i)$ are \emph{nonnegative}, i.e.\ $\tau_i= \inv{r_i} t_i r_i$ for $i=1,2$, then Lemma~\ref{lem:pospos} shows that the $\lcs$ and $\lcsext$ of the simple linear well-formed language $L$ is already determined by $(u,v)\ew$ and either some word $(u,v)(s_i,\tau_i)\ew$ or some word $(u,v)(s_i,\tau_i)(s_j,\tau_j)\ew$ for some $i\in\{1,2\}$ with the important point that $j$ can be chosen arbitrarily from $\{1,2\}$ --- this is central to the proof of Lemma~\ref{thm:converge-wf-main-text}. See also Example~\ref{ex:pospos-app}.

\begin{example}\label{ex:pospos-app}
Consider the well-formed language
\[
L= (a{bbaba}, \ew)[ (ba, \inv{aba} b{aba})+(a, bbaba)]^\ast \ew
\]
We have $R:=\rlcs(L) = bbaba$ as we will see in the following. To this end, set
\[
u:=a\quad r := aba\quad s_1 := ba \quad s_2 :=a \quad t_1 := b 
\]
so that
\[
L = (uR,\ew)[ (s_1, \inv{r} t_1{r})+(s_2, R)]^\ast \ew
\]
Note that we have
\begin{itemize}
\item ${r}={aba}={r'}s_1$ for ${r'}:=a$
\item ${r'}s_1 = a\,ba=ab\, a = \hat{s}_1 {r'}$ for $\hat{s}_1:=ab$
\item ${r'}s_2 = a\, a = \hat{s}_2\, {r'}$ for $\hat{s}_2:=s_2=a$
\item ${R} = t_1^2 {r}=\rlcs(L)$
\end{itemize}
In particular, note that $aba = r \not\sle rs_2 = abaa$, i.e.\ there is no conjugate that allows us to move $r$ ``through'' $s_2$;
we only find a conjugate $\hat{s}_2$ with $r's_2=\hat{s}_2r'$ (because of well-formedness).
By splitting the language depending on the innermost context, we can use these conjugates to rewrite the contexts so that all factors are words over $\al$:
\[
\begin{array}{cl@{\quad}l}
 &(u{R},\ew)[(s_1,\inv{r}t_1{r})+(s_2,R)]^\ast \ew \\
=& (ut_1^2{r},\ew)[(s_1,\inv{r}t_1{r})+(s_2,t_1^2 r)]^\ast \ew & (\text{split w.r.t.\ innermost context})\\
=& (ut_1^2{r},\ew) \ew\\
+& (ut_1^2{r},\ew)[(s_1,\inv{{r}}t_1{r})+(s_2,t_1^2{r})]^\ast [s_1\inv{r}t_1{r}+s_2t_1^2{r}] & (\text{move $t_1r$ to the right})\\
=& (ut_1,t_1{r}) \ew\\
+& (ut_1^2{r},t_1{r})[(s_1,t_1{r}\inv{{r}})+(s_2,t_1{r}t_1)]^\ast [s_1\inv{{r}}+s_2t_1] & (\text{cancel $r\inv{r}$})\\
=& (ut_1,t_1{r}) \ew\\
+& (ut_1^2{r},t_1{r})[(s_1,t_1)+(s_2,t_1{r}t_1)]^\ast [s_1\inv{{r}}+s_2t_1] & (\text{use $r=r's_1$})\\
=& (ut_1,t_1{r}) \ew\\
+& (ut_1^2{r'}s_1,t_1{r})[(s_1,t_1)+(s_2,t_1{r}t_1)]^\ast [s_1\inv{{r}}+s_2t_1] & (\text{move $r'$ to the middle})\\
=& (ut_1,t_1{r}) \ew\\
+& (ut_1^2\hat{s}_1,t_1{r})[(\hat{s}_1,t_1)+(\hat{s}_2,t_1{r}t_1)]^\ast [r's_1\inv{{r}}+r's_2t_1] & (\text{use $r_1=r_1's_1$ and cancel})\\
\end{array}
\]
Note that the derivation w.r.t.\ linear grammars underlying both languages are in bijection.
We split the rewritten language one last time, this time w.r.t.\ the outermost context, which leads us to:
\[
\begin{array}{lcl}
L & \rdeq & (ut_1^2{r},\ew) \ew\\
& + & (ut_1^2\hat{s}_1,t_1{r})\ew\\
& + & (ut_1^2\hat{s}_1,t_1{r}){r'}s_2t_1\\
& + & (ut_1^2\hat{s}_1,t_1{r})(\hat{s}_1,t_1)[(\hat{s}_2,t_1{r}\,t_1)]^\ast [\ew+{r'}s_2t_1]\\
& + & (ut_1^2\hat{s}_1,t_1{r})(\hat{s}_2,t_1{r}\,t_1)[(\hat{s}_1,t_1)+(\hat{s}_2,t_1{r}\,t_1)]^\ast [\ew+{r'}s_2t_1]\\
\end{array}
\]
Substituting the actual values yields
\[
\begin{array}{lcr}
L & \rdeq & {a}\,{bbaba}\\
& + & (abb{a}{b},{baba})\ew\\
& + & (abb{a}{b},{baba})a{a}{b}\\
& + & (abbg{a}{b},{baba})({a}{b},{b})[({a}{b},{b})+(a,bab{a}{b})]^\ast [\ew+a{a}{b}]\\
& + & (abb{a}{b},{baba})(a,bab{a}{b})[({a}{b},{b})+(a,bab{a}{b})]^\ast [\ew+a{a}{b}]\\
& = & {a}\,{bbaba}\\
& + & \ldots {a}\,{b}{baba}\\
& + & \ldots {a}\,{b}{baba}\\
& + & \ldots {a}\,{b}\,{b}{baba}\\
& + & \ldots {a}\,{b}{baba}\\
\end{array}
\]
i.e.\ only the words included in
\[
\begin{array}{cl}
& (a{bbaba}, \ew)(ba, \inv{aba} b{aba})[ (ba, \inv{aba} b{aba})+(a, bbaba)]^+ \ew\\
\rdeq & (abbg{a}{b},{baba})({a}{b},{b})[({a}{b},{b})+(a,bab{a}{b})]^\ast [\ew+a{a}{b}]\\
= &\ldots {a}\,{b}\,{b}{baba}\\
\end{array}
\]
define $R=\rlcs(L)$. In particular, we need in this case to use at least two contexts where the outermost context has to be $(a{bbaba}, \ew)(ba, \inv{aba} b{aba})$, while the remaining contexts can be chosen arbitrarily.
\end{example}

Before proving \cref{lem:negneg,lem:negpos,lem:pospos} 
we need the additional \cref{lem:ustst,lem:uststw} for the case without closing brackets (i.e.\ $r_i=\ew$).
Both lemmas are stronger versions of the analogous results for the $\lcp$ 
as presented in~\cite{DBLP:conf/stacs/LuttenbergerPS18}.
Most importantly, both lemmas now state that, if e.g.\ $(u,\ew)(s_1,t_1)^2w = us_1^2 w t_1^2$ is a witness of the $\lcs$ w.r.t.\ $(u,\ew)w=uw$, then also $(u,\ew)(s_1,t_1)(s_2,t_2)w= us_1s_2wt_2t_1$ is a witness w.r.t.\ $uw$; i.e.\ only the outer context resp.\ pumping tree matters in the end.

\begin{lemma}\label{lem:ustst}
Let $L=(u,\ew)[ (s_1,t_1) + (s_2,t_2)]^\ast \ew$ be $\wf$ with $s_1t_1\neq \ew \neq s_2t_2$.

Then $\lcs(L) =\lcsa{u\\ us_1t_1\\ us_1s_1t_1t_1\\ us_2t_2\\ us_2s_2t_2t_2}=\lcsa{u\\ us_1t_1\\ us_1s_2t_2t_1\\ us_2t_2\\ us_2s_1t_1t_2}$

If $s_i= \ew$, then $u s_is_it_it_i$ is not required.
\end{lemma}
\begin{proof}
Let $R := \lcs(L)$ and $u = u' R$. It exists $\hat{t_i}$ such that
\[R t_i = \hat{t}_i R\]
If $t_i=\ew$, then set $\hat{t}_i=\ew$; otherwise we have $R \slt t_i^\iomega$ as $R\sle us_i^k t_i^k$ for all $k\ge 0$.
\begin{claim}
If $R\slt t_1=t_1'R\wedge R\slt t_2=t_2'R$ then $\lcs(L)=\lcsa{u\\ us_1t_1\\ us_2t_2} = \lcsa{u'\\ us_1t_1'\\ us_2t_2'} R$.
\end{claim}
\begin{proof}
As $t_1'\neq \ew\neq t_2'$ we get $L=(u'R,\ew)[(s_1,t_1'R)+(s_2,t_2'R)]^\ast \ew$.
Therefore $\lcs(u',t_1')= \ew \vee \lcs(u',t_2')= \ew$ and the claim follows.
\end{proof}
Thus, assume w.l.o.g.\ that $t_1 \sle R=\dot{R} t_1$ from here on.
Then, 
\[R=\dot{R} t_1 = \hat{t}_1 \dot{R}\]
If $t_1= \ew$, we set $\hat{t}_1= \ew$ and $R=\dot{R}$;
otherwise $\dot{R} t_1 t_1 = R t_1 = \hat{t}_1 R = \hat{t}_1 \dot{R} t_1$ (cancel $t_1$ from the left).
Additionally, there exists $\dot{s}_1$ such that
\[\dot{R} s_1 = \dot{s}_1 \dot{R}\]
If $s_1= \ew$, we set $\dot{s}_1= \ew$; otherwise
$\dot{R} \slt \dot{R} s_1$ as $\dot{R} t_1 = R \sle u s_1 t_1 = u' \dot{R} t_1 s_1 t_1 = u' \hat{t}_1 \dot{R} s_1 t_1$.
\begin{claim}
If $R\slt t_2=t_2'R = t_2'\hat{t}_1 \dot{R}$ ($t_2'\neq \ew$) then the lemma follows.
\end{claim}
\begin{proof}
We have
\begin{align*}
(u,\ew)(s_1,t_1)^+ \ew
&=(u'\hat{t}_1 \dot{R},\ew)(s_1,t_1)^+ \ew\\
&=(u'\hat{t}_1 \dot{s}_1,R)(\dot{s}_1,\hat{t}_1)^\ast \ew\\
(u,\ew)(s_1,t_1)^\ast (s_2,t_2) [(s_1,t_1)+(s_2,t_2)]^\ast \ew
&=(u,\ew)(s_1,t_1)^\ast (s_2,t_2'R) [(s_1,t_1)+(s_2,t_2)]^\ast \ew\\
&=(u,R)(s_1,\hat{t}_1)^\ast (s_2,t_2') [(s_1,t_1)+(s_2,t_2)]^\ast \ew\\
\end{align*}
Therefore $\ew\overset{!}= \lcs(u', \hat{t}_1\dot{s}_1, \dot{s}_1\hat{t}_1^+, t'_2\hat{t}_1^*)$.
If $\lcs(u', \hat{t}_1\dot{s}_1) = \ew$ then $us_1t_1$ is a witness.
If $\lcs(u', \dot{s}_1\hat{t}_1^+) = \ew$ and $\hat{t}_1 \neq \ew$ then $\lcs(u', \dot{s}_1\hat{t}_1) = \ew$ and $us_1s_1t_1t_1$ is a witness
and $us_2s_1t_1t_2$. Note that if $\hat{t}_1 = \ew$ then $\dot{s}_1 \neq \ew$ and we are in the first case where $us_1t_1$ is a witness.
If $\lcs(u', t'_2) = \ew$ then $us_2t_2$ is a witness.
\end{proof}

We therefore consider the case that $t_2\sle R = \ddot{R} t_2$ from here on.
Then 
\[R=\ddot{R} t_2 = \hat{t}_2 \ddot{R} \text{ and } u = u'R = u'\ddot{R}t_2 = u' \hat{t}_2\ddot{R}\]
as $\ddot{R} t_2 t_2 = R t_2 = \hat{t}_2 R = \hat{t}_2 \ddot{R} t_2$ (cancel $t_2$ from the left).
W.l.o.g.\ we assume that $\abs{t_1} \le \abs{t_2}$.
\begin{claim}
We find the following conjugates
\begin{enumerate}
\item $\exists \ddot{s}_2\colon \ddot{R} s_2 = \ddot{s}_2 \ddot{R}$
\item $\exists z\colon t_2 = zt_1\wedge \dot{R}=\ddot{R}z$
\item $\exists \ddot{z}\colon \ddot{R} z = \ddot{z} \ddot{R}$
\item $\exists \ddot{s}_1\colon \ddot{R} s_1 = \ddot{s}_1 \ddot{R}$
\item $\exists \ddot{t}_1\colon \ddot{R} t_1 = \ddot{t}_1 \ddot{R}$
\item $\hat{t}_2= \ddot{z} \ddot{t}_1 = \hat{t}_1 \ddot{z}$
\end{enumerate}
\end{claim}
\begin{proof}
\begin{enumerate}
\item We have $\ddot{R} t_2 = R \sle u s_2 t_2 = u' \ddot{R} t_2 s_2 t_2 = u' \hat{t}_2 \ddot{R} s_2 t_2$
and therefore $\ddot{R} \sle \ddot{R} s_2$.
\item We have $R=\dot{R} t_1 = \ddot{R} t_2$ and $\abs{t_1}\le \abs{t_2}$
and therefore
$\ddot{R} z t_1 = \dot{R} t_1$.
\item We have $\hat{t}_2 \ddot{R} = \ddot{R}t_2=R=\dot{R} t_1 = \hat{t}_1 \dot{R} = \hat{t}_1 \ddot{R} z$
and therefore
$\ddot{R} \sle \ddot{R} z$.
\item If $t_2= \ew$, then $\ddot{R}=Rt_2=R = \dot{R} = \ddot{R}$ and $\ddot{s}_1 = \dot{s}_1 = s$.
Otherwise, we have
$\ddot{R} t_1 \sle \ddot{z} \ddot{R} t_1 = \ddot{R} z t_1 = \ddot{R} t_2 = R \sle u s_1 t_1 = u' \hat{t}_2 \ddot{R} s_1 t_1$
and therefore
$\ddot{R} \sle \ddot{R} s_1$.
\item We have $\hat{t}_2 \ddot{R}=\ddot{R} t_2 = R \sle \hat{t}_1 R = R t_1 = \ddot{R} t_2 t_1 = \hat{t}_2 \ddot{R} t_1$
and therefore
$\ddot{R} \sle \ddot{R} t_1$.
\item We have $\hat{t}_2 \ddot{R} = R = \dot{R} t_1 = \hat{t}_1 \dot{R} = \hat{t}_1 \ddot{R} z = \hat{t}_1 \ddot{z} \ddot{R}$
and thus $\hat{t}_2 = \hat{t}_1\ddot{z}$.
Additionally, $\hat{t}_2 \ddot{R} = R = \ddot{R} t_2 = \ddot{R} z t_1 = \ddot{z} \ddot{R} t_1 = \ddot{z} \ddot{t}_1 \ddot{R}$
holds and thus $\hat{t}_2 = \ddot{z}\ddot{t}_1$.
\end{enumerate}
\end{proof}

Using these conjugates we obtain:

\begin{tabular}{llll}
$(u,\ew)(s_1,t_1)^+ \ew$
&$= (u'\hat{t}_1 \dot{R}, \ew)(s_1, t_1)^+ \ew$\\
&$= (u'\hat{t}_1\dot{s}_1, R)(\dot{s}_1,\hat{t}_1)^\ast \ew$\\
&$= u'\hat{t}_1\dot{s}_1R$ && \textcolor{darkgray}{$us_1t_1$}\\
&$+ (u'\hat{t}_1\dot{s}_1\dot{s}_1, \hat{t}_1R)(\dot{s}_1,\hat{t}_1)^\ast \ew$&&\textcolor{darkgray}{$(u,\ew)(s_1,t_1)^{\geq2}\ew$}\\
\end{tabular}

\begin{tabular}{lll@{}l}
&$(u,\ew)(s_1,t_1)^\ast (s_2,t_2)[(s_1,t_1)+(s_2,t_2)]^\ast \ew$ \\
=& $(u'\hat{t}_2 \ddot{R}, \ew)(s_1,t_1)^\ast (s_2,t_2)[(s_1,t_1)+(s_2,t_2)]^\ast \ew$ & \rdelim\}{8}{*}&
\multirow{8}{5cm}{\textcolor{darkgray}{``extracting'' the $\lcs R$ by substituting the corresponding conjugates of $u, s_i, t_i$}}\\
=& $(u'\hat{t}_2, \ddot{R})(\ddot{s}_1,\ddot{t}_1)^\ast (\ddot{s}_2,\hat{t}_2)[(\ddot{s}_1,\ddot{t}_1)+(\ddot{s}_2,\hat{t}_2)]^\ast \ew$&\\
=& $(u'\hat{t}_2, \ddot{R})(\ddot{s}_1,\ddot{t}_1)^\ast (\ddot{s}_2,\hat{t}_1 \ddot{z})[(\ddot{s}_1,\ddot{t}_1)+(\ddot{s}_2,\hat{t}_2)]^\ast \ew$&\\
=& $(u'\hat{t}_2, \ddot{z}\ddot{R})(\ddot{s}_1,\hat{t}_1)^\ast (\ddot{s}_2,\hat{t}_1 )[(\ddot{s}_1,\ddot{t}_1)+(\ddot{s}_2,\hat{t}_2)]^\ast \ew$&\\
=& $(u'\hat{t}_2, \dot{R})(\ddot{s}_1,\hat{t}_1)^\ast (\ddot{s}_2,\hat{t}_1 )[(\ddot{s}_1,\ddot{t}_1)+(\ddot{s}_2,\hat{t}_2)]^\ast \ew$&\\
=& $(u'\hat{t}_2, \hat{t}_1\dot{R})(\ddot{s}_1,\hat{t}_1)^\ast (\ddot{s}_2,\ew )[(\ddot{s}_1,\ddot{t}_1)+(\ddot{s}_2,\hat{t}_2)]^\ast \ew$&\\
=& $(u'\hat{t}_2, R)(\ddot{s}_1,\hat{t}_1)^\ast (\ddot{s}_2,\ew )[(\ddot{s}_1,\ddot{t}_1)+(\ddot{s}_2,\hat{t}_2)]^\ast \ew$&\\
=& $(u'\ddot{z}\ddot{t}_1, R)(\ddot{s}_1,\hat{t}_1)^\ast (\ddot{s}_2,\ew )[(\ddot{s}_1,\ddot{t}_1)+(\ddot{s}_2,\ddot{z}\ddot{t}_1)]^\ast \ew$&\\
=& $ u'\ddot{z}\ddot{t}_1\ddot{s}_2R$&& \textcolor{darkgray}{$us_2t_2$}\\
& $+ (u'\ddot{z}\ddot{t}_1\ddot{s}_1, \hat{t}_1R)(\ddot{s}_1,\hat{t}_1)^\ast (\ddot{s}_2,\ew )[(\ddot{s}_1,\ddot{t}_1)+(\ddot{s}_2,\ddot{z}\ddot{t}_1)]^\ast \ew$&&\textcolor{darkgray}{$(u,\ew)(s_1,t_1)^+(s_2,t_2)[(s_1,t_1)+(s_2,t_2)]^*\ew$}\\
& $+(u'\ddot{z}\ddot{t}_1, \ddot{t}_1R)(\ddot{s}_2\ddot{s}_1,\ew )[(\ddot{s}_1,\ddot{t}_1)+(\ddot{s}_2,\ddot{z}\ddot{t}_1)]^\ast \ew$
&&\textcolor{darkgray}{$(u,\ew) (s_2,t_2)(s_1,t_1)[(s_1,t_1)+(s_2,t_2)^*\ew$}\\
&$+(u'\ddot{z}\ddot{t}_1, \ddot{z}\ddot{t}_1R)(\ddot{s}_2\ddot{s}_2,\ew )[(\ddot{s}_1,\ddot{t}_1)+(\ddot{s}_2,\ddot{z}\ddot{t}_1)]^\ast \ew$
&&\textcolor{darkgray}{$(u,\ew) (s_2,t_2)^{\geq2}[(s_1,t_1)+(s_2,t_2)]^*\ew$}\\
\end{tabular}\\[2mm]
If $us_1t_1$ or $us_2t_2$ is a witness then the claim of the lemma follows.
Thus, assume that neither $us_1t_1$ nor $us_2t_2$ is a witness w.r.t. $u$, i.e.
\[\lcs(u',\hat{t}_1\dot{s}_1, \ddot{z}\ddot{t}_1\ddot{s}_2)\neq \ew\]

W.l.o.g.\ $t_2\neq \ew$ and thus also $\hat{t}_2\neq \ew$ as otherwise
$t_1=\ew$ as $0=\abs{t_2}\ge \abs{t_1}$ s.t.\ $R=\dot{R}=\ddot{R}$ and $\dot{s}_1=\ddot{s}_1$.
Then $L=u(s_1+s_2)^\ast=u'(\ddot{s}_1+\ddot{s}_2)^\ast R$
and thus $\lcs(u',\ddot{s}_1,\ddot{s}_2)= \ew$. Therefore $us_1t_1 = us_1$ or $us_2t_2 = us_2$ would be a witness.
\begin{claim}
If $t_1 = \ew$ then the lemma follows.
\end{claim}
\begin{proof}
We have
$\hat{t}_1=\ddot{t}_1= \ew$ and $\dot{R}=R$ and $\dot{s}_1\neq \ew \neq \ddot{s}_1$ and $\hat{t}_2=\ddot{z}\ddot{t}_1=\ddot{z}\neq\ew$ s.t.
\begin{align*}
(u,\ew)(s_1,t_1)^+ \ew
&= u' \dot{s}_1^+ R\\
(u,\ew)(s_1,t_1)^\ast (s_2,t_2)[(s_1,t_1)+(s_2,t_2)]^\ast \ew
&=(u'\ddot{z}, R)(\ddot{s}_1,\ew)^\ast (\ddot{s}_2,\ew )[(\ddot{s}_1,\ew)+(\ddot{s}_2,\ddot{z})]^\ast \ew\\
\end{align*}
Therefore $\ew \overset{!}= \lcs(u', \dot{s}_1, \ddot{z}\ddot{s}_2, \ddot{s_1}, \ddot{z})$.
If $\lcs(u', \dot{s}_1) = \ew$ then $us_1t_1$ would be a witness.
If $\lcs(u', \ddot{z}\ddot{s}_2) = \ew$ then $us_2t_2$ would be a witness.
We therefore need to consider the cases $\lcs(u', \ddot{s}_1) = \ew$ and $\lcs(u', \ddot{z}) = \ew$.
In fact, $\lcs(\ddot{s}_1, \ddot{z}) \neq \ew$ holds as
$\ddot{z} \ddot{R}=\ddot{z}\ddot{t}_1 \ddot{R} = \hat{t}_2 \ddot{R}=R\sle us_1^kt_1^k = u'Rs_1^k =u'\ddot{z}\ddot{R}s_1^k = u'\ddot{z}\ddot{s}_1^k\ddot{R}$ for all $k$ and therefore $\ddot{z} \slt \ddot{s}_1^\iomega$.
Thus $\lcs(u', \ddot{s}_1) = \ew$ if and only if $\lcs(u', \ddot{z}) = \ew$ and therefore if
$us_1s_2t_2t_1$ is a witness if and only if $us_2s_2t_2t_2$ is a witness.
\end{proof}
Thus, assume that $t_1\neq \ew$.
Then $\hat{t}_1 \neq \ew \neq \ddot{t}_1$ and therefore $\ew\overset{!}=\lcs(u',\hat{t}_1,\ddot{t}_1)$.
We obtain
\[
\begin{array}{lllll}
\lcs(L)&
=&\lcsa{u\\ us_1s_1t_1t_1\\ us_2s_2t_2t_2}
&=&\lcsa{u'R\\ u'\hat{t}_1\dot{s}_1\dot{s}_1\hat{t}_1R\\ u'\ddot{z}\ddot{t}_1\ddot{s}_2\ddot{s}_2\ddot{z}\ddot{t}_1R}\\
&=&\lcsa{u\\ us_1s_2t_2t_1\\ us_2s_1t_1t_2}
&=&\lcsa{u'R\\ u'\ddot{z}\ddot{t}_1\ddot{s}_1\ddot{s}_2\hat{t}_1R\\ u'\ddot{z}\ddot{t}_1\ddot{s}_2\ddot{s}_1\ddot{t}_1R}
\end{array}
\]
\end{proof}

\begin{lemma}\label{lem:uststw}
Let $L=(u,\ew)[ (s_1,t_1) + (s_2,t_2)]^\ast w$ be $\wf$.

Then $\lcs(L) =\lcsa{uw\\ us_1wt_1\\ us_1s_1wt_1t_1\\ us_2wt_2\\ us_2s_2wt_2t_2}=\lcsa{uw\\ us_1wt_1\\ us_1s_2wt_2t_1\\ us_2wt_2\\ us_2s_1wt_1t_2}$

If $s_i=\ew$, then $us_is_iwt_it_i$ is not required.
\end{lemma}
\begin{prfblk}
$R := \lcs(L)$.

W.l.o.g. $s_it_i\neq \ew$.

Case $R\slt w=w'R$:
\begin{block}
$w'\neq \ew$

$\exists\hat{t}_i\colon R t_i = \hat{t}_i R$
\begin{block}
$R\sle us_iwt_i =us_i w' R t_i$

$R\sle Rt_i$
\end{block}
$L=(u,1)[(s_1,t_1)+(s_2,t_2)]^\ast w' R =(u,R)[(s_1,\hat{t}_1)+(s_2,\hat{t}_2)]^\ast w'$
$\ew\overset{!}= \lcs(w',w'\hat{t}_1,w'\hat{t}_2)$

$\lcs(w',w'\hat{t}_1)=\ew\vee \lcs(w',w'\hat{t}_2)=\ew$

$\lcs(w',\hat{t}_1)=\ew\vee \lcs(w',\hat{t}_2)=\ew$

$\lcs(L)=\lcsa{uw\\ us_1wt_1\\ us_2wt_2}$
\end{block}

Case $w\sle R=R'w$:
\begin{block}
$\exists \hat{t}_i\colon w t_i = \hat{t}_i w$
\begin{block}
$R=R' w\sle u s_i w t_i$

$w \sle w t_i$
\end{block}
$L= (u'R',w)[(s_1, \hat{t}_1)+ (s_2,\hat{t}_2)]^\ast 1$

Apply Lemma~\ref{lem:ustst} to $L'=(u'R',1)[(s_1,\hat{t}_1)+(s_2,\hat{t}_2)]^\ast 1$
\end{block}
\end{prfblk}
\newpage
\begin{lemma}\label{lem:rst}
Let $L = (r,\ew)(s,t)^+ \ew$ be $\wf$ with $r\sle\lcs(L)\wedge t\neq \ew\wedge rt = \hat{t}r\wedge r\not\sle rs$
Then: 
\[\exists x,y,\tilde{s}\forall k\ge 0\colon (r,\ew)(s,t)^{k+1} = (x,r)(\tilde{s},\hat{t})^k y\wedge \lcs(L)=\lcs(xy,\hat{t})r =\lcs(st,ts)t \slt \hat{t}r\]
\end{lemma}
\begin{prfblk}
$r\neq \ew$ as $r\not\sle rs$

If $r\not\sle rs$, then $st\neq ts$ as otherwise
\begin{block}
$\exists p, m,n\colon s=p^m\wedge t=p^n$

$\exists q\colon rp = qr\wedge \hat{t}=q^n$ as
\begin{block}
$rp^n=rt=\hat{t}r$ s.t. $\lcs(r,rp)=\lcs(r,rp^{k+1})=\lcs(r,rp^n)=r$ i.e. $r\sle rp$
\end{block}
Thus also $r\sle rs=rp^m = q^mr$
\end{block}

Case $r\slt t=t'r$:
\begin{block}
$t'\neq \ew$

$\hat{t} = rt'$ as $rt'r = rt = \hat{t}r$

$r s^{k+1} t^{k+1} = r s^{k+1} (t'r)^{k+1} = r s s^k t' \hat{t}^k r$

$L= (rs, r) (s, \hat{t})^* t' = rst'r + (rss,\hat{t}r) (s,\hat{t})^* t'$ with $rst\inv{r} \rdeq rst'$

$x:=rs$, $y:=t'$, $\tilde{s}:=s$

With $r\not\sle rs$:
\begin{block}
$\lcsa{rst\\ rs^{k+2}t^{k+2}} = \lcsa{rst'r\\ rs^{k+2} t^k t'rt'r} = \lcsa{rs\\ r}t'r \slt rt'r = \hat{t}r$

$\lcsa{rs\\ r}t'r = \lcsa{rst'\\ \hat{t}}r$

$\lcs(L) = \lcsa{rst\\ rsstt} = \lcsa{rst'\\ rsst'rt'}r = \lcsa{rst'\\ \hat{t}}r = \lcsa{xy\\ \hat{t}}r \slt \hat{t} r$

$\lcs(st,ts)t=\lcs(L)$ as 
\begin{block}
$\lcsa{rst\\ rsstt} = \lcsa{rst\\rsst'rt}=\rlcsa{rs\\ r}t=
\lcsa{t'rs\\ st'r}t = \lcsa{ts\\ st}t$
\end{block}
\end{block}
\end{block}

Case $t\sle r=r't$:
\begin{block}
$r't = \hat{t} r'$
\begin{block}
$r' t t = r t = \hat{t} r = \hat{t} r' t$
\end{block}
$\exists \check{s}\colon r' s = \check{s} r'$
\begin{block}
$r' t = r \sle r s t = r' t s t = \hat{t} r' s t$

$r' \sle r' s$
\end{block}
$r s^{k+1} t^{k+1} = r' t s^{k+1} t^{k+1} = \hat{t} \check{s} \check{s}^k \hat{t}^k r' t = \hat{t} \check{s} \check{s}^k \hat{t}^k r$

$L=(\hat{t}\check{s}, r) (\check{s},\hat{t})^* \ew$

$x:=\hat{t}\check{s}$, $y:=\ew$, $\tilde{s}:=\check{s}$

With $r\not\sle rs$:
\begin{block}
$\hat{t}\not\sle \hat{t}\check{s}$ as
\begin{block}
$\hat{t}r'=r' t = r\not\sle rs= r'ts = \hat{t}\check{s} r'$
\end{block}

$\lcsa{\hat{t}\check{s}\\ \check{s}\hat{t}} = \lcsa{\hat{t}\check{s}\\ \hat{t}} \slt \hat{t}$

$\lcsa{rst\\ rs^{k+2}t^{k+2}} = \lcsa{\hat{t}\check{s}r\\ \hat{t}\check{s}^{k+2} \hat{t}^{k} \hat{t} r} = \lcsa{\hat{t}\check{s}\\ \hat{t}}r \slt \hat{t}r$

$\lcs(L) = \lcsa{rst\\ rsstt} = \lcsa{\hat{t}\check{s}r\\ \hat{t}\check{s}\check{s}\hat{t}r} = \lcsa{\hat{t}\check{s}\\ \hat{t}} r
=\lcsa{xy\\ \hat{t}}r \slt \hat{t}r = rt$

$\lcs(st,ts)t=\lcs(L)$ as 
\begin{block}
$\lcsa{rst\\ rsstt}=\lcsa{r'tst\\ r'tsstt} = 
\lcsa{ts\\ st}t$
\end{block}

\end{block}
\end{block}
\end{prfblk}

\begin{lemma}\label{lem:pospos}
Let $L=(u,\ew)[(s_1,\inv{r_1}t_1 r_1)+(s_2,\inv{r_2}t_2 r_2)]^\ast w$ be $\wf$ with $r_2 \sle r_1$.

Then $\rlcs(L) = \rlcsa{uw\\ us_1w\inv{r_1}t_1 r_1\\ us_2w\inv{r_2}t_2r_2\\ us_1s_1w\inv{r}_1 t_1 t_1 r_1\\ us_2s_2w\inv{r_2} t_2 t_2r_2}$

If $r_1\sle\rlcs(L)$, then further
$\rlcs(L)=\rlcsa{uw\\ us_1w\inv{r_1}t_1 r_1\\ us_2w\inv{r_2}t_2r_2\\ us_1s_2w\inv{r_2} t_2 r_2 \inv{r_1} t_1 r_1\\ us_2s_1w\inv{r_1}t_1 r_1 \inv{r_2} t_2 r_2}$
\end{lemma}
\begin{prfblk}
W.l.o.g.\ $r_2\sle r_1=r_1'r_2$:
\begin{block}
$us_is_jw\inv{r_j}t_j\underline{r_j\inv{r_i}}t_i r_i$ has to be $\wf$

W.l.o.g.\ $\abs{r_1}\ge \abs{r_2}$

Hence, $r_1\inv{r_2}$ has to be $\wf$
\end{block}
\[L=(u,\ew)[(s_1,\inv{r_1'r_2}t_1 r_1'r_2)+(s_2,\inv{r_2}t_2 r_2)]^\ast w\]

$\exists \hat{t}_2\colon r_1' t_2 = \hat{t}_2 r_1'$:
\begin{block}
$t_2^k r_2 \inv{r_1} \rdeq t_2^k \inv{r_1'}$ has to be $\wf$ for all $k>0$

$r_1'\slt t_2^\iomega$
\end{block}

$R:=\rlcs(L)$

If $R\slt r_2$:
\begin{block}
$R\slt r_2\sle (u,\ew)[(s_1,\inv{r_1'r_2}t_1 r_1'r_2)+(s_2,\inv{r_2}t_2 r_2)]^+ w$

Thus $R=\rlcs(L)=\rlcs(uw, us_iw\inv{r_i}t_i r_i)$
\end{block}

{\em Assume $r_2\sle R=\rlcs(L)$ from here on.}

If $r_1\sle w$:
\begin{block}
$\exists w'\colon w=w'r_1$

Moving $r_1$ from $w$ to the end using $r_1' t_2 = \hat{t}_2 r_1$ yields

\[L\rdeq (u,r_1)[(s_1,t_1)+(s_2,\hat{t}_2)]^\ast w'\]

Apply lemma~\ref{lem:uststw} on $(u,r_1)[(s_1,\ew)+(s_2,\hat{t}_2)]^\ast w'$.
\end{block}

{\em Assume $w\slt r_1=\dot{r}_1 w\wedge \dot{r}_1\neq\ew$ from here on.}

If $w\slt r_2=\dot{r}_2w$:
\begin{block}
$\dot{r}_1=r_1'\dot{r}_2$
\begin{block}
$\dot{r}_1w = r_1=r_1'r_2=r_1'\dot{r}_2 w$
\end{block}
$u=u'\dot{r}_2$
\begin{block}
$\dot{r}_2w =r_2\sle R \sle uw$
\end{block}
Thus
\[L=(u,\ew)[(s_1,\inv{r_1'\dot{r}_2 w}t_1 r_1'\dot{r}_2 w)+(s_2,\inv{\dot{r}_2 w}t_2 \dot{r}_2 w)]^\ast w\rdeq (u'\dot{r}_2,w)[(s_1,\inv{r_1'\dot{r}_2}t_1 r_1'\dot{r}_2)+(s_2,\inv{\dot{r}_2}t_2 \dot{r}_2 )]^\ast \ew\]
$us_i\rdeq us_i\inv{\dot{r}_2}\dot{r}_2$ as
\begin{block}
\begin{itemize}
\item
$us_1w\inv{r_1}t_1r_1\rdeq \underline{us_1\inv{\dot{r}_2}}\inv{r_1'} t_1 r_1$ is $\wf$
\item
$us_2w\inv{r_2}t_2r_2\rdeq \underline{us_2\inv{\dot{r}_2}} t_2 r_2$ is $\wf$
\end{itemize}
\end{block}
$\exists \hat{s}_i\colon \dot{r}_2 s_i = \hat{s}_i \dot{r}_2$:  (this is more generic as required as we could directly use $u=u'\dot{r}_2$)
\begin{block}
\begin{itemize}
\item
$us_1s_1w\inv{r_1}t_1t_1r_1\rdeq us_1\inv{\dot{r}_2}\,\underline{\dot{r}_2 s_1\inv{\dot{r}_2}}\,\inv{r_1'} t_1t_1 r_1$ is $\wf$
\item
$us_2s_2w\inv{r_2}t_2t_2r_2\rdeq us_2\inv{\dot{r}_2}\,\underline{\dot{r}_2 s_2\inv{\dot{r}_2}}\, t_2t_2 r_2$ is $\wf$
\end{itemize}
\end{block}
Thus
\[L\rdeq (u',r_2)[(\hat{s}_1,\inv{r_1'}t_1 r_1')+(\hat{s}_2,t_2)]^\ast \ew\]
This case is thus a special case of $r_2\sle w$ with $r_2=w=\ew$.
\end{block}

{\em Assume $r_2\sle w=w'r_2$ from here on.}

As
\begin{itemize}
\item
$r_2 (\inv{r_1}t_1 r_1)\rdeq (\inv{r_1'}t_1r_1') r_2$ and
\item
$r_2 (\inv{r_2} t_2 r_1)\rdeq t_2 r_2$
\end{itemize}
we can move $r_2$ from $w=w'r_2$ to the end of $L$ s.t.

\[L=(u,\ew)[(s_1,\inv{r_1'r_2}t_1 r_1'r_2)+(s_2,\inv{r_2}t_2 r_2)]^\ast w'r_2\rdeq(u,r_2)[(s_1,\inv{r_1'}t_1 r_1')+(s_2,t_2)]^\ast w'\]

{\em Assume thus w.l.o.g. $r_2=\ew$ from here on s.t. $w=w'$ and $r_1=r_1'$ and $r_1t_2=\hat{t}_2r_1$ and}
\[L=(u,\ew)[(s_1,\inv{r_1}t_1 r_1)+(s_2,t_2)]^\ast w=(u,\ew)[(s_1,\inv{\dot{r}_1w}t_1 \dot{r}_1w)+(s_2,t_2)]^\ast w\]

If $R\slt w$:
\begin{block}
Then:
\begin{itemize}
\item
$w=w'R \wedge w'\neq \ew$
\item
$r_1=\dot{r}_1 w = \dot{r}_1w'R$
\end{itemize}
\[L\rdeq (u,\ew)[(s_1,\inv{\dot{r}_1w'R}t_1 \dot{r}_1w'R)+(s_2,t_2)]^\ast w'R\]
$\exists \tilde{t}_2\colon Rt_2 =\tilde{t}_2 R$ as
\begin{block}
$R\sle us_2wt_2=us_2w'Rt_2$ i.e. $R\sle Rt_2$
\end{block}
\[L \rdeq(u,R)[(s_1,\inv{\dot{r}_1w'}t_1 \dot{r}_1w')+(s_2,\tilde{t}_2)]^\ast w'\]
As $w'\neq\ew$ we have $R=\rlcs(w'R,\tilde{t}_2 R)$

$R=\rlcs(L)=\rlcs(uw,us_2wt_2)$
\end{block}
{\em Assume $w\sle R$ from here on.}

W.l.o.g.\ $w=\ew$ as
\begin{block}
$\exists \tilde{t}_2\colon w t_2 = \tilde{t}_2w$
\begin{block}
$w\sle R \sle us_2wt_2$ i.e. $w\sle wt_2$
\end{block}
Thus
\[L=(u,w)[(s_1,\inv{\dot{r}_1}t_1 \dot{r}_1)+(s_2,\tilde{t}_2)]^\ast \ew\]
\end{block}

Hence:
\begin{itemize}
\item
$u=u' R$
\item
$r_1=\dot{r}_1 w= \dot{r}_1$
\end{itemize}

\[L=(u'R,\ew)[(s_1,\inv{r_1}t_1 r_1)+(s_2,t_2)]^\ast \ew\]

$us_1\rdeq us_1\inv{r_1}r_1$ as
\begin{block}
$us_1\inv{r_1}t_1r_1$ is $\wf$
\end{block}

$\exists \hat{s}_1\colon r_1 s_1=\hat{s}_1r_1$ as
\begin{block}
$us_1s_1\inv{r_1}t_1t_1r_1\rdeq us_1\inv{r_1} \underline{r_1 s_1 \inv{r_1}t_1t_1r_1}$ is $\wf$
\end{block}

$r_1 s_2^l t_2^l \inv{r_1}$ is $\wf$ for all $l$ as
\begin{block}
$us_1 s_2^l t_2^l \inv{r_1} t_1 r_1 \rdeq u s_1 \inv{r_1}\, \underline{r_1 s_2^l t_2^l \inv{r_1}}\, t_1 r_1$ is $\wf$ for all $l$
\end{block}

If $R\slt r_1$:
\begin{block}
$r_1=r_1'R\wedge r_1'\neq \ew$
If we have at least one copy of $(s_1,\inv{r_1}t_1r_1)$, then the word ends on $r_1=r_1'R$:
\[\begin{array}{cc}
  & (u,\ew)[(s_1,\inv{r_1}t_1 r_1)+(s_2,t_2)]^\ast (s_1,\inv{r_1}t_1r_1)[(s_1,\inv{r_1}t_1 r_1)+(s_2,t_2)]^\ast \ew\\
\rdeq & (u,r_1'R)[(s_1,\inv{r_1}t_1 r_1)+(s_2,t_2)]^\ast (s_1,\inv{r_1}t_1)[(s_1,t_1)+(s_2,\hat{t}_2)]^\ast \ew
\end{array}\]
Thus by lemma~\ref{lem:ustst}:

$R=\lcs(us_1\inv{r_1}t_1r_1, u, us_2t_2, us_2s_2t_2t_2)$

{\em In this case we might not be able to replace $us_2s_2wt_2t_2$ by $us_2s_1w\inv{r_1}t_1r_1t_2=us_2s_1w\inv{r_1}t_1\hat{t}_2r_1$ or $us_1s_2wt_2\inv{r_1}t_1r_1$.}
\end{block}

{\em Assume $r_1\sle R=R'r_1$ from here on.}

Thus:
\begin{itemize}
\item
$u=u'R'r_1$
\item
$\exists \tilde{t}_1\colon R't_1=\tilde{t}_1R'$
 \begin{block}
 $R=R'r_1\sle us_1^k\inv{r_1}t_1^k r_1$ for all $k$
 \end{block}
\end{itemize}
\[L=(u'R'r_1,\ew)[(s_1,\inv{r_1}t_1 r_1)+(s_2,t_2)]^\ast \ew\]
If $r_1\sle r_1s_2$:
\begin{block}
$\exists \hat{s}_2\colon r_1s_2=\hat{s}_2r_1$

\[L=(u'R',r_1)[(\hat{s}_1,t_1)+(\hat{s}_2,\hat{t}_2)]^\ast \ew\]

Apply lemma~\ref{lem:ustst}.
\end{block}

{\em Assume $r_1\not\sle r_1s_2$ from here on.}

Thus (see also lemma~\ref{lem:rst})
\begin{itemize}
\item
$s_2t_2\neq t_2s_2$ as $r_1t_2=\hat{t}_2r_1$
\item
$\exists x_2,y_2,\tilde{s}_2\forall l\colon r_1 s_2^{l+1} t_2^{l+1} = x_2 \tilde{s}_2^l y_2 \hat{t}_2^l r_1$
\item
$\lcs((r_1,\ew)(s_2,t_2)^+\ew)=\lcs(x_2y_2,\hat{t}_2)r_1=\lcs(s_2t_2,t_2s_2)t_2\slt \hat{t}_2r_1$ with $\abs{x_2y_2}=\abs{s_2t_2}\ge\abs{t_2}=\abs{\hat{t}_2}$
\item
$R=R'r_1\sle \lcs((u'R'r_1,\ew)(s_2,t_2)^+\ew)=\lcs((x_2,r_1)(\tilde{s}_2,\hat{t}_2)^\ast \ew) \slt x_2y_2r_1, \hat{t}_2 r_1$
\item
$\exists \hat{t}_2'\colon\hat{t}_2=\hat{t}_2'R'\wedge \hat{t}_2'\neq \ew$
\item
$\exists z_2\colon x_2y_2=z_2R'\wedge z_2\neq \ew$
\end{itemize}

Partition $L$ as follows:
\[
\begin{array}{ccl}
1. &  & u\\
 &= & u'R\\[1mm]
2. & & (u,\ew) (s_1,\inv{r_1}t_1 r_1)^+ \ew\\
& \rdeq & (u'R',r_1)(\hat{s}_1,t_1)^+\ew\\[1mm]
3. & & (u,\ew) (s_1,\inv{r_1}t_1 r_1)^\ast (s_2, t_2) \ew\\
& \rdeq & (u'R',r_1) (\hat{s}_1,t_1)^\ast x_2y_2\\
& = & (u'R',r_1) (\hat{s}_1,t_1)^\ast z_2 R'\\
& = & (u'R',R) (\hat{s}_1,\tilde{t}_1)^\ast z_2\\[1mm]
4. & & (u,\ew) (s_1,\inv{r_1}t_1 r_1)^\ast (s_2, t_2) (s_2, t_2)^+ \ew\\
& \rdeq & (u'R',\ew) (\hat{s}_1,\inv{r_1}t_1 r_1)^\ast (r_1 s_2, t_2) (s_2, t_2)^+ \ew\\
& = & (u'R',\ew) (\hat{s}_1,\inv{r_1}t_1 r_1)^\ast (x_2\tilde{s}_2, \hat{t}_2r_1) (\tilde{s}_2, \hat{t}_2)^\ast \ew\\
& = & (u'R',r_1) (\hat{s}_1,t_1)^\ast (x_2\tilde{s}_2, \hat{t}_2'R') (\tilde{s}_2, \hat{t}_2)^\ast \ew\\
& = & (u'R',R) (\hat{s}_1,\tilde{t}_1)^\ast (x_2\tilde{s}_2, \hat{t}_2') (\tilde{s}_2, \hat{t}_2)^\ast \ew\\[1mm]
5. & & (u,\ew) (s_1,\inv{r_1}t_1 r_1)^\ast (s_2, t_2)^+ (s_1,\inv{r_1}t_1r_1) [(s_1,\inv{r_1}t_1 r_1) + (s_2, t_2)]^\ast \ew\\
& = & (u,r_1) (s_1,t_1)^\ast (s_2,\hat{t}_2) (s_2, \hat{t}_2)^\ast (s_1,\inv{r_1}t_1) [(s_1,\inv{r_1}t_1 r_1) + (s_2, t_2)]^\ast \ew\\
& = & (u,R) (s_1,\tilde{t}_1)^\ast(s_2,\hat{t}_2') (s_2, \hat{t}_2)^\ast (s_1,\inv{r_1}t_1) [(s_1,\inv{r_1}t_1 r_1) + (s_2, t_2)]^\ast \ew
\end{array}
\]

If $R'\slt t_1$:
\begin{block}
$t_1=t_1'R'\wedge t_1'\neq \ew\wedge \tilde{t}_1=R't_1'$

The partition of $L$ thus becomes
\begin{enumerate}
\item
$u'R$
\item
$(u'R',r_1)(\hat{s}_1,t_1)^+\ew=\ldots t_1'R$
\item
$(u'R',R) (\hat{s}_1,\tilde{t}_1)^\ast z_2=\ldots z_2 \tilde{t}_1^\ast R=\ldots z_2 R + \ldots t_1' R$
\item
$(u'R',R) (\hat{s}_1,\tilde{t}_1)^\ast (x_2\tilde{s}_2, \hat{t}_2') (\tilde{s}_2, \hat{t}_2)^\ast \ew=\ldots \hat{t}_2'\tilde{t}_1^\ast R = \ldots \hat{t}_2' R + \ldots t_1' R$
\item
$(u,R) (s_1,\tilde{t}_1)^\ast (s_2, \hat{t}_2)^\ast (s_2s_1,\inv{r_1}t_1\hat{t}_2') [(s_1,\inv{r_1}t_1 r_1) + (s_2, t_2)]^\ast \ew=\ldots \hat{t}_2'\tilde{t}_1^\ast R=\ldots\hat{t}_2' R + \ldots t_1' R$
\end{enumerate}
Hence $R=\lcs(u'R,t_1'R, z_2 R, \hat{t}_2'R)$ (as $\tilde{t}_1=R't_1'$)

$R=\rlcs(u,us_1\inv{r_1}t_1r_1, us_2t_2, us_2s_2t_2t_2)=\rlcs(u,us_1\inv{r_1}t_1r_1, us_2t_2, us_2s_1\inv{r_1}t_1r_1t_2)$
\end{block}
{\em Assume $t_1\sle R'=R''t_1$ from here on s.t. $R''t_1 =\tilde{t}_1R''$.}

$\exists \tilde{s}_1\colon R'' \hat{s}_1=\tilde{s}_1R''$
\begin{block}
$R=R''t_1r_1\sle us_1\inv{r_1}t_1r_1=u'R''t_1r_1s_1\inv{r_1}t_1r_1=u'\tilde{t}_1\underline{R''\hat{s}_1}t_1r_1$

i.e.\ $R''\sle R''\hat{s}_1$
\end{block}
Hence:
\[
\begin{array}{ccl}
1. &  & u\\
 &= & u'R\\[1mm]
2. & & (u,\ew) (s_1,\inv{r_1}t_1 r_1)^+ \ew\\
& \rdeq & (u'R''t_1,r_1)(\hat{s}_1,t_1)^+\ew\\
& \rdeq & (u'\tilde{t}_1\tilde{s}_1,R)(\tilde{s}_1,\tilde{t}_1)^\ast\ew\\[1mm]
3. & & (u,\ew) (s_1,\inv{r_1}t_1 r_1)^\ast (s_2, t_2) \ew\\
& \rdeq & (u'R',R) (\hat{s}_1,\tilde{t}_1)^\ast z_2\\[1mm]
4. & & (u,\ew) (s_1,\inv{r_1}t_1 r_1)^\ast (s_2, t_2) (s_2, t_2)^+ \ew\\
& \rdeq & (u'R',R) (\hat{s}_1,\tilde{t}_1)^\ast (x_2\tilde{s}_2, \hat{t}_2') (\tilde{s}_2, \hat{t}_2)^\ast \ew\\[1mm]
5. & & (u,\ew) (s_1,\inv{r_1}t_1 r_1)^\ast (s_2, t_2)^+ (s_1,\inv{r_1}t_1r_1) [(s_1,\inv{r_1}t_1 r_1) + (s_2, t_2)]^\ast \ew\\
& \rdeq& (u,R) (s_1,\tilde{t}_1)^\ast (s_2, \hat{t}_2') (s_2, \hat{t}_2)^\ast (s_1,\inv{r_1}t_1) [(s_1,\inv{r_1}t_1 r_1) + (s_2, t_2)]^\ast \ew
\end{array}
\]
If $\tilde{t}_1\neq \ew$:
\[\begin{array}{lcl}
R 
& = & \lcs(u'R,\tilde{s}_1R,\tilde{t}_1R,z_2R,\hat{t}_2'R)\\
& = & \rlcs(u,us_1\inv{r_1}t_1r_1, us_1s_1\inv{r_1}t_1t_1r_1,us_2t_2,us_2s_2t_2t_2)\\
& = & \rlcs(u,us_1\inv{r_1}t_1r_1, us_1s_2t_2\inv{r_1}t_1r_1,us_2t_2,us_2s_1\inv{r_1}t_1r_1t_2)
\end{array}\]
If $\tilde{t}_1=\ew$:
\[\begin{array}{lcl}
R 
& = & \lcs(u'R,\tilde{s}_1R,z_2R,\hat{t}_2'R)\\
& = & \rlcs(u,us_1\inv{r_1}t_1r_1, us_2t_2,us_2s_2t_2t_2)\\
& = & \rlcs(u,us_1\inv{r_1}t_1r_1, us_2t_2,us_2s_1\inv{r_1}t_1r_1t_2)
\end{array}\]
\end{prfblk}

\begin{lemma}\label{lem:negneg}
Let $L=(u,\ew)[(s_1,\inv{r_1}\inv{t_1} r_1)+(s_2,\inv{r_2}\inv{t_2} r_2)]^\ast w$ be $\wf$ with $t_1\neq \ew$.
Then 
\[\exists p, k_1, k_2\forall i_1,\ldots, i_l, j \in\{1,2\}^+\colon u s_{i_1}\ldots s_{i_l} s_j w\inv{r_j} \inv{t_j} r_j \rdeq us_{i_1}\ldots s_{i_l} p^{k_j} w\]
s.t. $L\rdeq u(p^{k_1}+p^{k_2})^\ast w$ and $\rlcs(L) = \rlcsa{uw\\ us_1w\inv{r_1}\inv{t_1} r_1\\ us_2w\inv{r_2}\inv{t_2} r_2}$.
\end{lemma}
\begin{prfblk}
$R:=\rlcs(L)$

$\exists p_i, \hat{p}_i, m_i,n_i \colon s_i = p_i^{m_i}\wedge m_i \ge n_i \wedge s_i w\inv{r_i} \inv{t_i} \rdeq p_i^{m_i-n_i} w\inv{r_i}$
\begin{block}

If $t_i = \ew$, then:
\begin{block}
Let $p_i, m_i$ s.t. $s_i=p_i^{m_i}$ and $p_i$ primitive; set $n_i:=0$.
\end{block}

Assume thus $t_i\neq \ew$. (Note that $t_1 \neq \ew$ already by assumption of the lemma.)

Then $s_i^\iomega w_i = t_i^\iomega r_i$ as 
\begin{block}
$us_i^k w \inv{r_i} \inv{t_i}^k$ is $\wf$ for all $k\ge 1$
\end{block}
Hence by Lemma~\ref{lem:lcs-calculus}
\[\exists p_i, \hat{p}_i, m_i,n_i \colon s_i = p_i^{m_i}\wedge t_i=\hat{p}_i^{n_i}\wedge p_iw\inv{r_i}\inv{\hat{p}_i}\rdeq w\inv{r_i}\]
$m_i\ge n_i$ as $L$ is $\wf$ and thus nonnegative.
\end{block}

For all $i_1,\ldots, i_l, j\in\{1,2\}^+$ we thus have
\[
u s_{i_1}\ldots s_{i_l} s_j w\inv{r_j}\inv{t_j}r_j \rdeq u s_{i_1}\ldots s_{i_l} (p_j^{m_j-n_j}) w\inv{r_j} r_j \rdeq u s_{i_1}\ldots s_{i_l} (p_j^{m_j-n_j}) w
\]

If $m_1=n_1\wedge m_2=n_2$, then:
\begin{block}
For all $i_1,\ldots, i_l\in\{1,2\}^+$:
\[
u s_{i_1}\ldots s_{i_l} w \inv{r_{i_l}} \inv{t_{i_l}} r_{i_l} \ldots  \inv{r_{i_1}} \inv{t_{i_1}} r_{i_1}\rdeq uw
\]
s.t.\ $L\rdeq uw$ with $\rlcs(L)=uw$.

Hence, $p$ can be chosen as $\ew$.
\end{block}
So assume that $m_1>n_1\vee m_2>n_2$.

We then have $p_1=p_2=:p$
\begin{block}
If $s_2= \ew$, then:
\begin{block}
$t_2= \ew$ as $L$ is $\wf$ and thus nonnegative.

Thus simply choose $p_2$ as $p_1$
\end{block}

So assume $s_2\neq \ew$ and so $p_2\neq \ew\wedge m_2>0$.

Then for all $k,l\ge 1$ the following has to be $\wwf$:
\[\begin{array}{lcl}
s_1^k s_2^l w\inv{r_2} \inv{t_2}^l r_2 \inv{r_1} \inv{t_1}^k
& = & p_1^{km_1} p_2^{l(m_2-n_2)} w\inv{r_2} r_2 \inv{r_1} \inv{t_1}^{k}\\
&\rdeq & p_1^{km_1} \underline{p_2^{l(m_2-n_2)} w \inv{r_1} \inv{t_1}^{k}}\\
\end{array}\]
If $m_2>n_2$, then:
\begin{block}
$ p_2^\iomega w = t_1^\iomega r_1 = p_1^\iomega w$
\end{block}
So assume $m_2 = n_2>0$ and thus $t_2\neq \ew$ and $m_1>n_1$.

Then for all $k,l\ge 1$ the following has to be $\wwf$:
\[\begin{array}{lcl}
s_2^k s_1^l w\inv{r_1} \inv{t_1}^l r_1 \inv{r_2} \inv{t_2}^k
& = & p_2^{km_2} p_1^{l(m_1-n_1)} w\inv{r_1} r_1 \inv{r_2} \inv{t_2}^{k}\\
&\rdeq & p_2^{km_2} \underline{p_1^{l(m_1-n_1)} w \inv{r_2} \inv{t_2}^{k}}\\
\end{array}\]
Hence:
$p_1^\iomega w = t_2^\iomega r_2 = p_2^\iomega w$
\end{block}
As $s_1,s_2\in p^+$:
\[L\rdeq u(p^{m_1-n_j}+p^{m_2-n_2})^\ast w\]

$\rlcs(L) = \rlcs(uw, us_1w\inv{r_1}\inv{t_1} r_1, us_2w\inv{r_2}\inv{t_2}r_2)$
\end{prfblk}

\begin{lemma}\label{lem:negpos}
Let $L=(u,\ew)[(s_1,\inv{r_1}\inv{t_1} r_1)+(s_2,\inv{r_2}t_2 r_2)]^\ast w$ be $\wf$ with $t_1\neq \ew$.

Then $\exists p, k_1, k_2\forall i_1,\ldots, i_l \in\{1,2\}^+$ s.t.\
\begin{enumerate}
\item
$u s_{i_1}\ldots s_{i_l} s_1 w\inv{r_1} \inv{t_1} r_1 \rdeq us_{i_1}\ldots s_{i_l} p^{k_1} w$
\item
$u s_{i_1}\ldots s_{i_l} s_2 w\inv{r_2} t_2 r_2 \rdeq us_{i_1}\ldots s_{i_l} p^{k_2} w$
\item
$L\rdeq u(p^{k_1}+p^{k_2})^\ast w$ 
\item
$\rlcs(L) = \rlcsa{uw\\ us_1w\inv{r_1}\inv{t_1} r_1\\ us_2w\inv{r_2}\inv{t_2} r_2}$
\end{enumerate}
\end{lemma}
\begin{prfblk}
W.l.o.g.\ $t_2\neq \ew$ otherwise see Lemma~\ref{lem:negneg}.

$R:=\lcs(L)$

$\exists p, \hat{p}, m_1, n_1\colon s_1 = p^{m_1} \wedge t_1 = \hat{p}^{n_1}\wedge m_1\ge n_1 \wedge s_1w\inv{r_1}\inv{t_1} \rdeq p^{m_1-n_1}w\inv{r_1}$
\begin{block}
$s_1^{k} w\inv{r_1} \inv{t_1}^{k}$ is $\wwf$ for all $k\ge 1$

$s_1^\iomega w\seq t_1^\iomega r_1$ as $t_1\neq \ew$
\end{block}

W.l.o.g. $p,\hat{p}$ primitive

$p^\iomega w = \hat{p}^\iomega r_1$

$\exists \check{p}, n_2\colon t_2 = \check{p}^{n_2} \wedge \check{p} r_2 \inv{r_1} \inv{\hat{p}}\rdeq r_2 \inv{r_1}$ as
\begin{block}
$t_2^{l}r_2 \inv{r_1} \inv{t}_1^{k} = t_2^l r_2 \inv{r_1} \hat{p}^{n_1\cdot k}$ has to be $\wwf$ for all $k,l\ge 1$

$t_2^\iomega r_2 = \hat{p}^\iomega r_1$
\end{block}

$\check{p}$ primitive as $\hat{p}$ primitive
 
$\exists m_2\colon s_2 = p^{m_2}$ as
\begin{block}
If $s_2= \ew$, set $m_2:=0$.

So assume $s_2\neq \ew$.

Then $r_2 \slt s_2^\iomega w$ as
\begin{block}
$s_2^{l} w\inv{r_2}$ is $\wf$ for $l$ sufficiently large
\end{block}

So $s_2^\iomega w \seq \hat{p}^\iomega r_1$ as
\begin{block}
$s_2^{l}w\inv{r_2} t_2^{l} r_2\inv{r_1}\inv{t_1}^{k}$ is $\wwf$ for all $k,l\ge 1$

Let $c> n_2$ and $l$ so large that $r_2 \slt s_2^l w$:

\[s_2^{l}w\inv{r_2} t_2^{l} r_2\inv{r_1}\inv{t_1}^{cl}=s_2^{l}w\inv{r_2} \check{p}^{n_2l} r_2\inv{r_1}\inv{\hat{p}}^{cln_1} \rdeq s_2^{l}w\inv{r_2} r_2\inv{r_1}\inv{\hat{p}}^{(cn_1-n_2)l}\rdeq s_2^{l}w\inv{r_1}\inv{\hat{p}}^{(cn_1-n_2)l}\]
\end{block}
\end{block}

$pw\inv{r_2} \inv{\check{p}} \rdeq w\inv{r_2}$
\begin{block}
$p^\iomega w = \hat{p}^\iomega r_1 = \check{p}^\iomega r_2$
\end{block}

For all $i_1\ldots i_l\in \{1,2\}^\ast$ we have:
\begin{itemize}
\item
$u s_{i_1}\ldots s_{i_l} s_1 w \inv{r_1} \inv{t_1} r_1 \rdeq  u s_{i_1}\ldots s_{i_l} (p^{m_1-n_1}) w$
\item
$u s_{i_1}\ldots s_{i_l} s_2 w \inv{r_2} t_2 r_2 \rdeq  u s_{i_1}\ldots s_{i_l} (p^{m_2+n_2}) w$ as
\begin{block}

W.l.o.g.\ $s_2\neq \ew$.

If $w\slt r_2=r_2'w$, then:
\begin{block}
$u s_{i_1}\ldots s_{i_l} s_2 w \inv{r_2}\rdeq u s_{i_1}\ldots s_{i_l} s_2 \inv{r_2'}$ is $\wf$ for all $i_1\ldots i_l\in \{1,2\}^\ast$.

Further $\check{p} r_2\inv{w}\inv{p}\rdeq r_2\inv{w}\rdeq r_2'$ s.t.\
\[t_2 r_2\inv{w} \rdeq \check{p}^{n_2} r_2 \inv{w}\inv{p}^{n_2} p^{n_2} \rdeq r_2 \inv{w} p^{n_2}\]

Hence:
\[\begin{array}{lcl}
u s_{i_1}\ldots s_{i_l} s_2 w \inv{r_2} t_2 r_2
&\rdeq & u s_{i_1}\ldots s_{i_l} s_2 \inv{r_2'} r_2' w \inv{r_2} t_2 r_2\\
&\rdeq & u s_{i_1}\ldots s_{i_l} s_2 \inv{r_2'} t_2 r_2\inv{w}w \text{\hspace{4mm} as } w\inv{r_2} \rdeq \inv{r_2'}\\
&\rdeq & u s_{i_1}\ldots s_{i_l} s_2 \inv{r_2'} r_2' p^{n_2} w\\
&\rdeq & u s_{i_1}\ldots s_{i_l} p^{m_2+n_2} w
\end{array}\]
\end{block}

So assume $r_2\sle w=w'r_2$.

Then $p w' = w' \check{p}$ as:
\begin{block}
$p w\inv{r_2} \rdeq p w\inv{r_2} \inv{\check{p}} \check{p} \rdeq w\inv{r_2} \check{p} \rdeq w' \check{p}$
\end{block}
Thus:
\[\begin{array}{lcl}
u s_{i_1}\ldots s_{i_l} s_2 w \inv{r_2} t_2 r_2
&\rdeq & u s_{i_1}\ldots s_{i_l} s_2 w' t_2 r_2\\
&\rdeq & u s_{i_1}\ldots s_{i_l} p^{m_2} w' \check{p}^{n_2} t_2 r_2\\
&\rdeq & u s_{i_1}\ldots s_{i_l} p^{m_2+n_2} w' r_2\\
&\rdeq & u s_{i_1}\ldots s_{i_l} p^{m_2+n_2} w\\
\end{array}\]
\end{block}
\end{itemize}
Thus:
\[L\rdeq u (p^{m_1-n_1} + p^{m_2+n_2})^\ast w\]
and
$\rlcs(L) = \lcs(uw,upw) = \rlcs(uw, us_1w\inv{r_1}\inv{t_1} r_1, us_2 w\inv{r_2} t_2 r_2)$
\end{prfblk}

%% file: proofs-wf-convergence.tex
\subsection{Lemma~\ref{thm:converge-wf-main-text} in the main work}\label{app:conv-wf}

We first show that the computation of $\lcsext$ can be reduced to that of $\lcs$
which is essential to the proof of Lemma~\ref{thm:converge-wf-main-text},
give an example on that and then prove Lemma~\ref{thm:converge-wf-main-text}.

\begin{lemma}\label{lem:exttolcs}
Let $L\subseteq \als$ with $R=\lcs(L)$. If $\lcsext(L)\in\als$, then:
\[\forall xR\in L\setminus\{R\}\,\exists m\in\N\colon \lcs(x^mL)=\lcsext(L)\lcs(L) \quad\text{and}\quad \lcsext(L) \slt x^m\]
\end{lemma}
\begin{prfblk}
Fix any $xR\in L\setminus\{R\}$.

Thus $x\neq \ew$ s.t.\ there is some $m\in\N$ with $\abs{x^m}> \abs{\lcsext(L)}$.

As $\lcsext(L)\in\als$ there is some $y\in L$ s.t.\ for all $zR\in L$:
\[\lcsext(L)
=\lcs(x^\iomega,y^\iomega)=\lcs(xy,yx)=\lcs(x^\iomega,y^\iomega,z^\iomega)=\lcsa{\lcs(x^\iomega,y^\iomega)\\ \lcs(x^\iomega,z^\iomega)}\slt x^m\sle x^\iomega\]

Pick any $zR \in L$.

If $\lcs(x^\iomega,z^\iomega)\sge x^m$, then:
\begin{block}
$\lcs(x^\iomega,z^\iomega)\sge\lcs(x^m, z^\iomega)=x^m$
\end{block}
If $\lcs(x^\iomega,z^\iomega)\slt x^m$, then:
\begin{block}
$\lcs(x^\iomega,z^\iomega)=\lcs(x^m, z^\iomega)= \lcs(x^m, x^mz)$ (Lemma~\ref{lem:lcs-calculus})
\end{block}
In particular for $z=y$ we thus have 
\[\lcsext(L)=\lcs(x^\iomega,y^\iomega)=\lcs(x^m,x^my)\]
Hence:
\[\begin{array}{lcl}
\lcs(x^mL)=\lcs(x^m zR \mid zR \in L) 
& = &  \lcs(x^m xR, x^m yR, x^m z R\mid zR\in L)\\
& = & \lcs(x^m xR, x^m y R) = \lcs(x^{m}, x^m y) R= \lcsext(L) R
\end{array}\]

\end{prfblk}

\begin{example}\label{ex:exttolcs}
In the case of $L=\{s^kt^kR\mid k\in\N_0\}$ with $R=\lcs(L)$
we have 
$\lcsext(L)=\lcs((st)^\iomega, (s^{k+2}t^{k+2})^\iomega\mid k\ge 0)$.
If $s$ and $t$ commute, then $\lcsext(L)=(st)^\iomega$, and $\lcsext(L)=(st)^\iomega$ is unbounded.
Thus assume $st\neq ts$ s.t.\ $\lcs(ts,st^{k+1})=\lcs(ts,st)$ and 
\[
\lcs((st)^\iomega, (s^{k+2}t^{k+2})^\iomega)= 
\lcs((ts)^\iomega tst, (st^{k+2}s^{k+1})^\iomega s t^{k+1} t) =
\lcs(tst,stt)
\]
Hence, $\lcs(ststL)=\lcs(tst,stt)R=\lcsext(L)\lcs(L)$.\qed
\end{example}

Recall, $\sT_X^{\le h}:=\tsr(r_X L_X^{\le })$ denotes an $\tseq$-equivalent sublanguage of $\rd(r_XL_X^{\le h})$.
\begin{lemma}[Lemma~\ref{thm:converge-wf-main-text} in the main work]\label{tm:converge-wf}
Let $G$ be a context-free grammar with $N$ nonterminals and $L:=L(G)$ be $\wf$.
For every nonterminal $X$ let $r_X\in\als$ s.t.\ $\abs{r_X}=d_X$ and $r_XL_X$ $\wf$.
Then $\rd(r_X L_X) \tseq \rd(r_X L_X^{\le 4N})$, and thus $\sT_X^{\le 4N}\tseq \sT_X^{\le 4N+1}$ for every nonterminal $X$.
\end{lemma}
\begin{prfblk}
Let $S$ be the axiom of $G$. 

W.l.o.g.\ $G$ is reduced to the nonterminals which are reachable from $S$ and which are productive.
Additionally assume that $L(G)$ contains at least two words; otherwise the claim of the lemma follows directly.

As $r_X L_X$ is $\wf$, we have for any $\zeta\in L_X$ that $\rd(\zeta)=\inv{u}v\in\al^\ast$ with $u\sle r_X=r_X'u$ s.t.\ $r_X \zeta \rdeq r_X' v$ and thus: 
\[\abs{\rd(r_X\zeta)}=\abs{r_X'v}=\abs{r_Xv}-\abs{u}=\hd(r_X\zeta)\]

Let $\kappa_0\in L_X$ be a shortest-word-after-reduction i.e.\ $\hd(r_X\kappa_0)=\min\{\hd(r_X\zeta)\mid \zeta \in r_X L_X\}$.

Let $\kappa_1\in L_X$ be a second-shortest-word-after-reduction i.e.\ $\hd(r_X\kappa_1)=\min\{\hd(r_X\zeta)\mid \zeta \in r_X L_X, \mid \hd(\zeta)>\hd(\kappa_0)\}$.

Then w.l.o.g.\ $\kappa_0 \in L_X^{\le N}$ and $\kappa_1\in L_X^{\le 2N}$ as:
\begin{block}
For any $S\to_G^\ast \alpha X \beta$ and $X\to_G^\ast \gamma$ we need to have that 
\[\forall k\ge 0\colon \hd((\alpha,\beta)(\sigma,\tau)^k\gamma)=\hd(\alpha\gamma\beta)+k\hd(\sigma\tau)\ge0.\]
In fact, this has to hold also for any prefix of a word of $L$ and $r_X L_X$ s.t.\ also
\[\forall k\ge 0\colon \hd(\alpha \sigma^k \gamma)=\hd(\alpha\gamma)+k\hd(\sigma)\ge0\]
and thus $\hd(\tau)\ge - \hd(\sigma)$.

Any word $\zeta\in L_X\setminus L_X^{\le N}$ has a derivation tree with a path from its root to some leaf along which at least $N+1$ nontermimals occur, i.e.\ along at least one nonterminal occurs twice which gives rise to a factorization of the form
\[\zeta = (\alpha,\beta)(\sigma,\tau)\gamma\]
s.t.\ 
\[\abs{\rd(\zeta)}=\hd(\zeta)=\hd(\alpha\gamma\beta)+ \hd(\sigma\tau)\ge \hd(\alpha\gamma\beta)=\abs{\rd(\alpha\beta\gamma)}\]
Removing the pumping tree that gives rise to the factor $(\sigma,\tau)$ thus leads to a word $\alpha\gamma\beta$ that is shorter than $\zeta$ before reduction, and at most as long as $\zeta$ after reduction.

Hence, $r_X L_X^{\le N}$ already contains all shortest-words-after-reduction, i.e.\
\[\min\{\hd(r_X\zeta)\mid \zeta \in r_X L_X\} = \min\{\hd(r_X\zeta)\mid \zeta \in r_X L_X^{\le N}\}\]
and thus w.l.o.g.\ $\kappa_0\in L_X^{\le N}$.

Assume there exists a second-shortest-word-after-reduction $\kappa_1\in L_X$.

Any path that consists of at least $2N+1$ nonterminals contains at least one terminal three times which gives rise to a factorization of the form
\[\kappa_1 = (\alpha,\beta)(\sigma_1,\tau_1)(\sigma_2,\tau_2)\gamma\]
If $\hd(\sigma_1\tau_1)=\hd(\sigma_2\tau_2)$, we can prune both pumping trees.

So assume $\hd(\sigma_i\tau_i)>0$ for either $i=1$ or $i=2$.

Pruning $(\sigma_i,\tau_i)$ leads to $(\alpha,\beta)(\sigma_j,\tau_j)\gamma$ ($j\neq i$) with
\[\hd(\kappa_1)> \hd((\alpha,\beta)(\sigma_j,\tau_j)\gamma)\ge \hd(\kappa_0)\]
As $\kappa_1$ is a second-shortest word after reduction, we have to have 
\[\hd((\alpha,\beta)(\sigma_j,\tau_j)\gamma)= \hd(\kappa_0)\]
and thus
\[\hd(\sigma_j\tau_j)=0\]
s.t.\ we can prune $(\sigma_j,\tau_j)$ to obtain
\[\hd(\kappa_1)= \hd((\alpha,\beta)(\sigma_i,\tau_i)\gamma)\]
a possible different second-shortest-word-after-reduction.
\end{block}

Let $R:=\rlcs(r_X L_X)$.

If the $\rlcs$ of $r_X L_X$ can be extended infinitely:
\begin{block}
Then $\rd(\kappa_0)=R$ and $\rd(\kappa_1)= xR$ for some $x\in\al^+$.

Then for any $\zeta\in L_X$ we have $\rd(r_X \zeta)= yR$ with 
$x^\iomega\seq y^\iomega$, i.e.\ $xy=yx$.

Thus $\tsn(r_X L_X)$ is in this case given by $\{\rd(r_X\kappa_0), \rd(r_X \kappa_1)\}\subseteq r_X L_X^{\le 2N}$
\end{block}

{\em Assume thus w.l.o.g.\ that the $\rlcs$ of $r_X L_X$ can at most be finitely extended.}
\begin{block}
This implies that $\rd(r_X L_X)$ contains at least three distinct words.
\end{block}

We distinguish the two cases whether $R=\rlcs(r_X L_X)$ is a strict suffix of every word in $r_X L_X$, in particular $R\slt \rd(r_X \kappa_0)$, or if $R$ is a, and thus the shortest-word-after-reduction, in particular $R=\rd(r_X \kappa_0)$.
\begin{itemize}
\item
If $R=\rlcs(r_X L_X)\slt \rd(r_X\kappa_0)$, there is some witness $\kappa\in L_X$ s.t.\ 
\[R=\lcs(\rd(r_X \kappa_0),\rd(r_X \kappa))\slt r_X \kappa_0\]
In particular, we have that $r_X \kappa_0\rdeq \ldots a R$ and $r_X \kappa \rdeq \ldots b R$ for two distinct opening parenthesis $a,b\in \al$ ($a\neq b$).
\item
If $R=\rlcs(r_XL_X)=\rd(r_X \kappa_0)$, then recall Lemma~\ref{lem:exttolcs}:

The maximal extension $E$ of the $\lcs$ $R$ is given by
\[E =\lcsext^{\rd}(L)= \lcs( \rd(r_X \zeta\inv{R})^\iomega \mid \zeta \in L_X, \rd(r_X\zeta)\neq R)\]
As $E$ is assumed to be finite, $\rd(r_XL_X)$ has to contain at least two other reduced words, both longer than $\rd(r_X \kappa_0)$.
In particular, there has to be a second-shortest-word-after-reduction $\kappa_1$ s.t.\ we find a witness $\kappa$ for $E$ w.r.t.\ $\kappa_1$:
\[E = \lcs( \rd(r_X \zeta\inv{R})^\iomega \mid \zeta \in L_X, \rd(r_X\zeta)\neq R)= \lcs( \rd(r_X\kappa_1 \inv{R})^\iomega, \rd(r_X\kappa\inv{R})^\iomega)\]
Let $\rd(r_X \kappa_1)= xR$ with $x\neq \ew$. Choose $m>0$ s.t.\ $E\slt x^m$. 

Pick any $\zeta\in L_X$ with $\rd(r_X \zeta)=zR\wedge z\neq \ew$ and w.l.o.g.\ $xz\neq zx$.

We then have
\[\lcs(x^\iomega,z^\iomega)=\lcs(xz,zx)\]
If $\lcs(x^\iomega,z^\iomega) \slt x^m$, then
\[\rlcs( x^\iomega, z^\iomega )=\rlcs( x^m, x^mz)\slt x^m\]
If otherwise $\lcs(x^\iomega,z^\iomega) \sge x^m$, then
\[\rlcs( x^\iomega, z^\iomega )\sge \rlcs( x^m, x^mz)=x^m\]
Thus, 
\[ER = \rlcs(x^m r_X L_X)\slt \rd(x^m r_X \kappa_0)=x^m R= x^{m-1}\rd(r_X \kappa_1)\]
i.e.\ we can reduce this case to the case where $R$ is a strict suffix of any word in $r_X L_X$ by extending $r_X$ to $x^m r_X$.

Note that then any witness for $ER$ w.r.t.\ $x^mr_X\kappa_0$ is also a witness w.r.t.\ $x^mr_X\kappa_1$ and vice versa.
\end{itemize}

{\em Assume thus w.l.o.g.\ that $R=\rlcs(r_X L_X)\slt \rd(r_X\kappa_0)$ from here on.}

Choose $\kappa$ in $L_X$ s.t.\
\begin{enumerate}
\item $\lcs(\rd(r_X\kappa_0), \rd(r_X\kappa)) = \rlcs(r_XL_X)$.
\item $\abs{\kappa}$ is minimal w.r.t.\ to all words in $r_XL_X$ satisfying {\bf 1.},
\item $\abs{\rd(r_X\kappa)}$ is minimal w.r.t.\ to all words in $r_XL_X$ satisfying {\bf 2.}.
\end{enumerate} 

W.l.o.g.\ $\rd(r_X\kappa_0) = \ldots aR$ and $\rd(r_X\kappa) = \ldots bR$ with $a\neq b$ and $a,b\in\al$.

Note that there is a unique factorization $r_X\kappa = \zeta b \xi$ s.t.\ both $\zeta$ and $\xi$ are $\wf$ and $\rd(\xi)=R$:
\begin{block}
For every prefix (before reduction) $\pi$ of $r_X\kappa$ we can interpret $\hd(\pi)$ as the height of the last letter of $\pi$.

Then the $b$ in $\rd(r_X\kappa)= \ldots b R$ is the last letter in $r_X\kappa$ of height $\hd(r_X\kappa) - \abs{R}$ and is, thus, uniquely identified.

This specific $b$ splits $r_X\kappa$ into $r_X\kappa= \zeta b \xi$; 
as this $b$ is the last letter in $\kappa$ on height $\hd(\kappa) - \abs{R}$, 
$\xi$ has to be $\wf$ with $\rd(\xi) = R$;
as $r_X\kappa$ is $\wf$ and $\zeta$ is a prefix thereof, trivially also $\zeta$ is $\wf$.
\end{block}
Assume every derivation tree of $\kappa$ contains a path to a letter within $b \xi$ 
along which some nontermimal $A$ occurs at least $4$ times (see Fig.~\ref{fig:lcp-height}).

This gives rise to a factorization
\[
\kappa = (\alpha, \beta) (\sigma_1,\tau_1) (\sigma_2,\tau_2) (\sigma_3,\tau_3) \gamma
\]
Any word
\[
\zeta\in\{ (\alpha, \beta) (\sigma_1,\tau_1)^{k_1} (\sigma_2,\tau_2)^{k_2} (\sigma_3,\tau_3)^{k_3} \gamma \mid k_i \in \{0,1\}, k_1+k_2+k_3 < 3\}
\]
is shorter (before reduction) than $\kappa$, hence cannot be a witness w.r.t.\ $\kappa_0$ i.e.\ $aR\sle \rlcs(r_X\kappa, r_X\zeta)$.

Let 
\[L=(r_X\alpha,\beta)[(\sigma_1,\tau_1)+\ldots+(\sigma_3,\tau_3)]^\ast \gamma\]
$L$ is $\wf$ with $R=\rlcs(r_X L_X)=\rlcs(L)$ as $L\subseteq r_X L_X$ contains both $r_X\kappa$ and $r_X\alpha\gamma\beta$ with the latter not a witness w.r.t.\ $r_X\kappa_0$ s.t.\ $\rlcs(r_X\kappa, r_X\alpha\gamma\beta)=R$.

Our goal is to show that already $r_X\alpha\sigma_i \gamma \tau_i \beta$ or $r_X\alpha \sigma_i \sigma_j \gamma \tau_j \tau_i \beta$ for some $i\neq j$ is a witness w.r.t. $r_X\alpha\gamma\beta$.

Note that:
\begin{itemize}
\item 
$r_X \alpha$ has to be $\wf$, all other factors $\beta,\gamma,\sigma_i,\tau_i$ have to be $\wwf$.
\item
As noted already at the beginning, we have both $\hd(\sigma_i)\ge 0$ and $\hd(\sigma_i\tau_i)\ge0$.
\item
By choice of the path used for the factorization, we have $\rd(\tau_3\tau_2\tau_1\beta)=\inv{x}y$ with $y\sle R=\rlcs(L)=\rlcs(r_X L_X)$.
\end{itemize}

We first reduce the factors $\alpha,\beta,\gamma,\sigma_i$ to words in $\al^\ast$.

As $r_X \alpha$ has to be $\wf$, simply set $u:=\rd(r_X\alpha)\in\al^\ast$.

W.l.o.g.\ we may assume $\beta=\ew$.
\begin{block}
This amounts to changing $R=\rlcs(L)$ to $R:=\rd(\rlcs(L)\inv{\beta})$.
\end{block}

W.l.o.g.\ we may also assume $\rd(\sigma_i) = s_i\in\al^\ast$:
\begin{block}
Let $\rd(\sigma_i) = \inv{x_i} y_i$ for any $i\in[3]$.

Then $u\inv{x_i}$ has to be $\wf$ for all $i\in[3]$, i.e. we have $u\rdeq u\inv{x_i} x_i$ for all $i\in[3]$.

As $\hd(\sigma_i) \ge 0$, we have $\inv{x_i} \sle \inv{y_i}$.

As $u \sigma_i \sigma_i \rdeq u \inv{x_i} y_i \inv{x_i} y_i$ has to be $\wf$, we have $y_i = s_i x_i$.

Pick $J\in [3]$ s.t. $x_J$ is a longest word of $\{x_1,x_2,x_3\}$.

Then $u\sigma_i \sigma_J \rdeq u \inv{x_J} x_J \inv{x_i} s_i x_i \inv{x_J} s_J x_J$ has to be $\wf$, i.e. $(x_J \inv{x_i}) s_i \inv{x_J \inv{x_i}}$ is $\wf$ for all $i\in [3]$.

Thus there exist $\hat{s}_i$ s.t. $\rd(x_J \inv{x_i}) s_i = \hat{s}_i \rd(x_J \inv{x_i})$ resp. $x_J \inv{x_i} s_i \rdeq \hat{s}_i x_J \inv{x_i}$.

So $u \sigma_{i_1} \ldots \sigma_{i_l} \rdeq u\inv{x_J} \hat{s}_{i_1} \ldots \hat{s}_{i_l} x_J$.

Thus, set $u:=\rd(u\inv{x_J})$, $\gamma:= \rd(x_J \gamma)$ and $\sigma_i:=\hat{s}_i$.
\end{block}

Analogously, we may further assume $\rd(\gamma)= w\in\al^\ast$:
\begin{block}
If $\rd(\gamma) = \inv{x} w$, then $u=u'x$ resp. $u\rdeq u\inv{x} x$ and thus $us_i\gamma \rdeq u\inv{x}x s_i \inv{x} w$.

So $x s_i \inv{x}$ is $\wf$ for all $i\in[3]$.

Hence, we find $\hat{s}_i$ with $x s_i = \hat{s}_i x$ s.t. $u s_{i_1} \ldots s_{i_l} \gamma \rdeq u\inv{x} \hat{s}_{i_1} \ldots \hat{s}_{i_l} x \inv{x} w$

Thus, set $s_i := \hat{s}_i$, $u:= \rd(u\inv{x})$ and $\gamma:=w$.
\end{block}
We thus may simply assume that
\[L=(u,1)[(s_1,\tau_1)+\ldots+(s_3,\tau_3)]^\ast w\]
with $L$ $\wf$ and $\rd(\tau_3\tau_2\tau_1\beta)=\inv{x}y$ with $y\sle R=\rlcs(L')=\rlcs(r_X L_X)$.

For $i\in[3]$ we have either $\rd(\tau_i) = \inv{r_i} t_i r_i$ or $\rd(\tau_i) = \inv{r_i} \inv{t_i} r_i \wedge t_i \neq \ew$:
\begin{block}
Let $\rd(\tau_i)=\inv{x_i} y_i$ for $i\in[3]$.

Then $L_i:=(u,\ew)(s_i,\tau_i)^\ast w$ has to be $\wf$ for any $i\in[3]$ as $L_i\subseteq L$ and $L$ is $\wf$ by assumption.

Thus $\tau_i^2\rdeq\inv{x_i}y_i\inv{x_i} y_i$ has to be $\wwf$.

If $\abs{y_i}\ge \abs{x_i}$, then $y_i\inv{x_i}$ has to be $\wf$, i.e.\ $x_i\sle y_i$. 
Setting $r_i:=x_i$ and $t_i:=\rd(y_i\inv{x_i})$, we have 
\[\tau_i\rdeq\inv{x_i}y_i = \inv{x_i}\rd(y_i\inv{x_i})x_i = \inv{r_i} t_i r_i\]

Otherwise $\abs{y_i}<\abs{x_i}$ and $x_i\inv{y_i}$ is $\wf$ with $y_i\slt x_i$. 
Then set $r_i:=y_i$ and $t_i:=\rd(x_i\inv{y_i})\neq\ew$ s.t.\ 
\[\tau_i\rdeq\inv{x_i} y_i = \inv{\rd(x_i\inv{y_i})y_i}y_i = \inv{t_ir_i}r_i=\inv{r_i}\inv{t_i}r_i\]

\end{block}

Assume that for some $i\in [3]$ we have $\rd(\tau_i) = \inv{r_i} \inv{t_i} r_i$ with $t_i \neq \ew$:
\begin{block}
As shown in Lemma~\ref{lem:negneg} and Lemma~\ref{lem:negpos} we always have for $j\neq i$ and any sequence $i_1\ldots i_l\in\{1,2\}^+$:
\begin{itemize}
\item
$us_{i_1}\ldots s_{i_l} s_j w\inv{r_j} t_j r_j \rdeq us_{i_1}\ldots s_{i_l} p^{m_j+n_j}w$
\item
$us_{i_1}\ldots s_{i_l} s_j w\inv{r_j} \inv{t_j} r_j \rdeq us_{i_1}\ldots s_{i_l} p^{m_j-n_j}w$
\end{itemize}
Hence $L\rdeq u(p^{m_1\pm n_1} + p^{m_2\pm n_2} + p^{m_3\pm n_3})^\ast w$ and thus

$\rlcs(L)=\rlcs(uw, us_1w\tau_1, us_2 w\tau_2, us_3w\tau_3)$
\end{block}

So it remains the case that for all $i\in[3]$ we have $\tau_i = \inv{r_i} t_i r_i$:
\begin{block}
Hence $L=(u,1)[(s_1,\inv{r_1} t_1 r_1) + (s_2,\inv{r_2} t_2 r_2) + (s_3, \inv{r_3} t_3 r_3)]^\ast w$.

As before $r_i \inv{r_j}$ has to be $\wwf$ for any $i,j\in[3]$.

Let $\{i_1,i_2,i_3\} = [3]$ s.t. $r_{i_3} \sle r_{i_2} \sle r_{i_1}$.

Note that $\rd(t_3 r_3 \inv{r_2} t_2 r_2 \inv{r_1} t_1 r_1) = \inv{x} y$ with $y\sle \rlcs(L)=R$.

Hence $\abs{R}\ge \abs{y}=\hd(t_3r_3\inv{r_2}t_2r_2\inv{r_1} t_1 r_1) + \abs{x} = \abs{t_1t_2t_3} + \abs{r_3} + \abs{x}$.

We show that $\abs{R} \ge \abs{y} \ge \abs{t_ir_i}$ for all $i\in[3]$:
\begin{itemize}
\item
If $\abs{r_1} \le \abs{t_2r_2} \wedge \abs{r_2} \le \abs{t_3r_3}$

then $\abs{y}=\abs{t_1t_2t_3r_3} \ge \abs{t_1t_2r_2} \ge \abs{t_1r_1}$
\item
If $\abs{r_1}> \abs{t_2r_2} \wedge \abs{r_2} + \abs{r_1} - \abs{t_2r_2} \le \abs{t_3r_3}$

then $\abs{t_2r_2} < \abs{r_1} \wedge \abs{r_1} \le \abs{t_2t_3r_3}$

and $\abs{y}=\abs{t_1t_2t_3r_3} \ge \abs{t_1r_1} > \abs{t_1t_2r_2}$
\item
If $\abs{r_1} \le \abs{t_2r_2} \wedge \abs{r_2} > \abs{t_3r_3}$, 

then $\abs{x} = \abs{r_2} - \abs{t_3r_3}$ 

and $\abs{y}= \abs{t_1t_2t_3r_3} + \abs{r_2} - \abs{t_3r_3 }  = \abs{t_1t_2r_2} \ge \max(\abs{t_1r_1}, \abs{t_1t_2t_3r_3})$
\item
If $\abs{r_1} >  \abs{t_2r_2} \wedge \abs{t_3r_3} < \abs{r_2}+\abs{r_1} - \abs{t_2r_2}$, 

then $\abs{x} = \abs{r_1} - \abs{t_2t_3r_3}$ 

and $\abs{y} = \abs{t_1t_2t_3r_3} + \abs{r_1} - \abs{t_2t_3r_3} = \abs{t_1r_1} \ge \max(\abs{t_1t_2r_2}, \abs{t_1t_2t_3r_3})$
\end{itemize}
Consider $L' =(u,1)[(s_1,\inv{r}_1 t_1 r_1)+(s_2,\inv{r_2} t_2 r_2)]^\ast (s_3,\inv{r}_3 t_3 r_3) w$

We have $L' \subseteq L$
and thus $R=\rlcs(L)\sle \rlcs(L')$.

Note that $(u,1)(s_3,\inv{r_3}t_3r_3)w$ cannot be a witness w.r.t.\ $uw$ as its length after reduction is strictly smaller than that of $r_X\kappa$, hence the two words have to coincide on at least the last $1+\abs{R}$ letters s.t. $\rlcs(L')\sle \rlcs(r_X\kappa, (u,1)(s_3,\inv{r_3}t_3r_3)w)=R$ i.e.\ $\rlcs(L)=\rlcs(L')=R$.

Let $\tilde{w}=\rd(s_3w\inv{r_3}t_3r_3)$.
\begin{block}
If $\tilde{w}\rdeq \inv{x} y$ is only $\wwf$, then $u=u'x$ and suitable conjugates of $s_i$ exist that allow us to move $x$ from $u=u'x$ through any sequence $s_{i_1}\ldots s_{i_l}$ next to $\tilde{w}$ as done before. Thus assume w.l.o.g.\ that $\tilde{w}$ is already $\wf$.
\end{block}
By Lemma~\ref{lem:pospos}, we have:

\[R=\rlcs(L') = \rlcsa{u\tilde{w}\\ us_1\tilde{w}\inv{r_1}t_1 r_1\\ us_2\tilde{w}\inv{r_2}t_2r_2\\ us_1s_1\tilde{w}\inv{r_1} t_1 t_1 r_1\\ us_2s_2\tilde{w}\inv{r_2} t_2 t_2r_2}\]

Neither $us_1\tilde{w}\inv{r_1}t_1 r_1$ nor $us_2\tilde{w}\inv{r_2}t_2r_2$ can be witnesses again because their length before reduction is strictly less than that of $r_X\kappa$.

Hence, either $us_1s_1\tilde{w}\inv{r}_1 t_1 t_1 r_1$ or $us_2s_2\tilde{w}\inv{r}_2 t_2 t_2 r_2$ is a witness w.r.t.\ $us_3w\inv{r_3}t_3 r_3$ and thus also w.r.t. $uw$.

W.l.o.g.\ $us_1s_1\tilde{w}\inv{r}_1 t_1 t_1 r_1$ is a witness.

Consider then $L''=(u,1)[(s_1,\inv{r}_1 t_1 r_1)^\ast+(s_3,\inv{r_3} t_3 r_3)^\ast]  w$.

Again $L''\subseteq L$ s.t.\ $R=\rlcs(L)\sle \rlcs(L'')$.

But also $\rlcs(L'')\sle \rlcs(uw, us_1s_1w\inv{r_1}t_1t_1r_1)=R$ s.t. $R=\rlcs(L)=\rlcs(L'')$

Using Lemma~\ref{lem:pospos} and now that $r_1,r_2,r_3\sle R=\rlcs(L'')$, we obtain

\[R=\rlcs(L'') = \rlcsa{u{w}\\ us_1{w}\inv{r_1}t_1 r_1\\ us_3{w}\inv{r_3}t_3r_3\\ us_1s_3{w}\inv{r_3} t_3 r_3 \inv{r_1} t_1 r_1\\ us_3s_1{w}\inv{r_1}t_1 r_1 \inv{r_3} t_3 r_3}\]

But by our assumption that $r_X\kappa$ is a witness w.r.t.\ $uw$ of minimal length before reduction, none of theses words can be witnesses. Hence, our assumption that such a factorization exists, cannot hold.
\end{block}
So, every path leading to the occurrence of $b$ that defines the $\rlcs$ of $L$ or to a letter right of it has to have height at most $3N$.
By minimality, we can also assume that any path fragment that leads from the main path (leading to $\rlcs$-defining occurrence of $b$) to a letter left of this $b$ contains any nonterminal at most once (see Fig.\ \ref{fig:lcp-height-2}). Hence, the derivation tree can have height at most $4N$.

\begin{figure}
\begin{center}
\scalebox{0.8}{
\begin{tikzpicture}
\draw (0,0) -- (-5,-5) -- (5,-5) -- (0,0);

\draw (-5,-5) -- (-5,-6) -- (5,-6) -- (5,-5);
\draw[dashed] (-5,-5.5) -- (5,-5.5);

\draw (0,-0.6) -- (4.4,-5);
\draw (0,-0.6) -- (-4.4,-5);
\node at (0+0.3,-0.6) {$A$};

\draw (0,-1.2) -- (-3.8,-5);
\draw (0,-1.2) -- (3.8,-5);
\node at (0+0.3,-1.2) {$A$};

\draw (0,-1.8) -- (-3.2,-5);
\draw (0,-1.8) -- (3.2,-5);
\node at (0+0.3,-1.8) {$A$};

\draw (0,-4.4) -- (-0.6,-5);
\draw (0,-4.4) -- (0.6,-5);
\node at (0+0.3,-4.4) {$A$};

\draw[dashed] (0,0) -- (0.0,-0.6) -- (-0,-1.2) -- (0,-1.8) -- (0,-4.4) -- (0,-5);
\draw[dashdotdotted] (0,0) -- (0.0,-0.6) -- (-0,-1.2) -- (0,-1.8) -- (-0,-4.1) -- (-0.9,-5);

\node at (-4.7,-5.25) {$\alpha$};
\draw[dotted] (-4.4,-5) -- (-4.4,-5.5);
\node at (-4.1,-5.25) {$\sigma_1$};
\draw[dotted] (-3.8,-5) -- (-3.8,-5.5);
\node at (-3.5,-5.25) {$\sigma_2$};
\draw[dotted] (-3.2,-5) -- (-3.2,-5.5);
\node at (-1.9,-5.25) {$\sigma_3$};
\draw[dotted] (-0.6,-5) -- (-0.6,-5.5);
\node at (0,-5.25) {$\gamma$};
\draw[dotted] (0.6,-5) -- (0.6,-5.5);
\node at (1.9,-5.25) {$\tau_3$};
\draw[dotted] (3.2,-5) -- (3.2,-5.5);
\node at (3.5,-5.25) {$\tau_2$};
\draw[dotted] (3.8,-5) -- (3.8,-5.5);
\node at (4.1,-5.25) {$\tau_1$};
\draw[dotted] (4.4,-5) -- (4.4,-5.5);
\node at (4.7,-5.25) {$\beta$};

\draw (-0.8,-6) -- (-0.8,-5.5);
\draw (-1.2,-6) -- (-1.2,-5.5);
\node at (-1,-5.75) {$b$};
\node at (1.6,-5.75) {$\xi$};
\node at (-3.55,-5.75) {$\zeta$};

\end{tikzpicture}
}
\end{center}
\caption{
Assume that  the $\rlcs$-defining letter $b$ is contained within $\kappa$ and not $r_X$, i.e.\
that $\rd(\kappa) = \ldots b\rlcs(L)$ with $r_X\kappa$ $\wf$. 
As any opening letter within $\kappa$ can only be canceled by a closing letter from right, 
by canceling out always the pair of matching opening and closing letters that is farthest to the right,
we obtain a unique factorization of $\kappa= \zeta b \xi$ s.t.\ 
$r_X\zeta$ and $\xi$ are both \wf\ with $\rd(r_X\kappa) = \rd(r_X\zeta) b \rd(\xi) = \rd(r_X\zeta) b 
\rlcs(L)$
s.t.\ this specific occurrence of the letter $b$  defines the reduced suffix $\rlcs(L)$.
(If $b$ is contained within $r_X$, then $\kappa=\xi$, $r_X=r_X'bx$ and $\rd(x\kappa)=\rlcs(L)$.)
We assume that the given derivation tree of the witness $\kappa$ contains a path (drawn as dashed line) which (i) leads to one of the letters within $b\xi$ and (ii) consists of at least $3N+1$ nonterminals so that by the pigeon-hole principle at least one nontermimal $A$ occurs at least $4$ times; specifically, consider precisely the first $3N+1$ nonterminals along such path and let $A$ be the nonterminal that occurs both at least $4$ times within this fragment and also occurs the earliest.
W.r.t.\ the nonterminal $A$ we factorize the witness as $\kappa=(\alpha,\beta)(\sigma_1,\tau_1)(\sigma_2,\tau_2)(\sigma_3,\tau_3)\gamma\rdeq \ldots b
\rlcs(L)$.} 
\label{fig:lcp-height}
\end{figure}
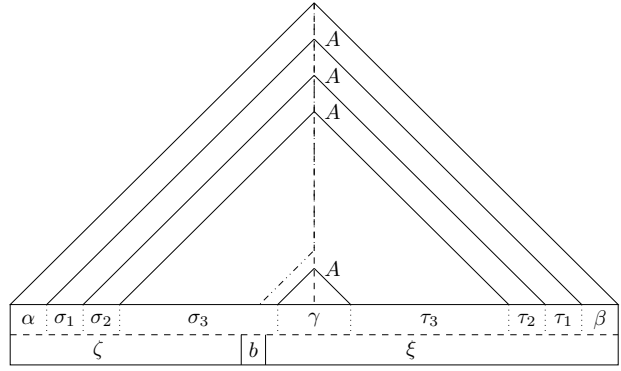

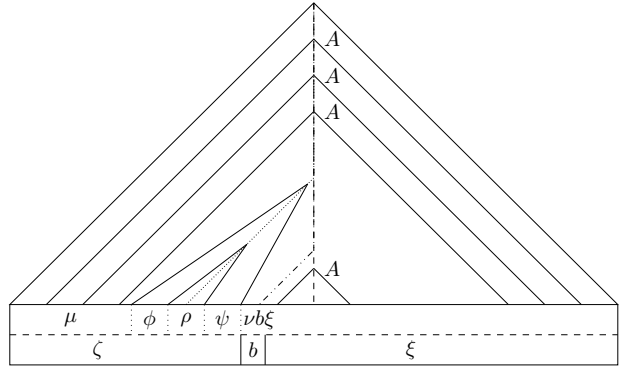
\begin{figure}
\begin{center}
\scalebox{0.8}{
\begin{tikzpicture}
\draw (0,0) -- (-5,-5) -- (5,-5) -- (0,0);

\draw (-5,-5) -- (-5,-6) -- (5,-6) -- (5,-5);
\draw[dashed] (-5,-5.5) -- (5,-5.5);

\draw (0,-0.6) -- (4.4,-5);
\draw (0,-0.6) -- (-4.4,-5);
\node at (0+0.3,-0.6) {$A$};

\draw (0,-1.2) -- (-3.8,-5);
\draw (0,-1.2) -- (3.8,-5);
\node at (0+0.3,-1.2) {$A$};

\draw (0,-1.8) -- (-3.2,-5);
\draw (0,-1.8) -- (3.2,-5);
\node at (0+0.3,-1.8) {$A$};

\draw (0,-4.4) -- (-0.6,-5);
\draw (0,-4.4) -- (0.6,-5);
\node at (0+0.3,-4.4) {$A$};

\draw[dashed] (0,0) -- (0.0,-0.6) -- (-0,-1.2) -- (0,-1.8) -- (0,-4.4) -- (0,-5);
\draw[dashdotdotted] (0,0) -- (0.0,-0.6) -- (-0,-1.2) -- (0,-1.8) -- (-0,-4.1) -- (-0.9,-5);

\draw[densely dotted] (0,0) -- (0.0,-0.6) -- (-0,-1.2) -- (0,-1.8) -- (-0,-2.9) -- (-2.1,-5);

\draw (-0.1, -3) -- (-1.2,-5);
\draw (-0.1, -3) -- (-3,-5);

\draw (-1.1, -4) -- (-1.8,-5);
\draw (-1.1, -4) -- (-2.4,-5);

\node at (-4,-5.25) {$\mu$};
\draw[dotted] (-3,-5) -- (-3,-5.5);
\node at (-2.7,-5.25) {$\phi$};
\draw[dotted] (-2.4,-5) -- (-2.4,-5.5);
\node at (-2.1,-5.25) {$\rho$};
\draw[dotted] (-1.8,-5) -- (-1.8,-5.5);
\node at (-1.5,-5.25) {$\psi$};
\draw[dotted] (-1.2,-5) -- (-1.2,-5.5);
\node at (-0.9,-5.25) {$\nu b \xi$};

\draw (-0.8,-6) -- (-0.8,-5.5);
\draw (-1.2,-6) -- (-1.2,-5.5);
\node at (-1,-5.75) {$b$};
\node at (1.6,-5.75) {$\xi$};
\node at (-3.55,-5.75) {$\zeta$};

\end{tikzpicture}
}
\end{center}
\caption{
Assume that the $\rlcs$-defining occurrence of $b$ is not contained in $r_X$ s.t.\
$\kappa = \zeta b \xi$ with both $r_X\zeta$ and $\xi$ \wf\ and $\rd(r_X\kappa) = \rd(r_X\zeta) b 
\rlcs(L)$.
Consider any path that leads to a letter within $\zeta$. (If $b$ is contained in $r_X$, then this cannot happen.)
The first nonterminal along this path that is not also contained in the path leading to $b$ defines a subtree
that does not contain the marked $b$ anymore.
Assume this subtree contains a path with at least $N+1$ nonterminals s.t.\ we can factorize $\zeta = (\mu, \nu)(\phi,\psi)\rho$.
Then $\kappa = (\mu,\nu b \xi)(\phi, \psi)\rho$, and $L'=(\mu,\nu b \xi)(\phi, \psi)^*\rho$ is a sublanguage of $L_X$ and thus $r_X L'$ is $\wf$.
Hence, $(r_X\mu,\nu b \xi)(\phi, \psi)^0\rho = r_X\mu \rho \nu b\xi$ is \wf. As $b\xi$ is \wf, too, we have that $r_X\mu \rho \nu b\xi \rdeq \rd(r_X\mu\rho\nu) b
\rlcs(L)$ is a shorter (before reduction) witness than $\kappa$.
Hence, we can always assume that all subtrees rooted at a node left of the path leading to the marked $b$ have height at most $n-1$.
Thus, if all paths leading to a letter within $b\xi$ contain at most $3N$ nonterminals, then the derivation tree can have at most height $4N$.
} 
\label{fig:lcp-height-2}
\end{figure}
\end{prfblk}

%% file: proofs-not-wf.tex
\subsection{Lemma~\ref{thm:not-wf--main-text} in the main work}

Recall, $\sT_X^{\le h}:=\tsr(r_X L_X^{\le })$ denotes an $\tseq$-equivalent sublanguage of $\rd(r_XL_X^{\le h})$, assuming that $r_X L_X^{\le h}$ is still well-formed.
\begin{lemma}[Lemma~\ref{thm:not-wf--main-text} in the main work]
If $L=L(G)$ is not $\wf$, then there is some least $h_0$ s.t.\ $r_X L_Y^{\le h_0}\inv{r_Z}$ is not $\wf$ with $X\to_G YZ$.
For $h\le h_0$, all $r_X L_X^{\le h}$ are $\wf$ s.t.\ $\sT_X^{\le h}\tseq \rd(r_X L_X^{\le h})$.
If $h_0 \ge 4N+1$, then at least for one nonterminal $X$ we have $\sT_X^{\le 4N+1}\not\tseq \sT_X^{\le 4N}$.
\end{lemma}
\begin{prfblk}
We write $\tsr(r_XL_X^{\le h})$ for $\tsn(\rd(r_XL_X^{\le h}))$.

For simplicity, we also assume that all linear rules have been removed.

Assume thus that $G$ is not $\wf$.

We assume that all nullary rules $X\to \inv{u}v$ are already reduced and w.l.o.g. $\wwf$.

Further w.l.o.g. $G$ is nonnegative.

Then there is some $\alpha\in L(G)$ that is not $\wf$.

We show that then there is some rule $X\to_G YZ$ and words $\alpha_X=\alpha_Y\alpha_Z$ with $\alpha_Y\in L_Y$ and $\alpha_Z\in L_Z$ s.t.:
\begin{itemize}
\item
$r_X \alpha_Y\inv{r_Z}$ is not $\wf$
\item
$r_X\inv{r_Y}$ is $\wf$.
\item
$r_Y \alpha_Y$ is $\wf$.
\end{itemize}

To this end, consider any derivation of $\alpha$:
\begin{block}
Set $X:=S$ and $\alpha_X:=\alpha$

We have $r_{X}=r_S=\ew$ with $r_{X}\alpha_X$ not $\wf$

While $r_{X} \alpha_X$ is not $\wf$:
\begin{block}
Then there is some rule $X\to_G YZ$ and factorization $\alpha_X=\alpha_Y\alpha_Z$ as by assumption $r_X r$ is $\wf$ for all constant rules $X\to_G r$.

If $r_Y\alpha_Y$ is not $\wf$:
\begin{block}
Redefine $X:=Y$ and $\alpha_X:=\alpha_Y$ and descend accordingly into the derivation tree of $\alpha_Y$.
\end{block}

If $r_Z\alpha_Z$ is not $\wf$:
\begin{block}
Redefine $X:=Z$ and $\alpha_X:=\alpha_Z$ and descend accordingly into the derivation tree of $\alpha_Z$.
\end{block}

Otherwise $r_Z \alpha_Z$ is $\wf$, thus $\alpha_Z\rdeq \inv{u_Z}v_Z$ with $r_Z=r_Z'u_Z$.

Thus $r_X \alpha_Y \inv{r_Z}$ is not $\wf$ as $r_{X}\alpha_X=r_X \alpha_Y \alpha_Z \rdeq r_X \alpha_Y \inv{r_Z} r_Z' v_Y$
\end{block}
\end{block}

So, there is some least derivation height $n_0$ s.t.
\begin{itemize}
\item
$r_X L_X^{\le n_0}$ is $\wf$ for every nonterminal $X$
\item
$r_X\inv{r_Y}$ is $\wf$ for all rules $X\to_G YZ$
\item
there exists a nonterminal $X_0$ with $X_0\to_G YZ$, $\alpha_Y\in L_Y^{n_0}$, and $r_{X_0} \alpha_Y \inv{r_Z}$ not $\wf$ anymore.
\end{itemize}

As all $r_Y L_Y^{\le n_0}$ are $\wf$, we have $\rd(r_Y L_Y^{\le n_0})\tseq \tsr(r_Y L_Y^{\le n_0})\tseq \sT_Y^{\le n_0}$.

Thus also
\[\rd(r_XL_Y^{\le n_0}) = \rd(r_X\inv{r_Y}) \rd(r_YL_Y^{\le n_0}) \tseq \rd(r_X\inv{r_Y}) \sT_Y^{\le n_0}\]

Finally, as $G$ is nonnegative, also $r_X L_Y^{\le n_0} \inv{r_Z}$ is nonnegative.

Thus, as $r_X \alpha_Y \inv{r_Z}$ is not $\wf$, we have that $\rlcs(r_X L_Y^{\le n_0})\inv{r_Z}$ is not $\wf$, and thus 
$\rd(r_X\inv{r_Y}) \sT_Y^{\le n_0}\inv{r_Z}$ is not $\wf$.

So, if $n_0\le 4N+1$, by iteratively computing $\sT_X^{\le h}\tseq\tsn(\rd(r_XL_X^{\le h}))$, we discover the error;

Otherwise $r_Z \sle \rlcs(r_X L_Y^{\le 4N+1})=\lcs(\rd(r_X\inv{r_Y})\sT_X^{\le 4N+1})$ but $r_Z \not\sle \rlcs(r_X L_Y^{\le n_0})$;

Thus $\rd(r_X L_Y^{\le 4N+1})\not\tseq r_X L_Y^{\le n_0}$ and thus $r_Y L_Y^{\le 4N+1}\not\tseq r_YL_Y^{\le n_0}$.

So, for at least one nonterminal $\sig$ cannot have converged.
\end{prfblk}